\pdfoutput=1
\documentclass[10pt]{article}
\usepackage[margin=1in]{geometry}

\usepackage{amsmath,amssymb,amsthm}
\usepackage{graphicx}
\usepackage{booktabs}
\usepackage{multirow}
\usepackage{url}
\usepackage{algorithm}
\usepackage[noend]{algpseudocode}
\usepackage{tikz}
\usetikzlibrary{arrows.meta,positioning,calc,shapes.geometric}
\usepackage[hidelinks]{hyperref}
\hypersetup{
  pdftitle={SOLO: Certified-Recall Metric Similarity Search with
            Scan-Only Sampled Inverted Lists},
  pdfauthor={Edgar Chavez}}

\definecolor{sloblue}{HTML}{2a78d6}
\definecolor{sloorange}{HTML}{eb6834}

\usepackage{expl3}
\ExplSyntaxOn
\prop_new:N \g_num_prop
\NewDocumentCommand{\defnum}{mm}{\prop_gput:Nnn \g_num_prop {#1} {#2}}
\NewDocumentCommand{\num}{m}{\prop_item:Nn \g_num_prop {#1}}
\ExplSyntaxOff
\defnum{abl/hyb-qps-ks128}{35}
\defnum{abl/hyb-qps-ks64}{72}
\defnum{abl/hyb-recall-ks128}{0.9976}
\defnum{abl/hyb-recall-ks64}{0.9890}
\defnum{abl/rec-qps-ks128}{52}
\defnum{abl/rec-qps-ks64}{124}
\defnum{abl/rec-recall-ks128}{0.9978}
\defnum{abl/rec-recall-ks64}{0.9899}
\defnum{capped/1g-io-ks256}{37}
\defnum{capped/1g-io-ks512}{71}
\defnum{capped/1g-qps-ks256}{19.8}
\defnum{capped/1g-qps-ks512}{10.2}
\defnum{capped/1g-recall-d2}{0.9977}
\defnum{capped/1g-recall-ks256}{0.9965}
\defnum{capped/1g-recall-ks512}{0.9979}
\defnum{capped/bytes-per-object}{10.7}
\defnum{capped/diskann-died-gb}{8}
\defnum{capped/diskann-floor-gb}{12}
\defnum{capped/fullset-recall-ks256}{0.9982}
\defnum{capped/fullset-recall-ks512}{0.9990}
\defnum{capped/ivfpq-1g-qps}{4.5}
\defnum{capped/ivfpq-bytes}{56.1}
\defnum{capped/ivfpq-ceiling}{0.8747}
\defnum{capped/peak-gb}{1}
\defnum{capped/rabitq-rq1-codes/12g-qps}{1{,}300}
\defnum{capped/rabitq-rq1-codes/12g-recall}{0.3450}
\defnum{capped/rabitq-rq1-codes/1g-qps}{167}
\defnum{capped/rabitq-rq1-codes/1g-recall}{0.3450}
\defnum{capped/rabitq-rq1-codes/2g-qps}{624}
\defnum{capped/rabitq-rq1-codes/2g-recall}{0.3450}
\defnum{capped/rabitq-rq1-codes/4g-qps}{1{,}371}
\defnum{capped/rabitq-rq1-codes/4g-recall}{0.3450}
\defnum{capped/rabitq-rq1-codes/6g-qps}{1{,}350}
\defnum{capped/rabitq-rq1-codes/6g-recall}{0.3450}
\defnum{capped/rabitq-rq1-codes/8g-qps}{1{,}382}
\defnum{capped/rabitq-rq1-codes/8g-recall}{0.3450}
\defnum{capped/rabitq-rq1-codes/bytes}{28.1}
\defnum{capped/rabitq-rq1-codes/ceiling}{0.3450}
\defnum{capped/rabitq-rq1-codes/died-gb}{--}
\defnum{capped/rabitq-rq1-codes/floor-gb}{1}
\defnum{capped/rabitq-rq1-codes/uncapped-qps}{1{,}357}
\defnum{capped/rabitq-rq1-codes/uncapped-recall}{0.3450}
\defnum{capped/rabitq-rq1-refine/48g-qps}{98.0}
\defnum{capped/rabitq-rq1-refine/48g-recall}{0.9942}
\defnum{capped/rabitq-rq1-refine/bytes}{412.1}
\defnum{capped/rabitq-rq1-refine/ceiling}{0.9942}
\defnum{capped/rabitq-rq1-refine/died-gb}{12}
\defnum{capped/rabitq-rq1-refine/floor-gb}{48}
\defnum{capped/rabitq-rq1-refine/uncapped-qps}{97}
\defnum{capped/rabitq-rq1-refine/uncapped-recall}{0.9942}
\defnum{capped/rabitq-rq4-codes/12g-qps}{96.5}
\defnum{capped/rabitq-rq4-codes/12g-recall}{0.8715}
\defnum{capped/rabitq-rq4-codes/1g-qps}{4.1}
\defnum{capped/rabitq-rq4-codes/1g-recall}{0.8715}
\defnum{capped/rabitq-rq4-codes/2g-qps}{5.1}
\defnum{capped/rabitq-rq4-codes/2g-recall}{0.8715}
\defnum{capped/rabitq-rq4-codes/4g-qps}{9.7}
\defnum{capped/rabitq-rq4-codes/4g-recall}{0.8715}
\defnum{capped/rabitq-rq4-codes/6g-qps}{32.3}
\defnum{capped/rabitq-rq4-codes/6g-recall}{0.8715}
\defnum{capped/rabitq-rq4-codes/8g-qps}{96.5}
\defnum{capped/rabitq-rq4-codes/8g-recall}{0.8715}
\defnum{capped/rabitq-rq4-codes/bytes}{76.1}
\defnum{capped/rabitq-rq4-codes/ceiling}{0.8715}
\defnum{capped/rabitq-rq4-codes/died-gb}{--}
\defnum{capped/rabitq-rq4-codes/floor-gb}{1}
\defnum{capped/rabitq-rq4-codes/uncapped-qps}{96}
\defnum{capped/rabitq-rq4-codes/uncapped-recall}{0.8715}
\defnum{capped/rabitq-rq4-refine/48g-qps}{95.1}
\defnum{capped/rabitq-rq4-refine/48g-recall}{0.9987}
\defnum{capped/rabitq-rq4-refine/bytes}{460.1}
\defnum{capped/rabitq-rq4-refine/ceiling}{0.9987}
\defnum{capped/rabitq-rq4-refine/died-gb}{24}
\defnum{capped/rabitq-rq4-refine/floor-gb}{48}
\defnum{capped/rabitq-rq4-refine/uncapped-qps}{95}
\defnum{capped/rabitq-rq4-refine/uncapped-recall}{0.9987}
\defnum{capped/scann-ah/12g-qps}{4{,}608}
\defnum{capped/scann-ah/12g-recall}{0.4281}
\defnum{capped/scann-ah/8g-qps}{4{,}525}
\defnum{capped/scann-ah/8g-recall}{0.4281}
\defnum{capped/scann-ah/bytes}{52.2}
\defnum{capped/scann-ah/ceiling}{0.4281}
\defnum{capped/scann-ah/died-gb}{6}
\defnum{capped/scann-ah/floor-gb}{8}
\defnum{capped/scann-ah/uncapped-qps}{4{,}423}
\defnum{capped/scann-ah/uncapped-recall}{0.4281}
\defnum{capped/scann-bf16/32g-qps}{985}
\defnum{capped/scann-bf16/32g-recall}{0.9915}
\defnum{capped/scann-bf16/bytes}{244.2}
\defnum{capped/scann-bf16/ceiling}{0.9937}
\defnum{capped/scann-bf16/died-gb}{24}
\defnum{capped/scann-bf16/floor-gb}{32}
\defnum{capped/scann-bf16/uncapped-qps}{607}
\defnum{capped/scann-bf16/uncapped-recall}{0.9937}
\defnum{capped/solo-1g-vs-2g-pct}{12}
\defnum{capped/solo-bytes}{2.7}
\defnum{capped/solo-floor-mb}{256}
\defnum{capped/solo-floor-recall}{0.9964}
\defnum{capped/solo-floor-threads}{4}
\defnum{capped/solo-floor512-mb}{256}
\defnum{capped/solo-floor512-qps}{2.9}
\defnum{capped/solo-floor512-recall}{0.9977}
\defnum{capped/solo-floor512-threads}{2}
\defnum{capped/solo-ks512-256m-qps}{2.9}
\defnum{capped/solo-ks512-256m-threads}{2}
\defnum{capped/solo-ks512-384m-qps}{5.0}
\defnum{capped/solo-ks512-384m-threads}{4}
\defnum{capped/solo-ks512-512m-qps}{5.3}
\defnum{capped/solo-ks512-512m-threads}{4}
\defnum{capped/solo-ks512-768m-qps}{8.3}
\defnum{capped/solo-ks512-768m-threads}{8}
\defnum{capped/spann-8g-qps}{185}
\defnum{capped/spann-8g-recall}{0.9911}
\defnum{capped/spann-died-gb}{6}
\defnum{capped/spann-floor-gb}{8}
\defnum{capped/spann-uncapped-qps}{186}
\defnum{ceiling/1e10-d2-router-gib}{1.8}
\defnum{ceiling/1e10-d3-router-mib}{38}
\defnum{ceiling/1e9-d2-analytic-mib}{171}
\defnum{ceiling/1e9-flat-analytic-gib}{8.34}
\defnum{ceiling/bytes-per-object}{8.96}
\defnum{ceiling/nstar-1gb}{1.2}
\defnum{cert/msmarco100m-best-recall}{0.8987}
\defnum{cert/msmarco100m-cov-ceiling}{0.9069}
\defnum{cert/msmarco100m-max-err}{0.0001}
\defnum{cert/msmarco100m-points}{6}
\defnum{cert/msmarco100m-router-cost}{0.0194}
\defnum{certboot/deep10m/cell-b}{4}
\defnum{certboot/deep10m/cell-cov}{0.9944}
\defnum{certboot/deep10m/cell-ks}{256}
\defnum{certboot/deep10m/cell-ratio}{1.25}
\defnum{certboot/deep10m/m}{10{,}000}
\defnum{certboot/deep10m/min-cov}{0.9110}
\defnum{certboot/deep10m/nboot}{1{,}000}
\defnum{certboot/deep10m/ratio-max}{1.51}
\defnum{certboot/glove/cell-b}{32}
\defnum{certboot/glove/cell-cov}{0.9905}
\defnum{certboot/glove/cell-ks}{128}
\defnum{certboot/glove/cell-ratio}{1.23}
\defnum{certboot/glove/m}{10{,}000}
\defnum{certboot/glove/min-cov}{0.7319}
\defnum{certboot/glove/nboot}{1{,}000}
\defnum{certboot/glove/ratio-max}{2.06}
\defnum{certboot/halfwidth-max-99}{0.0007}
\defnum{certboot/ratio-max}{2.06}
\defnum{certboot/ratio-max-99}{1.26}
\defnum{certboot/sift/cell-b}{16}
\defnum{certboot/sift/cell-cov}{0.9972}
\defnum{certboot/sift/cell-ks}{32}
\defnum{certboot/sift/cell-ratio}{1.15}
\defnum{certboot/sift/m}{10{,}000}
\defnum{certboot/sift/min-cov}{0.9567}
\defnum{certboot/sift/nboot}{1{,}000}
\defnum{certboot/sift/ratio-max}{1.35}
\defnum{cov/glove-b16-ks128}{0.9683}
\defnum{cov/glove-b16-ks256}{0.9890}
\defnum{cov/glove-b32-ks128}{0.9902}
\defnum{cov/glove-b32-ks256}{0.9977}
\defnum{cov/glove-b4-ks128}{0.8626}
\defnum{cov/glove-b4-ks256}{0.9126}
\defnum{cov/glove-b64-ks128}{0.9980}
\defnum{cov/glove-b64-ks256}{0.9997}
\defnum{cov/glove-b8-ks128}{0.9272}
\defnum{cov/glove-b8-ks256}{0.9635}
\defnum{cov/sift-b16-ks128}{1.0000}
\defnum{cov/sift-b16-ks64}{0.9997}
\defnum{cov/sift-b32-ks128}{1.0000}
\defnum{cov/sift-b32-ks64}{1.0000}
\defnum{cov/sift-b4-ks128}{0.9965}
\defnum{cov/sift-b4-ks64}{0.9849}
\defnum{cov/sift-b64-ks128}{1.0000}
\defnum{cov/sift-b64-ks64}{1.0000}
\defnum{cov/sift-b8-ks128}{0.9996}
\defnum{cov/sift-b8-ks64}{0.9977}
\defnum{d2c/d1-1g-ks256-qps-cold}{8.0}
\defnum{d2c/d1-1g-ks256-qps-warm}{1.7}
\defnum{d2c/d1-1g-ks256-recall}{0.9965}
\defnum{d2c/d1-1g-ks256-route-ms}{547.3}
\defnum{d2c/d1-2g-ks256-qps-cold}{20.9}
\defnum{d2c/d1-2g-ks256-qps-warm}{20.2}
\defnum{d2c/d1-2g-ks256-recall}{0.9965}
\defnum{d2c/d1-2g-ks256-route-ms}{1.0}
\defnum{d2c/d1-2g-ks512-qps-cold}{10.7}
\defnum{d2c/d1-2g-ks512-qps-warm}{10.5}
\defnum{d2c/d1-2g-ks512-recall}{0.9979}
\defnum{d2c/d1-2g-ks512-route-ms}{1.3}
\defnum{d2c/d2-1g-ks256-qps-cold}{17.6}
\defnum{d2c/d2-1g-ks256-qps-warm}{16.3}
\defnum{d2c/d2-1g-ks256-recall}{0.9964}
\defnum{d2c/d2-1g-ks256-route-ms}{12.8}
\defnum{d2c/d2-1g-ks512-qps-cold}{9.1}
\defnum{d2c/d2-1g-ks512-qps-warm}{8.7}
\defnum{d2c/d2-1g-ks512-recall}{0.9977}
\defnum{d2c/d2-1g-ks512-route-ms}{18.4}
\defnum{d2c/d2-256m-b-per-object}{2.7}
\defnum{d2c/d2-256m-ks256-qps-warm}{7.3}
\defnum{d2c/d2-2g-ks256-qps-cold}{19.9}
\defnum{d2c/d2-2g-ks256-qps-warm}{19.9}
\defnum{d2c/d2-2g-ks256-recall}{0.9964}
\defnum{d2c/d2-2g-ks256-route-ms}{3.6}
\defnum{d2c/d2-2g-ks512-qps-cold}{10.2}
\defnum{d2c/d2-2g-ks512-qps-warm}{9.8}
\defnum{d2c/d2-2g-ks512-recall}{0.9977}
\defnum{d2c/d2-2g-ks512-route-ms}{7.5}
\defnum{d2c/d2-512m-b-per-object}{5.4}
\defnum{d2c/d2-512m-ks256-qps-warm}{12.8}
\defnum{d2c/router-d2-mib}{18.3}
\defnum{d2c/router-flat-mib}{963}
\defnum{d2c/router-io-mb}{4.8}
\defnum{d2c/router-ratio}{53}
\defnum{deep100m/d2/floor256-b-per-object}{2.7}
\defnum{deep100m/d2/floor256-cap-gb}{0.25}
\defnum{deep100m/d2/floor256-qps-warm}{7.3}
\defnum{deep100m/d2/floor256-recall}{0.9964}
\defnum{deep100m/d2/floor256-threads}{4}
\defnum{deep100m/d2/floor512-b-per-object}{2.7}
\defnum{deep100m/d2/floor512-cap-gb}{0.25}
\defnum{deep100m/d2/floor512-qps-warm}{2.9}
\defnum{deep100m/d2/floor512-recall}{0.9977}
\defnum{deep100m/d2/floor512-threads}{2}
\defnum{deep100m/matched-qps}{936}
\defnum{deep100m/matched-recall}{0.9989}
\defnum{deep100m/matched-speedup}{1.8}
\defnum{deep100m/repl-hnsw-gain}{1.6}
\defnum{deep100m/repl-hnsw-qps}{876}
\defnum{deep100m/repl-matched-qps}{1{,}918}
\defnum{deep100m/repl-matched-speedup}{2.2}
\defnum{deep100m/repl-max-qps}{1{,}631}
\defnum{deep100m/repl-solo-gain}{2.2}
\defnum{deep100m/repl-speedup-hi}{2.2}
\defnum{deep100m/repl-speedup-lo}{1.2}
\defnum{deep100m/solo-max-qps}{268}
\defnum{deep100m/solo-max-recall}{0.9991}
\defnum{deep100m/speedup-floor}{0.9974}
\defnum{deep100m/speedup-hi}{1.8}
\defnum{deep100m/speedup-lo}{1.1}
\defnum{deep100m/stream-max-qps}{755}
\defnum{deep100m/stream-max-recall}{0.9990}
\defnum{deep1b/1e8-512m-b-per-object}{5.4}
\defnum{deep1b/1e8-512m-recall}{0.9964}
\defnum{deep1b/1e8-flat-floor-b-per-object-hi}{21}
\defnum{deep1b/1e8-flat-floor-b-per-object-lo}{16}
\defnum{deep1b/1g-16thr-qps}{11.92}
\defnum{deep1b/2g-qps}{20.52}
\defnum{deep1b/2g-recall}{0.9832}
\defnum{deep1b/best-cov}{0.9986}
\defnum{deep1b/best-mb}{37.5}
\defnum{deep1b/bitmap-gib-64}{7.45}
\defnum{deep1b/bitmap-mib-thread}{119}
\defnum{deep1b/cov-b16-ks128}{0.9889}
\defnum{deep1b/cov-b16-ks256}{0.9960}
\defnum{deep1b/cov-b16-ks512}{0.9986}
\defnum{deep1b/cov-b16-ks64}{0.9725}
\defnum{deep1b/cov-b8-ks128}{0.9711}
\defnum{deep1b/cov-b8-ks256}{0.9885}
\defnum{deep1b/cov-b8-ks512}{0.9955}
\defnum{deep1b/cov-b8-ks64}{0.9389}
\defnum{deep1b/d1-collapse-8g}{64}
\defnum{deep1b/d2-gain-16g}{1.19}
\defnum{deep1b/d2-gain-8g}{63}
\defnum{deep1b/floor-b-per-object}{0.54}
\defnum{deep1b/floor-cap-mb}{512}
\defnum{deep1b/floor-nq}{200}
\defnum{deep1b/floor-qps}{6.93}
\defnum{deep1b/floor-recall}{0.9925}
\defnum{deep1b/floor-threads}{8}
\defnum{deep1b/floor-vs-1e8}{10}
\defnum{deep1b/list-skew}{47}
\defnum{deep1b/ns}{20{,}000{,}000}
\defnum{deep1b/qps-d1-12g}{1.84}
\defnum{deep1b/qps-d1-16g}{25.44}
\defnum{deep1b/qps-d1-8g}{0.40}
\defnum{deep1b/qps-d2-16g}{30.15}
\defnum{deep1b/qps-d2-4g}{19.59}
\defnum{deep1b/qps-d2-8g}{25.18}
\defnum{deep1b/recall-d1}{0.9912}
\defnum{deep1b/recall-d2}{0.9904}
\defnum{deep1b/route-ms-d1-12g}{516.6}
\defnum{deep1b/route-ms-d1-16g}{12.4}
\defnum{deep1b/route-ms-d1-8g}{2442.2}
\defnum{deep1b/route-ms-d2-16g}{9.1}
\defnum{deep1b/route-ms-d2-4g}{11.4}
\defnum{deep1b/route-ms-d2-8g}{9.2}
\defnum{deep1b/router-d2-mib}{189}
\defnum{deep1b/router-flat-mib}{9{,}882}
\defnum{deep1b/router-gap}{0.0002}
\defnum{deep1b/router-h}{1.29}
\defnum{deep1b/router-io-mb}{9.2}
\defnum{deep1b/router-ratio}{52}
\defnum{deep1b/s2/cand}{387{,}665}
\defnum{deep1b/s2/cascade-gain}{1.03}
\defnum{deep1b/s2/cov-same}{0.9944}
\defnum{deep1b/s2/identity-err}{0.0000}
\defnum{deep1b/s2/io-mb}{20.0}
\defnum{deep1b/s2/nq}{500}
\defnum{deep1b/s2/pread-gain}{2.9}
\defnum{deep1b/s2/qps-cascade}{8.08}
\defnum{deep1b/s2/qps-ks256}{26.4}
\defnum{deep1b/s2/qps-mmap}{7.81}
\defnum{deep1b/s2/qps-pread}{22.3}
\defnum{deep1b/s2/recall}{0.9912}
\defnum{deep1b/s2/recall-ks256}{0.9838}
\defnum{deep1b/s2/route-ms}{12.8}
\defnum{deep1b/s2/router-gb}{9.65}
\defnum{deep1b/s2/scan-ms-mmap}{116.0}
\defnum{deep1b/s2/scan-ms-pread}{32.0}
\defnum{deep1b/s2/screen-loss}{0.0032}
\defnum{deep1b/s2/stream-gb}{358}
\defnum{deep1b/sig-gib}{60}
\defnum{deep1b/starts-mib-after}{153}
\defnum{deep1b/starts-mib-before}{1{,}373}
\defnum{deep1b/starts-ratio}{9}
\defnum{floor/1500m-b-per-object}{16.1}
\defnum{floor/1500m-qps-ks256-warm}{19.9}
\defnum{floor/1500m-qps-ks512-warm}{1.1}
\defnum{floor/1g-evict-qps-ks256-cold}{7.7}
\defnum{floor/1g-evict-qps-ks256-warm}{1.7}
\defnum{floor/1g-prewarm-qps-ks256-cold}{20.5}
\defnum{floor/1g-prewarm-qps-ks256-warm}{10.9}
\defnum{floor/2g-b-per-object}{21.5}
\defnum{floor/2g-qps-ks512-warm}{10.3}
\defnum{glove/flat-max-qps}{271}
\defnum{glove/flat-max-recall}{0.9955}
\defnum{glove/h2-qps-ks128-ksl16}{68}
\defnum{glove/h2-qps-ks128-ksl32}{52}
\defnum{glove/h2-qps-ks64-ksl16}{150}
\defnum{glove/h2-qps-ks64-ksl32}{124}
\defnum{glove/h2-recall-ks128-ksl16}{0.9973}
\defnum{glove/h2-recall-ks128-ksl32}{0.9978}
\defnum{glove/h2-recall-ks64-ksl16}{0.9883}
\defnum{glove/h2-recall-ks64-ksl32}{0.9899}
\defnum{glove/interleaved-seeds}{3}
\defnum{glove/kb}{32}
\defnum{glove/matched-qps}{808}
\defnum{glove/matched-recall}{0.9881}
\defnum{glove/matched-speedup}{3.0}
\defnum{glove/qps-at-099}{641}
\defnum{glove/recall-at-099}{0.9910}
\defnum{glove/seeds}{3}
\defnum{glove/solo-max}{0.9978}
\defnum{glove/speedup-floor}{0.9533}
\defnum{glove/speedup-hi}{3.0}
\defnum{glove/speedup-lo}{1.4}
\defnum{insurance/j-median}{13}
\defnum{insurance/j1-frac}{2.5}
\defnum{law/glove-contour1024-hi}{0.9332}
\defnum{law/glove-contour1024-lo}{0.9126}
\defnum{lidlaw/beyond-deep10m}{0}
\defnum{lidlaw/beyond-glove}{0}
\defnum{lidlaw/beyond-laion-in}{0}
\defnum{lidlaw/beyond-laion-ood}{1}
\defnum{lidlaw/beyond-pubmed23-in}{0}
\defnum{lidlaw/beyond-pubmed23-ood}{4}
\defnum{lidlaw/beyond-sift}{0}
\defnum{lidlaw/beyond-t2i-in}{0}
\defnum{lidlaw/beyond-t2i-ip-in}{0}
\defnum{lidlaw/beyond-t2i-ip-ood}{66}
\defnum{lidlaw/beyond-t2i-ood}{86}
\defnum{lidlaw/beyond-wiki-in}{0}
\defnum{lidlaw/beyond-wiki-ood}{4}
\defnum{lidlaw/cov-8-64-deep10m}{0.9851}
\defnum{lidlaw/cov-8-64-glove}{0.8792}
\defnum{lidlaw/cov-8-64-laion-in}{0.9913}
\defnum{lidlaw/cov-8-64-laion-ood}{0.9920}
\defnum{lidlaw/cov-8-64-pubmed23-in}{0.9937}
\defnum{lidlaw/cov-8-64-pubmed23-ood}{0.9853}
\defnum{lidlaw/cov-8-64-sift}{0.9977}
\defnum{lidlaw/cov-8-64-t2i-in}{0.9953}
\defnum{lidlaw/cov-8-64-t2i-ip-in}{0.9964}
\defnum{lidlaw/cov-8-64-t2i-ip-ood}{0.9543}
\defnum{lidlaw/cov-8-64-t2i-ood}{0.9281}
\defnum{lidlaw/cov-8-64-wiki-in}{0.9959}
\defnum{lidlaw/cov-8-64-wiki-ood}{0.9926}
\defnum{lidlaw/err-deep10m}{0.0000}
\defnum{lidlaw/err-glove}{0.0021}
\defnum{lidlaw/err-laion-ood}{0.0003}
\defnum{lidlaw/err-max}{0.0039}
\defnum{lidlaw/err-naive-deep10m}{0.0001}
\defnum{lidlaw/err-naive-glove}{0.0055}
\defnum{lidlaw/err-naive-laion-ood}{0.0006}
\defnum{lidlaw/err-naive-pubmed23-ood}{0.0084}
\defnum{lidlaw/err-naive-sift}{0.0004}
\defnum{lidlaw/err-naive-t2i-ip-ood}{0.0421}
\defnum{lidlaw/err-naive-t2i-ood}{0.0672}
\defnum{lidlaw/err-naive-wiki-ood}{0.0033}
\defnum{lidlaw/err-pubmed23-ood}{0.0039}
\defnum{lidlaw/err-sift}{0.0003}
\defnum{lidlaw/err-t2i-ip-ood}{0.0233}
\defnum{lidlaw/err-t2i-ood}{0.0440}
\defnum{lidlaw/err-wiki-ood}{0.0009}
\defnum{lidlaw/exact-frac-max}{1.1}
\defnum{lidlaw/laion-adaptive-cov-max}{0.9963}
\defnum{lidlaw/laion-adaptive-cov-min}{0.9946}
\defnum{lidlaw/laion-adaptive-mass-saving-max}{36}
\defnum{lidlaw/laion-adaptive-mass-saving-min}{23}
\defnum{lidlaw/laion-beyond-p99}{1}
\defnum{lidlaw/laion-cov-in-8-64}{0.9913}
\defnum{lidlaw/laion-cov-ood-8-64}{0.9920}
\defnum{lidlaw/lid-med-deep10m}{21.1}
\defnum{lidlaw/lid-med-glove}{53.0}
\defnum{lidlaw/lid-med-laion-in}{24.5}
\defnum{lidlaw/lid-med-laion-ood}{24.1}
\defnum{lidlaw/lid-med-pubmed23-in}{22.8}
\defnum{lidlaw/lid-med-pubmed23-ood}{25.8}
\defnum{lidlaw/lid-med-sift}{18.4}
\defnum{lidlaw/lid-med-t2i-in}{16.9}
\defnum{lidlaw/lid-med-t2i-ip-in}{18.3}
\defnum{lidlaw/lid-med-t2i-ip-ood}{57.4}
\defnum{lidlaw/lid-med-t2i-ood}{58.4}
\defnum{lidlaw/lid-med-wiki-in}{35.6}
\defnum{lidlaw/lid-med-wiki-ood}{40.0}
\defnum{lidlaw/n-corpora}{7}
\defnum{lidlaw/n-sets}{13}
\defnum{lidlaw/pending}{0}
\defnum{lidlaw/pubmed23-adaptive-cov-max}{0.9940}
\defnum{lidlaw/pubmed23-adaptive-cov-min}{0.9926}
\defnum{lidlaw/pubmed23-adaptive-mass-saving-max}{34}
\defnum{lidlaw/pubmed23-adaptive-mass-saving-min}{29}
\defnum{lidlaw/pubmed23-beyond-p99}{4}
\defnum{lidlaw/pubmed23-cov-in-8-64}{0.9937}
\defnum{lidlaw/pubmed23-cov-ood-8-64}{0.9853}
\defnum{lidlaw/ratio-glove-sift}{2.72}
\defnum{lidlaw/ratio-glove-sift-bins}{8}
\defnum{lidlaw/ratio-laion-in-sift}{2.31}
\defnum{lidlaw/ratio-laion-in-sift-bins}{12}
\defnum{lidlaw/ratio-laion-in-wiki-in}{0.80}
\defnum{lidlaw/ratio-laion-in-wiki-in-bins}{13}
\defnum{lidlaw/ratio-laion-ood-laion-in}{0.99}
\defnum{lidlaw/ratio-laion-ood-laion-in-bins}{17}
\defnum{lidlaw/ratio-pubmed23-in-sift}{2.17}
\defnum{lidlaw/ratio-pubmed23-in-sift-bins}{11}
\defnum{lidlaw/ratio-pubmed23-in-wiki-in}{0.65}
\defnum{lidlaw/ratio-pubmed23-in-wiki-in-bins}{12}
\defnum{lidlaw/ratio-pubmed23-ood-pubmed23-in}{1.04}
\defnum{lidlaw/ratio-pubmed23-ood-pubmed23-in-bins}{16}
\defnum{lidlaw/ratio-sift-deep10m}{1.00}
\defnum{lidlaw/ratio-sift-deep10m-bins}{13}
\defnum{lidlaw/ratio-t2i-in-deep10m}{1.29}
\defnum{lidlaw/ratio-t2i-in-deep10m-bins}{14}
\defnum{lidlaw/ratio-t2i-ip-in-deep10m}{1.50}
\defnum{lidlaw/ratio-t2i-ip-in-deep10m-bins}{15}
\defnum{lidlaw/ratio-t2i-ip-in-t2i-in}{1.11}
\defnum{lidlaw/ratio-t2i-ip-in-t2i-in-bins}{14}
\defnum{lidlaw/ratio-t2i-ip-ood-glove}{1.48}
\defnum{lidlaw/ratio-t2i-ip-ood-glove-bins}{28}
\defnum{lidlaw/ratio-t2i-ip-ood-t2i-ip-in}{2.51}
\defnum{lidlaw/ratio-t2i-ip-ood-t2i-ip-in-bins}{5}
\defnum{lidlaw/ratio-t2i-ip-ood-t2i-ood}{1.05}
\defnum{lidlaw/ratio-t2i-ip-ood-t2i-ood-bins}{34}
\defnum{lidlaw/ratio-t2i-ood-glove}{1.40}
\defnum{lidlaw/ratio-t2i-ood-glove-bins}{28}
\defnum{lidlaw/ratio-t2i-ood-t2i-in}{2.70}
\defnum{lidlaw/ratio-t2i-ood-t2i-in-bins}{4}
\defnum{lidlaw/ratio-wiki-in-glove}{1.11}
\defnum{lidlaw/ratio-wiki-in-glove-bins}{18}
\defnum{lidlaw/ratio-wiki-in-sift}{3.27}
\defnum{lidlaw/ratio-wiki-in-sift-bins}{7}
\defnum{lidlaw/ratio-wiki-ood-wiki-in}{0.95}
\defnum{lidlaw/ratio-wiki-ood-wiki-in-bins}{16}
\defnum{lidlaw/rho-deep10m}{-0.57}
\defnum{lidlaw/rho-glove}{-0.71}
\defnum{lidlaw/rho-laion-in}{-0.59}
\defnum{lidlaw/rho-laion-ood}{-0.60}
\defnum{lidlaw/rho-max}{-0.71}
\defnum{lidlaw/rho-min}{-0.46}
\defnum{lidlaw/rho-pubmed23-in}{-0.55}
\defnum{lidlaw/rho-pubmed23-ood}{-0.58}
\defnum{lidlaw/rho-sift}{-0.52}
\defnum{lidlaw/rho-t2i-in}{-0.58}
\defnum{lidlaw/rho-t2i-ip-in}{-0.57}
\defnum{lidlaw/rho-t2i-ip-ood}{-0.46}
\defnum{lidlaw/rho-t2i-ood}{-0.47}
\defnum{lidlaw/rho-wiki-in}{-0.52}
\defnum{lidlaw/rho-wiki-ood}{-0.48}
\defnum{lidlaw/spread-deep10m}{7}
\defnum{lidlaw/spread-glove}{18}
\defnum{lidlaw/spread-laion-in}{5}
\defnum{lidlaw/spread-laion-ood}{6}
\defnum{lidlaw/spread-max}{18}
\defnum{lidlaw/spread-min}{4}
\defnum{lidlaw/spread-pubmed23-in}{5}
\defnum{lidlaw/spread-pubmed23-ood}{5}
\defnum{lidlaw/spread-sift}{5}
\defnum{lidlaw/spread-t2i-in}{6}
\defnum{lidlaw/spread-t2i-ip-in}{6}
\defnum{lidlaw/spread-t2i-ip-ood}{4}
\defnum{lidlaw/spread-t2i-ood}{4}
\defnum{lidlaw/spread-wiki-in}{4}
\defnum{lidlaw/spread-wiki-ood}{4}
\defnum{lidlaw/t2i-adaptive-cov-max}{0.9689}
\defnum{lidlaw/t2i-adaptive-cov-min}{0.9523}
\defnum{lidlaw/t2i-adaptive-mass-saving-max}{1}
\defnum{lidlaw/t2i-adaptive-mass-saving-min}{1}
\defnum{lidlaw/t2i-beyond-p99}{86}
\defnum{lidlaw/t2i-cov-in-8-64}{0.9953}
\defnum{lidlaw/t2i-cov-ood-8-64}{0.9281}
\defnum{lidlaw/t2i-ip-adaptive-cov-max}{0.9779}
\defnum{lidlaw/t2i-ip-adaptive-cov-min}{0.9762}
\defnum{lidlaw/t2i-ip-adaptive-mass-saving-max}{2}
\defnum{lidlaw/t2i-ip-adaptive-mass-saving-min}{2}
\defnum{lidlaw/t2i-ip-beyond-p99}{66}
\defnum{lidlaw/t2i-ip-cov-in-8-64}{0.9964}
\defnum{lidlaw/t2i-ip-cov-ood-8-64}{0.9543}
\defnum{msmarco/cand-1e7}{285,892}
\defnum{msmarco/cand-1e8}{309,191}
\defnum{msmarco/drift}{0.055}
\defnum{msmarco/drift-1e7}{0.9539}
\defnum{msmarco/drift-1e8}{0.8987}
\defnum{msmarco/n-d}{100,000,000}
\defnum{msmarco/qps-1e8-b16-ks128}{440}
\defnum{msmarco/qps-1e8-b16-ks256}{235}
\defnum{msmarco/qps-1e8-b16-ks64}{820}
\defnum{msmarco/qps-1e8-b8-ks128}{575}
\defnum{msmarco/qps-1e8-b8-ks256}{308}
\defnum{msmarco/qps-1e8-b8-ks64}{1{,}057}
\defnum{msmarco/recall-1e7-b16-ks128}{0.9212}
\defnum{msmarco/recall-1e7-b16-ks256}{0.9539}
\defnum{msmarco/recall-1e7-b16-ks64}{0.8728}
\defnum{msmarco/recall-1e7-b8-ks128}{0.8623}
\defnum{msmarco/recall-1e7-b8-ks256}{0.9111}
\defnum{msmarco/recall-1e7-b8-ks64}{0.7983}
\defnum{msmarco/recall-1e8-b16-ks128}{0.8644}
\defnum{msmarco/recall-1e8-b16-ks256}{0.8987}
\defnum{msmarco/recall-1e8-b16-ks64}{0.8178}
\defnum{msmarco/recall-1e8-b8-ks128}{0.8199}
\defnum{msmarco/recall-1e8-b8-ks256}{0.8602}
\defnum{msmarco/recall-1e8-b8-ks64}{0.7695}
\defnum{ood/adaptive-cov-b16}{0.9939}
\defnum{ood/adaptive-cov-b32}{0.9948}
\defnum{ood/adaptive-cov-b8}{0.9934}
\defnum{ood/cov-in-8-32}{0.9869}
\defnum{ood/cov-in-8-64}{0.9959}
\defnum{ood/cov-ood-8-32}{0.9775}
\defnum{ood/cov-ood-8-64}{0.9926}
\defnum{ood/err-naive-8-32}{0.0094}
\defnum{ood/err-naive-8-64}{0.0033}
\defnum{ood/err-table-8-32}{0.0035}
\defnum{ood/err-table-8-64}{0.0009}
\defnum{ood/exact-frac-max}{0.1}
\defnum{ood/fixed-cov-b16}{0.9911}
\defnum{ood/fixed-cov-b32}{0.9921}
\defnum{ood/fixed-cov-b8}{0.9911}
\defnum{ood/kappa-hi}{124}
\defnum{ood/kappa-lo}{479}
\defnum{ood/lid-med-in}{35.6}
\defnum{ood/lid-med-ood}{40.0}
\defnum{ood/miss-ratio-8-64}{1.8}
\defnum{ood/miss-reduction-max}{34}
\defnum{ood/miss-reduction-min}{26}
\defnum{ood/rho-max}{0.52}
\defnum{rq/deep-ceiling}{0.9838}
\defnum{rq/deep-d}{96}
\defnum{rq/deep-evp-cov}{0.9718}
\defnum{rq/deep-evp-diff}{-0.0121}
\defnum{rq/deep-evptt-cov}{0.9504}
\defnum{rq/deep-evptt-diff}{-0.0335}
\defnum{rq/deep-float-cov}{0.9839}
\defnum{rq/deep-nq}{2{,}000}
\defnum{rq/deep-screenK-evpft}{5{,}000}
\defnum{rq/deep-screenK-sq4}{2{,}000}
\defnum{rq/deep-screenK-sq8}{50}
\defnum{rq/deep-sq4-cov}{0.9839}
\defnum{rq/deep-sq4-diff}{+0.0000}
\defnum{rq/deep-sq8-cov}{0.9839}
\defnum{rq/deep-sq8-diff}{+0.0000}
\defnum{rq/deep-sq8-ndiff}{5}
\defnum{rq/gain-evp-d1024}{13.0}
\defnum{rq/gain-evp-d96}{5.1}
\defnum{rq/gain-sq8-d1024}{3.8}
\defnum{rq/gain-sq8-d96}{2.8}
\defnum{rq/glove-ceiling}{0.8618}
\defnum{rq/glove-d}{100}
\defnum{rq/glove-evp-cov}{0.7883}
\defnum{rq/glove-evp-diff}{-0.0731}
\defnum{rq/glove-evptt-cov}{0.7157}
\defnum{rq/glove-evptt-diff}{-0.1456}
\defnum{rq/glove-float-cov}{0.8613}
\defnum{rq/glove-nq}{2{,}000}
\defnum{rq/glove-screenK-evpft}{2{,}000}
\defnum{rq/glove-screenK-sq4}{200}
\defnum{rq/glove-screenK-sq8}{20}
\defnum{rq/glove-sq4-cov}{0.8541}
\defnum{rq/glove-sq4-diff}{-0.0072}
\defnum{rq/glove-sq8-cov}{0.8613}
\defnum{rq/glove-sq8-diff}{+0.0000}
\defnum{rq/glove-sq8-ndiff}{49}
\defnum{rq/nstar-2}{6.1}
\defnum{rq/nstar-4}{4.8}
\defnum{rq/nstar-8}{3.4}
\defnum{rq/nstar-float32}{1.2}
\defnum{rq/page-d768-evp}{20}
\defnum{rq/page-d768-evpsplit}{21}
\defnum{rq/page-d768-f32}{1}
\defnum{rq/page-d768-sq4}{10}
\defnum{rq/page-d768-sq8}{5}
\defnum{rq/page-d96-evp}{146}
\defnum{rq/page-d96-evpsplit}{170}
\defnum{rq/page-d96-f32}{10}
\defnum{rq/page-d96-sq4}{78}
\defnum{rq/page-d96-sq8}{40}
\defnum{rq/page-gain-d96}{4.25}
\defnum{rq/sift-ceiling}{0.9862}
\defnum{rq/sift-d}{128}
\defnum{rq/sift-float-cov}{0.9867}
\defnum{rq/sift-nq}{2{,}000}
\defnum{rq/sift-screenK-sq4}{50}
\defnum{rq/sift-screenK-sq8}{20}
\defnum{rq/sift-sq4-cov}{0.9866}
\defnum{rq/sift-sq4-diff}{-0.0001}
\defnum{rq/sift-sq8-cov}{0.9867}
\defnum{rq/sift-sq8-diff}{+0.0000}
\defnum{rq/sift-sq8-ndiff}{4}
\defnum{rq/sq8-worst-ci}{0.0008}
\defnum{rq/sq8-worst-diff}{0.0000}
\defnum{rq/sq8-worst-ndiff}{49}
\defnum{rq/wiki-ceiling}{0.9864}
\defnum{rq/wiki-d}{1024}
\defnum{rq/wiki-evp-cov}{0.9841}
\defnum{rq/wiki-evp-diff}{-0.0018}
\defnum{rq/wiki-evptt-cov}{0.9827}
\defnum{rq/wiki-evptt-diff}{-0.0032}
\defnum{rq/wiki-float-cov}{0.9859}
\defnum{rq/wiki-nq}{2{,}000}
\defnum{rq/wiki-screenK-evpft}{200}
\defnum{rq/wiki-screenK-sq4}{50}
\defnum{rq/wiki-screenK-sq8}{20}
\defnum{rq/wiki-sq4-cov}{0.9857}
\defnum{rq/wiki-sq4-diff}{-0.0002}
\defnum{rq/wiki-sq8-cov}{0.9859}
\defnum{rq/wiki-sq8-diff}{+0.0000}
\defnum{rq/wiki-sq8-ndiff}{0}
\defnum{s4/d2-floor-cap-mb}{128}
\defnum{s4/d2-floor-qps}{0.5}
\defnum{s4/d2-floor-recall}{0.9925}
\defnum{s4/d2-floor-threads}{2}
\defnum{s4/d2-qps-16g}{25.6}
\defnum{s4/d2-route-ms}{9.1}
\defnum{s4/d3-cov-diff-16g}{+0.0000}
\defnum{s4/d3-extra-io-mb}{4.2}
\defnum{s4/d3-floor-b-per-object}{0.10}
\defnum{s4/d3-floor-cap-mb}{96}
\defnum{s4/d3-floor-qps}{1.6}
\defnum{s4/d3-floor-recall}{0.9925}
\defnum{s4/d3-floor-threads}{2}
\defnum{s4/d3-qps-16g}{29.1}
\defnum{s4/d3-recall-diff-16g}{+0.0000}
\defnum{s4/d3-route-ms}{10.0}
\defnum{s4/d3-router-mib}{3.6}
\defnum{s4/d3-router-ratio}{53}
\defnum{s4/e1-cov-diff-16g}{+0.0000}
\defnum{s4/e1-recall-diff-16g}{+0.0000}
\defnum{s4/e1-route-ms-float}{9.1}
\defnum{s4/e1-route-ms-sq8}{9.1}
\defnum{s5/128m-t2-u8-cov}{0.9925}
\defnum{s5/128m-t2-u8-qps}{1.9}
\defnum{s5/128m-t2-u8-recall}{0.9925}
\defnum{s5/128m-t2-u8-route-ms}{169.2}
\defnum{s5/16g-t64-cov-diff}{+0.0000}
\defnum{s5/16g-t64-recall-diff}{+0.0000}
\defnum{s5/16g-t64-route-ms-float}{9.1}
\defnum{s5/16g-t64-u8-cov}{0.9936}
\defnum{s5/16g-t64-u8-qps}{30.3}
\defnum{s5/16g-t64-u8-recall}{0.9904}
\defnum{s5/16g-t64-u8-route-ms}{8.7}
\defnum{s5/256m-t4-cov-diff}{+0.0000}
\defnum{s5/256m-t4-recall-diff}{+0.0000}
\defnum{s5/256m-t4-route-ms-float}{90.7}
\defnum{s5/256m-t4-u8-cov}{0.9925}
\defnum{s5/256m-t4-u8-qps}{3.7}
\defnum{s5/256m-t4-u8-recall}{0.9925}
\defnum{s5/256m-t4-u8-route-ms}{87.3}
\defnum{s5/512m-t8-cov-diff}{+0.0000}
\defnum{s5/512m-t8-recall-diff}{+0.0000}
\defnum{s5/512m-t8-route-ms-float}{48.1}
\defnum{s5/512m-t8-u8-cov}{0.9925}
\defnum{s5/512m-t8-u8-qps}{6.9}
\defnum{s5/512m-t8-u8-recall}{0.9925}
\defnum{s5/512m-t8-u8-route-ms}{47.3}
\defnum{s5/96m-t2-u8-cov}{0.9925}
\defnum{s5/96m-t2-u8-qps}{1.9}
\defnum{s5/96m-t2-u8-recall}{0.9925}
\defnum{s5/96m-t2-u8-route-ms}{168.0}
\defnum{s5/float-router-mib}{188.7}
\defnum{s5/router-ratio}{2.39}
\defnum{s5/u8-floor-b-per-object}{0.10}
\defnum{s5/u8-floor-cap-mb}{96}
\defnum{s5/u8-floor-qps}{1.9}
\defnum{s5/u8-floor-recall}{0.9925}
\defnum{s5/u8-floor-threads}{2}
\defnum{s5/u8-router-mib}{78.9}
\defnum{sat/graft-deep100m}{0.9992}
\defnum{sat/hnsw-deep100m}{0.9989}
\defnum{sat/hnsw-deep100m-cold-qps}{540}
\defnum{sat/hnsw-deep100m-qps}{531}
\defnum{sat/hnsw-glove}{0.9863}
\defnum{sat/hnsw-glove-cold-qps}{239}
\defnum{sat/hnsw-glove-qps}{268}
\defnum{sat/hnsw-sift}{0.9991}
\defnum{sat/hnsw-sift-cold-qps}{6{,}140}
\defnum{sat/hnsw-sift-qps}{4{,}166}
\defnum{sat/misi-glove}{0.9803}
\defnum{sat/spann-int8}{0.9561}
\defnum{scale/evals-1e6-995}{21,285}
\defnum{scale/evals-1e6-998}{30,603}
\defnum{scale/evals-1e7-995}{50,514}
\defnum{scale/evals-1e7-998}{76,391}
\defnum{scale/evals-1e8-995}{185,856}
\defnum{scale/evals-1e8-998}{185,856}
\defnum{scale/evals-1e9-995}{373,192}
\defnum{scale/evals-1e9-998}{696,138}
\defnum{scale/exp-hi-995}{0.57}
\defnum{scale/exp-hi-998}{0.39}
\defnum{scale/exp-lo-995}{0.38}
\defnum{scale/exp-lo-998}{0.40}
\defnum{scale/exp-top-995}{0.30}
\defnum{scale/exp-top-998}{0.57}
\defnum{scale/frac-1e6-995}{2.128}
\defnum{scale/frac-1e6-998}{3.060}
\defnum{scale/frac-1e7-995}{0.505}
\defnum{scale/frac-1e7-998}{0.764}
\defnum{scale/frac-1e8-995}{0.186}
\defnum{scale/frac-1e8-998}{0.186}
\defnum{scale/frac-1e9-995}{0.037}
\defnum{scale/frac-1e9-998}{0.070}
\defnum{sift/flat-max-qps}{6{,}977}
\defnum{sift/flat-max-recall}{0.9993}
\defnum{sift/interleaved-seeds}{3}
\defnum{sift/kb}{16}
\defnum{sift/matched-qps}{8{,}836}
\defnum{sift/matched-recall}{0.9992}
\defnum{sift/matched-speedup}{2.1}
\defnum{sift/qps-at-099}{27{,}614}
\defnum{sift/recall-at-099}{0.9917}
\defnum{sift/seeds}{5}
\defnum{sift/speedup-floor}{0.9893}
\defnum{sift/speedup-hi}{2.1}
\defnum{sift/speedup-lo}{1.4}
\defnum{split/bytes-per-query}{1312}
\defnum{split/identical}{identical}
\defnum{split/recall}{0.9978}
\defnum{split/router-mib}{19.5}
\defnum{sq8/gain-deep10m}{2.8}
\defnum{sq8/gain-max}{2.8}
\defnum{upward/agree-d1-ks64}{0.9995}
\defnum{upward/agree-d2-ks64}{0.9843}
\defnum{upward/compose-err}{0.0008}
\defnum{upward/max-recall-loss}{0.0009}
\defnum{upward/qps-d1-ks256}{1{,}793}
\defnum{upward/qps-d1-ks64}{4{,}829}
\defnum{upward/qps-d2-ks256}{1{,}341}
\defnum{upward/qps-d2-ks64}{1{,}874}
\defnum{upward/ratio}{50}
\defnum{upward/recall-d1-ks256}{0.9981}
\defnum{upward/recall-d1-ks64}{0.9874}
\defnum{upward/recall-d2-ks256}{0.9975}
\defnum{upward/recall-d2-ks64}{0.9866}
\defnum{upward/resident-flat-mib}{854}
\defnum{upward/resident-mib}{17.1}
\defnum{upward/thru-ks256}{0.75}
\defnum{upward/thru-ks64}{0.39}

\newtheorem{proposition}{Proposition}

\newtheorem{remark}{Remark}

\newcommand{\solo}{\textsf{SOLO}}
\newcommand{\misi}{\textsf{misi}}
\newcommand{\graft}{\textsf{GRAFT}}
\newcommand{\kb}{k_b}
\newcommand{\ks}{k_s}
\newcommand{\kslocal}{k_s'}
\newcommand{\blocal}{b'}

\title{\solo: Certified-Recall Metric Similarity Search\\
with Scan-Only Sampled Inverted Lists}
\author{Edgar Ch\'avez\\ Computer Science Department, CICESE\\ \texttt{elchavez@cicese.edu.mx}}
\date{October 4, 2026}

\begin{document}
\maketitle

\begin{abstract}
In every fast nearest-neighbor index, recall is measured, never predicted: each operating point is tuned by serving it against ground truth. \solo{} is an index whose recall is computed from the index itself, before any query is served.

\solo{} is an inverted file whose vocabulary is a uniform random
sample $S$ of the database, $|S|=\alpha n$. At build time each object
$x$ is searched once in a router over $S$ for its $\kb$ nearest sample
points, and that ranked list, $\mathrm{sig}(x)$, is the only thing
stored about $x$. At serve time two dials are chosen without touching
the stored data: a replication $b\le\kb$, under which $x$ is posted to
the first $b$ entries of $\mathrm{sig}(x)$, and a per-query $\ks$, the
number of nearest sample points the query $q$ is routed to. Serving
reads the $b$-lists of those $\ks$ sample points, evaluates the true
distance $d(q,\cdot)$ on every object in their union, and returns the
$k$ smallest. There is no beam, no vote, no pruning bound: the only
approximate step is the choice of which lists are read.

This makes recall a counting statement. Write $T_b(y)$ for the first
$b$ entries of $\mathrm{sig}(y)$ and $R_{\ks}(q)$ for the $\ks$ sample
points $q$ is routed to. A true $k$-nearest neighbor $y$ of $q$ is
returned if and only if $T_b(y)\cap R_{\ks}(q)\neq\emptyset$: if the
two sets share a sample point, $y$ lies in a list that is scanned in
full and its exact distance ranks it; if they share none, $y$ is never
evaluated and nothing downstream can recover it. Recall@$k$ is
therefore the probability of that intersection over the query
distribution and over the $k$ true neighbors, which we call coverage,
$\mathrm{cov}(b,\ks)$. The equality is exact up to ties at rank $k$,
independent of the router (replacing it by exact search over $S$ moves
recall by at most $10^{-4}$), and independent of the scan kernel as
long as the kernel re-ranks with the true distance.

Because the intersection is a test on stored integers, coverage is
computable without serving. The certificate is obtained in three
steps. (1) Draw $m$ queries from the operator's own workload and
compute their exact $k$ nearest neighbors once by brute force; this
pass is the only cost that grows with $n$. (2) Route each query once to
its $\ks^{\max}$ nearest sample points. (3) For every pair $(q,y)$ and
every $(b,\ks)$ on the grid, test whether $\mathrm{sig}(y)[1..b]$ and
$R(q)[1..\ks]$ intersect; the fraction of pairs that do is
$\mathrm{cov}(b,\ks)$. Because $b$ and $\ks$ index prefixes of the two
stored rankings, one pass over the pairs yields the whole $(b,\ks)$
surface in $m\,k\,\kb\,\ks^{\max}$ integer comparisons, seconds at
$m=10^4$, including configurations never deployed and values of $b$
never materialized as lists. For the recursive index the test composes
level by level: $y$ is returned iff it is covered at every level on
the way to a leaf.

The certificate is then checked, not assumed. The same queries are
served at the same $(b,\ks)$ and recall is measured against the same
ground truth; the residual is the certificate's error. On SIFT-1M,
GloVe-1.2M and Deep-10M at $10^4$ queries it is at most $10^{-4}$
except where exact ground truth itself has ties at rank $k$ (SIFT,
$9\times10^{-4}$, with zero membership misses). On Deep-100M it is
identical to four decimals with exact and with approximate routing. On
MS MARCO Web Search at $10^8$, $768$-dimensional embeddings under inner
product with queries from a different distribution than the passages,
$\num{cert/msmarco100m-points}$ operating points from recall
$\num{msmarco/recall-1e8-b8-ks64}$ to $\num{cert/msmarco100m-best-recall}$
agree to $\num{cert/msmarco100m-max-err}$, so the certificate is right
where recall is poor as well as where it is high. On Deep-1B at
$b=8$, $\ks=512$ under a 16\,GB cap, coverage $\num{deep1b/s2/cov-same}$ less a
measured screen loss of $\num{deep1b/s2/screen-loss}$ equals the served
$\num{deep1b/s2/recall}$. The two-level recursive index lands on its
composed prediction to $2\times10^{-4}$ on GloVe and $8\times10^{-4}$
on Deep-100M. The certificate's own sampling precision is measured by a
clustered bootstrap over queries: at coverage $\ge0.99$ the $95\%$
interval has half-width $\le\num{certboot/halfwidth-max-99}$, and the
clustered standard error exceeds the binomial one by at most
$\num{certboot/ratio-max-99}\times$, rising to
$\num{certboot/ratio-max}\times$ where coverage is
$\num{certboot/glove/min-cov}$ because misses concentrate in hard
queries.

What the certificate guarantees is the expected recall@$k$ on the
query distribution it was computed on; it is an identity, not a bound,
and not a per-query statement. A different workload needs its own
certificate, which costs one brute-force pass and no serving, and the
index carries a per-query provision for queries from outside the
certified distribution, from the order statistics of the routed
distances. Under updates the integer part is maintained in place (an
insert adds one signature, a delete removes one); the ground-truth
part is refreshed with the database when a re-stated figure is wanted.
No navigable graph has an analogous object: its recall is a property
of the traversal, so the only way to learn it is to run it, cell by
cell, and a rebuilt graph must be re-served.

The rest is the same rule applied recursively: a list that outgrows a bound is sampled and split like the database, and so is the vocabulary itself, which is what drives resident memory down. Deep-100M is served at recall $\num{capped/1g-recall-d2}$ from $\num{capped/peak-gb}$\,GB of enforced resident memory, Deep-1B at $\num{deep1b/floor-recall}$ from $\num{s4/d3-floor-cap-mb}$\,MB. Inserts are one search, deletes are exact, and throughput reaches $\num{deep100m/speedup-hi}\times$ a tuned HNSW at $10^8$. Every number is reported against HNSW, DiskANN, \graft{}, NAPP, \misi{}, SPANN, ScaNN and RaBitQ on the same hardware and ground truth.
\end{abstract}

\section{Introduction}\label{sec:intro}

In every fast nearest-neighbor index, recall is measured and never
predicted. A graph's recall is an emergent property of its construction,
known only after serving against ground truth, one operating point at a
time; and at $10^8$ objects the dial that buys it runs out --- an
exhaustively tuned HNSW on Deep-100M stops at recall@10
$\num{sat/hnsw-deep100m}$ at $ef{=}4000$, with no operating point above it
at any cost. Two further frictions follow the same architecture: a
disk-resident graph pins gigabytes of navigation state (DiskANN needs
12\,GB for Deep-100M on our hardware and fails below 9.6\,GB), and
deleting from a navigable graph is delicate enough that update layers are
systems in their own right \cite{singh2021freshdiskann,xu2023spfresh}.

\solo{} (\emph{Scan-Only Lists, One rule})\footnote{Officially
\textsf{SOLOVINO}, after the archetypal stray of Spanish folk naming
(\emph{``solo vino''} --- it came alone): no clustering, no training, no
graph; a random sample and one rule.} is an index whose recall is a
property of the stored index. It is an inverted file over a vocabulary
that is a uniform random sample $S$ of the database itself,
$|S|=\alpha n$. Each object's $\kb$ nearest sample points are stored once,
ranked; at serve time an object is posted under the first $b\le\kb$ of
them, a query is routed to its $\ks$ nearest sample points, and the
touched lists are scanned exhaustively with the true distance. There is no
vote, no beam, no pruning bound: the only approximation in the serving
path is which lists are read. A true neighbor is therefore returned if
and only if it shares a sample point with the query, which is a test on
stored integers, so one ground-truth pass over a sample of the operator's
own queries \emph{certifies} the recall of every $(b,\ks)$ configuration
at once, with nothing served (\S\ref{sec:identity}). The certificate
matches served recall to four decimals from $10^6$ to $10^9$ objects,
including MS MARCO Web Search at $10^8$ --- $768$-dimensional text
embeddings under inner product with queries from a different
distribution than the passages --- where the index performs worst and
the certificate is still right to $\num{cert/msmarco100m-max-err}$
(\S\ref{sec:msmarco}).

One rule completes the design. Random sampling does not balance list
sizes, so a list that outgrows a bound is treated as what it is, a
database: sampled, split by the same rank-$b$ assignment, routed the
same way, its leaves always scanned (\S\ref{sec:recursion}). Pointed at
the router --- itself a database of $\alpha n$ points --- the rule divides
the resident state by $1/\alpha$ per level, and that is what the memory
results rest on. Deep-100M is served at recall
$\num{capped/1g-recall-d2}$ from $\num{capped/peak-gb}$\,GB of enforced
resident memory and from $\num{capped/solo-floor512-mb}$\,MB on
$\num{capped/solo-floor512-threads}$ threads; Deep-1B at
$\num{deep1b/floor-recall}$ from $\num{deep1b/floor-cap-mb}$\,MB, and
from $\num{s4/d3-floor-cap-mb}$\,MB at depth 3
(\S\ref{sec:capped},~\S\ref{sec:deep1b}). Six other systems run through
the same enforced caps either stop below 6--48\,GB or saturate below
recall $0.88$ (\S\ref{sec:capped-baselines}). In RAM, a list-major
stream of 4-bit codes with exact re-ranking serves SIFT-1M at
$\num{sift/matched-recall}$ at $\num{sift/matched-qps}$ queries/s on a
laptop and reaches $\num{deep100m/speedup-hi}\times$ a tuned HNSW at
$10^8$, with operating points to the right of where the graph saturates
(\S\ref{sec:compare}). Inserts are one search and deletes are exact.

\paragraph{Contributions.}
\begin{enumerate}
\item \textbf{The recall certificate.} Under scan-only serving, recall
  equals a coverage probability computable from the stored signatures
  (Proposition~\ref{obs:identity}); one ground-truth pass certifies the
  whole operating surface. Measured to hold within $5\times10^{-4}$ at
  $10^6$--$10^8$ on three datasets and within $10^{-4}$ at $10^8$ and
  $10^9$ on two more, with a measured precision (\S\ref{sec:identity}),
  and carrying a per-query provision for queries from outside the
  certified distribution, fitted on seven corpora and thirteen query sets
  (\S\ref{sec:oos}).
\item \textbf{The equal-work law.} Coverage is a function of the touched
  posting mass $b\cdot\ks$, nearly independent of its factorization; its
  level is a one-scalar signature of the dataset, and the evaluations
  needed at fixed recall grow as $n^{0.4\text{--}0.6}$ over four measured
  decades (\S\ref{sec:laws}).
\item \textbf{The one-rule recursion.} Bounded lists without clustering,
  validated against an ablation that attaches a graph index to each
  oversized list: the recursion wins on recall and throughput together,
  with a build $25\times$ cheaper (\S\ref{sec:recursion},
  \S\ref{sec:ablation}).
\item \textbf{The resident-memory law.} Serving from storage keeps three
  terms resident --- a router of $\alpha^{L}n$ points, per-thread scratch,
  page cache --- and the rule applied to the router divides the first by
  $1/\alpha$ per level. The law is measured through depth 3 at $10^9$,
  where the floor is $\num{s4/d3-floor-b-per-object}$ resident bytes per
  object, and it is stated with what it does not say about $n$
  (\S\ref{sec:ceiling}).
\end{enumerate}

\noindent The serving stack that makes these claims practical --- lossless
8-bit and 4-bit screens with exact re-ranking, a list-major code stream
with rank-ordered postings, an mmap-able router, posting-ordered leaf
blocks read by explicit \texttt{pread} --- is described in
\S\ref{sec:sq}--\ref{sec:kernels}, and every number in the paper is
measured on the same machines, queries and ground truth as HNSW, DiskANN,
\graft{}, NAPP, vote-based serving of the identical postings (\misi{}),
SPANN, ScaNN and RaBitQ.

\paragraph{Scope.} The index is defined over a black-box distance
\cite{chavez2001searching} and one interchangeable subroutine, an
approximate nearest-neighbor index over the $\alpha n$ sample points used
only to route; replacing it with exhaustive search over the sample moves
no recall reported here by more than $10^{-4}$
(Remark~\ref{rem:router}). The representation --- a vocabulary drawn from
the database at a fixed fraction of it, and each object's ranked signature
over it --- was introduced by \misi{} \cite{chavez2026misi}, which served
it by voting; \solo{} keeps the representation and replaces the serving
theory, and \S\ref{sec:design} is self-contained. The paper does not claim
efficiency per gigabyte against DiskANN, which is faster where it fits;
it does not claim competitive high recall on MS MARCO at $10^8$, where
the certificate predicts the shortfall; filtered search is analyzed, not
measured (\S\ref{sec:filter}); and the throughput comparison is on one
machine (\S\ref{sec:setup}).

Section~\ref{sec:design} defines the index, the certificate and the
serving stack. Section~\ref{sec:laws} gives the measured laws and
\S\ref{sec:oos} the certificate out of sample.
Section~\ref{sec:exp} is the evaluation: the operating surface, serving
under enforced caps against six systems, the in-RAM comparison, MS MARCO
and Deep-1B. Section~\ref{sec:ceiling} states the resident-memory law and
\S\ref{sec:where} what it places \solo{} next to. Related work,
limitations and the conclusion follow.

\section{The \solo{} index}\label{sec:design}

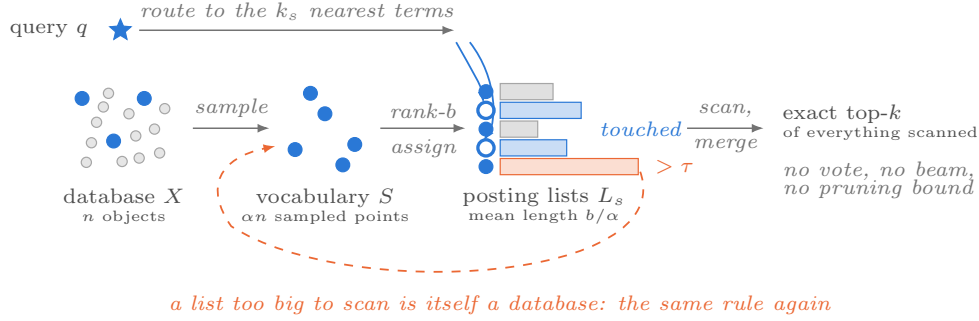
\begin{figure}[t]
\centering
\resizebox{0.80\textwidth}{!}{
\begin{tikzpicture}[
  font=\small,
  obj/.style={circle, draw=black!35, fill=black!10, inner sep=0pt,
              minimum size=3.4pt},
  smp/.style={circle, draw=sloblue, fill=sloblue, inner sep=0pt,
              minimum size=4.6pt},
  sel/.style={circle, draw=sloblue, fill=white, inner sep=0pt,
              minimum size=5.4pt, line width=1.1pt},
  bar/.style={draw=black!35, fill=black!12, minimum height=5.4pt,
              inner sep=0pt, anchor=west},
  hit/.style={draw=sloblue, fill=sloblue!25, minimum height=5.4pt,
              inner sep=0pt, anchor=west},
  big/.style={draw=sloorange, fill=sloorange!22, minimum height=5.4pt,
              inner sep=0pt, anchor=west},
  note/.style={font=\scriptsize, color=black!70},
  tag/.style={font=\scriptsize\itshape, color=black!55},
  ar/.style={-{Latex[length=4pt,width=3pt]}, draw=black!55, line width=0.5pt},
  dar/.style={-{Latex[length=4pt,width=3pt]}, draw=sloorange,
              line width=0.6pt, dashed}
]

\begin{scope}[shift={(0,0)}]
  \foreach \p in {(0.10,0.05),(0.72,0.10),(0.22,0.52),(0.60,0.62),
                  (0.90,0.44),(0.36,0.86),(1.02,0.68),
                  (0.16,0.28),(0.52,0.06),(0.96,0.16),(0.30,0.70),(0.68,0.36)}
     \node[obj] at \p {};
  \foreach \p in {(0.42,0.30),(0.05,0.80),(0.78,0.80)} \node[smp] at \p {};
  \node[note, align=center] at (0.54,-0.46)
       {database $X$\\[-2pt] \tiny $n$ objects};
\end{scope}

\draw[ar] (1.35,0.45) -- node[tag, above, pos=0.5, yshift=0pt] {sample}
          (2.25,0.45);

\begin{scope}[shift={(2.45,0)}]
  \foreach \p in {(0.10,0.20),(0.45,0.62),(0.80,0.26),(0.28,0.86),(0.66,0.02)}
     \node[smp] at \p {};
  \node[note, align=center] at (0.45,-0.46)
       {vocabulary $S$\\[-2pt] \tiny $\alpha n$ sampled points};
\end{scope}

\draw[ar] (3.55,0.45) -- node[tag, above, pos=0.5] {rank-$b$}
                          node[tag, below, pos=0.5] {assign}
          (4.55,0.45);

\begin{scope}[shift={(4.95,0)}]
  \node[smp] at (-0.16,0.88) {}; \node[bar, minimum width=0.62cm] at (0,0.88) {};
  \node[sel] at (-0.16,0.66) {}; \node[hit, minimum width=0.95cm] at (0,0.66) {};
  \node[smp] at (-0.16,0.44) {}; \node[bar, minimum width=0.44cm] at (0,0.44) {};
  \node[sel] at (-0.16,0.22) {}; \node[hit, minimum width=0.78cm] at (0,0.22) {};
  \node[smp] at (-0.16,0.00) {}; \node[big, minimum width=1.62cm] at (0,0.00) {};
  \node[tag, color=sloorange, anchor=west] at (1.70,0.00) {$>\tau$};
  \node[note, align=center] at (0.52,-0.46)
       {posting lists $L_s$\\[-2pt] \tiny mean length $b/\alpha$};
  \node[tag, color=sloblue, anchor=west] at (1.05,0.44) {touched};
\end{scope}

\draw[dar] (6.60,-0.06) .. controls (7.05,-1.35) and (3.20,-1.60) ..
           (1.80,-0.70) .. controls (1.48,-0.24) and (1.66,0.10) .. (2.32,0.24);
\node[tag, color=sloorange, align=center] at (4.95,-1.62)
     {a list too big to scan is itself a database: the same rule again};

\node[note] at (-0.34,1.60) {query $q$};
\node[star, star points=5, star point ratio=2.2, draw=sloblue,
      fill=sloblue, inner sep=1.2pt] (q) at (0.50,1.60) {};
\draw[ar] (0.72,1.60) -- node[tag, above] {route to the $k_s$ nearest terms}
          (4.42,1.60);
\draw[ar, draw=sloblue, line width=0.6pt]
      (4.44,1.46) .. controls (4.62,1.14) and (4.74,0.94) .. (4.76,0.74);
\draw[ar, draw=sloblue, line width=0.6pt]
      (4.60,1.46) .. controls (4.94,1.02) and (4.90,0.52) .. (4.76,0.30);

\draw[ar] (7.15,0.45) -- node[tag, above, pos=0.5] {scan,}
                          node[tag, below, pos=0.5] {merge}
          (8.05,0.45);
\node[note, align=left, anchor=west] at (8.15,0.52)
     {exact top-$k$\\[-2pt] \tiny of everything scanned};
\node[tag, align=left, anchor=west] at (8.15,-0.16)
     {no vote, no beam,\\[-2pt] no pruning bound};

\end{tikzpicture}}
\caption{The whole index. A uniform random sample of the database is the
vocabulary; each object is posted to its $b$ nearest sample points; a
query is routed to its $\ks$ nearest sample points and the touched lists
are scanned with the true distance. The one rule that is not a
straight line is the dashed one: a list that outgrows $\tau$ is treated
as a database and decomposed by the same construction. Nothing in the
serving path ranks anything except by $d(q,\cdot)$, which is why recall
is Proposition~\ref{obs:identity}.}
\label{fig:arch}
\end{figure}

The index and its serving path are Algorithm~\ref{alg:solo} in full,
and Figure~\ref{fig:arch} is the same thing as a picture.
The subsections explain the four choices it contains --- a sampled
vocabulary, rank-$b$ assignment, exhaustively scanned leaves, and the
one recursion --- and what each one buys.

\begin{table}[t]
\centering\small
\begin{tabular}{@{}p{2.6cm}p{9.6cm}p{2.6cm}@{}}
\toprule
Name & Meaning & In this paper\\
\midrule
$X$, $n$ & the database and its size & $10^6$--$10^9$\\
$d$ & the distance, a black box: the index evaluates it and compares its values, nothing else & $\ell_2$, angular, inner product\\
$S$, $\alpha$ & the \emph{vocabulary}: a uniform random sample of $X$ with $|S|=\alpha n$ & $\alpha=2\%$\\
term & a sample point $s\in S$ in its role as a word of the vocabulary & \\
router & any high-recall ANN index over $S$; its only job is to find nearest terms & \graft{}\\
$\mathrm{sig}(x)$, $\kb$ & the \emph{signature} of $x$: its $\kb$ nearest terms, ranked by distance, stored once at build & $\kb=64$\\
$b$ & \emph{replication}: $x$ is posted to the first $b$ entries of $\mathrm{sig}(x)$; a prefix, so any $b\le\kb$ is a serve-time choice & $4$--$64$\\
$L_s(b)$ & the \emph{posting list} of term $s$: every $x$ with $s$ among its first $b$ terms; mean length $b/\alpha$ & \\
$\ks$ & the per-query \emph{dial}: $q$ is routed to its $\ks$ nearest terms and their lists are scanned & $32$--$512$\\
$\tau$ & the list bound: a list longer than $\tau$ is decomposed by the same rule, with its own $\blocal,\kslocal$ & \\
leaf & a list of length $\le\tau$ at any depth; always scanned in full with $d$ & \\
$\mathrm{cov}(b,\ks)$ & \emph{coverage}: the probability that a true neighbor and its query share a term & $=\mathrm{recall}$\\
certificate & a recall figure computed from the stored signatures and one ground-truth pass, not by serving & \S\ref{sec:identity}\\
\bottomrule
\end{tabular}

\caption{Every name the paper uses. The representation (a sampled
vocabulary and ranked signatures over it) was introduced by \misi{}
\cite{chavez2026misi}; nothing from that paper is needed here.}
\label{tab:names}
\end{table}

\paragraph{The index in one page.}
Table~\ref{tab:names} fixes the vocabulary, and a worked instance
fixes the orders of magnitude. Take $n=10^8$ objects under a black-box
distance $d$. \emph{Vocabulary.} Draw two million of them at random
($\alpha=2\%$); this set $S$ is the vocabulary, and a router --- any
high-recall ANN index --- is built over it. \emph{Build.} Each of the
$10^8$ objects is searched once in the router for its $\kb=64$
nearest sample points; the ranked list of their identifiers is its
signature (256 bytes at four bytes per identifier). That search is the
whole build: $10^8$ independent searches, in any order, on any number
of machines, with no step that revisits an earlier decision.
\emph{Posting lists} are the signatures read column-wise: the list of
term $s$ at replication $b$ holds every object with $s$ among the first
$b$ entries of its signature. Because the signature is ranked, the
lists for $b=8$ are prefixes of the lists for $b=16$; one build
materializes every $b\le\kb$, and $b$ is chosen when the index is
served, not when it is built. At $b=16$ the mean list holds
$16/0.02=800$ objects. \emph{Serving.} Route $q$ to its $\ks$ nearest
sample points --- one router search, the same operation the build
performs --- read those $\ks$ lists, evaluate $d(q,\cdot)$ on every
object in their union, and return the $k$ closest. At $\ks=256$ the
union holds about $256\times800\approx2\times10^5$ candidates before
deduplication, $0.2\%$ of the database, each evaluated exactly once.
\emph{Recursion.} A term whose list outgrows the bound $\tau$ (hubs
produce them) is handled by the same three steps inside the list:
sample it, post its members to their $\blocal$ nearest local sample
points, route an entering query to $\kslocal$ local lists. Leaves are
scanned in full at every depth. \emph{What is approximate.} Exactly
one thing: which lists are read. The router decides that, and
Remark~\ref{rem:router} measures how little its approximation matters.
Everything downstream of the choice of lists is exact, and
\S\ref{sec:identity} turns that into a recall figure the stored index
computes about itself.

\begin{algorithm}[t]
\caption{\solo{}. $d$ is the black-box distance;
$\textsc{Route}(q,S,m)$ returns the $m$ nearest points of $S$ to $q$
by any high-recall ANN index (Remark~\ref{rem:router}); and
$L_s(b)=\{x : s\in\mathrm{sig}(x)[1..b]\}$ is the posting list of $s$
at replication $b$, a prefix restriction of the stored signatures.}
\label{alg:solo}
\begin{algorithmic}[1]
\Procedure{Build}{$X,\alpha,\kb,\tau$}
  \State $S \gets$ uniform random sample of $X$, $|S| = \alpha|X|$
  \State build a router over $S$
  \ForAll{$x \in X$ \textbf{in parallel}}
    \State $\mathrm{sig}(x) \gets \textsc{Route}(x,S,\kb)$
      \Comment{ranked; stored once, serves every $b \le \kb$}
  \EndFor
  \ForAll{$s \in S$ with $|L_s(\kb)| > \tau$}
    \State \Call{Build}{$L_s(\kb),\alpha',\kb',\tau$}
      \Comment{a list too big to scan is a database}
  \EndFor
\EndProcedure
\Statex
\Procedure{Serve}{$q,b,\ks,k$}
  \State $V \gets \emptyset$
  \ForAll{$s \in \textsc{Route}(q,S,\ks)$}
    \If{$L_s(b)$ is a leaf}
      \State $V \gets V \cup L_s(b)$ \Comment{scan it all; no pruning}
    \Else
      \State $V \gets V \cup {}$\Call{Serve$_{L_s}$}{$q,\blocal,\kslocal,k$}
    \EndIf
  \EndFor
  \State \Return the $k$ objects of $V$ with smallest $d(q,\cdot)$
    \Comment{deduplicated; exact over $V$}
\EndProcedure
\end{algorithmic}
\end{algorithm}

Two properties are visible in the pseudocode and are the source of
everything that follows. \textsc{Build} contains no loop that revisits
a decision: no clustering iteration, no graph refinement, no insertion
order --- the objects' signatures are independent, so construction is
embarrassingly parallel and bitwise reproducible. \textsc{Serve}
contains no comparison between candidates other than the true
distance: no vote, no scoreboard, no beam, no bound. The second
property is what makes recall a theorem (\S\ref{sec:identity}); the
first is what makes the index cheap to build and exact to update
(\S\ref{sec:dynamics}).

\subsection{Vocabulary and assignment}\label{sec:vocab}
Given $X$ with $|X|=n$ in a metric space with black-box distance $d$,
draw a uniform sample $S\subseteq X$, $|S|=\alpha n$ (all experiments:
$\alpha=2\%$). Build any high-recall ANN index over $S$ (the
\emph{router}). For each $x\in X$, one router search yields its ranked
signature: the $\kb$ nearest sample points in order (we build with
$\kb{=}64$). \emph{Assignment is truncation}: posting $x$ to its $b$
nearest sample points, $b\le\kb$, is a serve-time choice --- the
inverted file for any $b$ is a prefix-restriction of the stored ranked
signatures, so one build serves every $b$. Mean list length is
$b/\alpha$. Construction is $n$ independent searches: embarrassingly
parallel, deterministic for fixed seed, with no iterative refinement of
any kind. Sample points are indexed like any other object; each
carries itself at rank~1 of its own signature, so a query whose true
neighbor \emph{is} a vocabulary point retrieves it from that point's
own list (measured: probability $1.0$) --- the double role costs
$0.3\%$ of scanned work in re-evaluated routing distances and one of
$b$ replication slots for the $\alpha$ fraction of objects that are
sampled.

\begin{remark}[Router independence]\label{rem:router}
The router is the only approximate subroutine in \solo{}, and the
index does not inherit its error. On Deep-100M we replaced query
routing with \emph{exact} top-$\ks$ over the two-million-point sample
(chunked exhaustive search) and re-served every frontier
configuration: the served recalls are identical to four decimals at
$b\in\{16,32\}$ and $\ks\in\{256,384,512\}$ (the routed term sets
themselves overlap at $0.9992$--$0.9999$). Whatever recall \solo{}
attains is therefore a property of the assignment geometry, not of
the router --- which we report with \graft{} \cite{graft2026}, and
which any high-recall index over $\alpha n$ points may supply. The
same holds on the build side by the insurance argument of
\S\ref{sec:insurance}: a signature that misses one sample point
forfeits one of $b$ replicas, and coverage depends on the union of
$b$.
\end{remark}

\subsection{Scan-only serving and the recall identity}\label{sec:identity}

\paragraph{What is stored, and what is decided at serve time.}
The build leaves exactly one record per object: $\mathrm{sig}(x)$, the
$\kb$ sample points nearest to $x$, ranked by distance. Nothing else
about $x$ is consulted by the serving path except $x$ itself, when its
distance to a query is evaluated. Two dials are then chosen when the
index is served, and changing either touches no stored data. The
replication $b\le\kb$ decides how many of the $\kb$ entries are
\emph{active}: object $x$ is in the posting list $L_s(b)$ of sample
point $s$ iff $s$ is among the first $b$ entries of $\mathrm{sig}(x)$.
Since the entries are ranked, $L_s(b)\subseteq L_s(b')$ for $b\le b'$:
the lists for every $b$ are prefixes of the lists for $\kb$, and one
build materializes all of them. The dial $\ks$ decides how many lists a
query reads: $q$ is routed to its $\ks$ nearest sample points,
$R_{\ks}(q)$, by the same router search the build used for the
objects.

\paragraph{Serving, step by step.}
Given $q$, $b$, $\ks$ and $k$: (1) route $q$ to $R_{\ks}(q)$; (2) form
the candidate set $V(q)=\bigcup_{s\in R_{\ks}(q)}L_s(b)$, with each
object counted once; (3) evaluate $d(q,x)$ for every $x\in V(q)$;
(4) return the $k$ objects of $V(q)$ with the smallest distance. The
cost of a query is $|V(q)|$ distance evaluations plus the routing
search, and $|V(q)|\le\ks\cdot\max_s|L_s(b)|$, with mean list length
$b/\alpha$. A list that outgrows the bound $\tau$ is not scanned
whole but decomposed by the same rule (\S\ref{sec:recursion}); at
every depth, the lists that are finally read are read in full.

\paragraph{What is approximate, and what is not.}
Steps (2)--(4) are exact: step (3) evaluates the true distance on every
candidate, and step (4) is an exact selection over those values. Step
(1) is approximate only in which sample points it returns, and
Remark~\ref{rem:router} measures that approximation at $\le10^{-4}$ of
recall. Compare this with a navigable graph, where recall depends on
the entry point, the beam width, the order in which the beam visits
candidates, the visited-set bookkeeping and the pruning of the
neighborhood lists at build time --- none of which is a property of
stored data, and all of which interact. \solo{}'s serving path has one
decision that can lose a neighbor, and it is made before any distance
to a database object is evaluated: which lists to read.

\paragraph{The membership event.}
Fix a query $q$ and one of its true $k$ nearest neighbors $y$. Write
$T_b(y)$ for the first $b$ entries of $\mathrm{sig}(y)$, the sample
points $y$ is posted to. Then $y\in V(q)$ iff some $s\in R_{\ks}(q)$
has $y\in L_s(b)$, i.e.\ iff $T_b(y)\cap R_{\ks}(q)\neq\emptyset$. If
the intersection is non-empty, $y$ is a candidate, its exact distance
is evaluated, and since it is a true $k$-nearest neighbor no more than
$k-1$ candidates can rank ahead of it, so it is returned (up to ties
at rank $k$). If the intersection is empty, $y$ is not a candidate, is
never evaluated, and no later step can return an object it has not
evaluated. There is no third case. Recall@$k$ is the fraction of
$(q,y)$ pairs for which the intersection is non-empty, averaged over
the query distribution; we call that fraction the \emph{coverage}
$\mathrm{cov}(b,\ks)$, and the proposition below states the equality.
Both $T_b(y)$ and $R_{\ks}(q)$ are short lists of integers, one stored
and one produced by routing; whether they intersect is decided by
comparing integers, not by running a search over the database.

\paragraph{How the certificate is measured in this paper.}
Every certified figure reported here was checked by serving. The
protocol is the same at every scale: a certification sample of $m$
queries ($m=10^4$ at $10^6$--$10^8$, $500$ at $10^9$) with exact
$k$-nearest neighbors computed by brute force; coverage computed from
the stored signatures and the routed term sets by
Algorithm~\ref{alg:certify}, with exact routing over $S$ where the
router's contribution is to be separated and with the deployed router
where the deployed figure is wanted; then the same queries served at
the same $(b,\ks)$ through the kernels of \S\ref{sec:kernels}, with
recall measured against the same ground truth. The residual
$\mathrm{cov}-\mathrm{recall}$ is the certificate's error, and the
paper reports it rather than absorbing it: ties at rank $k$ are
counted in the ground truth itself; the loss of a quantized screen
that keeps only $R$ survivors is measured separately and, at the
serving defaults, is zero (\S\ref{sec:sq}) or stated
(\S\ref{sec:deep1b}); and at $10^9$ the certificate was computed
(stage 1 of \S\ref{sec:deep1b}) before any serving kernel existed for
that scale, so the served numbers of stages 2--4 landed inside a
surface written down in advance. The precision of the certificate as a
statistic --- a mean over $m$ queries whose $k$ neighbors are not
independent --- is measured by a clustered bootstrap over queries
(Table~\ref{tab:certboot}).

\begin{proposition}[Recall identity]\label{obs:identity}
Write $T_b(x)\subseteq S$ for the $b$ nearest sample points of an
object $x$ and $R_{\ks}(q)\subseteq S$ for the $\ks$ nearest sample
points of a query $q$, and let
\[
\mathrm{cov}(b,\ks)\;=\;\Pr\bigl[\,T_b(y)\cap R_{\ks}(q)\neq\emptyset\,\bigr]
\]
for $y$ a true $k$-NN of $q$: the probability over the query
distribution that a true neighbor and the query share at least one
term. Under scan-only serving,
$\mathrm{recall@}k=\mathrm{cov}(b,\ks)$, exactly, up to ties at
rank $k$.
\end{proposition}

\begin{proof}
Let $V(q)=\bigcup_{s\in R_{\ks}(q)}L_s$ be the scanned set, where $L_s$
is the posting list of $s$. By construction $y\in L_s$ iff
$s\in T_b(y)$, so $y\in V(q)$ iff $T_b(y)\cap R_{\ks}(q)\neq\emptyset$.
Serving evaluates $d(q,\cdot)$ on every element of $V(q)$ --- the true
distance, no bound, no early exit --- and returns the exact top-$k$ of
$V(q)$. Hence a true $k$-NN $y$ is returned iff $y\in V(q)$, i.e.\ iff
$y$ and $q$ share a term; no other event can lose it and no other
event can recover it. Taking expectations over queries and over the
$k$ true neighbors gives the identity. The only slack is at the
boundary: if $d(q,y)$ equals the $k$-th distance, another object at
the same distance may be returned in its place, which is why the
identity is stated up to ties.
\end{proof}

The proof is short because the design is: a scan-only path has exactly
one failure mode, and it is a membership question about stored data
rather than a question about search dynamics. Every quantity a graph
must be trusted on --- entry points, beam width, pruning, the order in
which candidates are met --- is absent, so there is nothing left for
recall to depend on.

\begin{figure}[t]
\centering
\resizebox{0.96\textwidth}{!}{
\begin{tikzpicture}[
  font=\small,
  smp/.style={circle, draw=black!45, fill=black!15, inner sep=0pt, minimum size=4.6pt},
  rq/.style={circle, draw=sloblue, fill=sloblue, inner sep=0pt, minimum size=4.8pt},
  ty/.style={circle, draw=sloorange, fill=sloorange, inner sep=0pt, minimum size=4.8pt},
  both/.style={circle, draw=sloorange, fill=sloblue, inner sep=0pt, minimum size=6.6pt, line width=1.3pt},
  note/.style={font=\scriptsize, color=black!70},
  quiet/.style={font=\scriptsize\itshape, color=black!55},
  ar/.style={-{Latex[length=4pt,width=3pt]}, draw=black!55, line width=0.5pt}
]
\begin{scope}
  \foreach \p in {(0.25,0.30),(3.75,0.45),(3.95,2.35),(0.35,2.30),(3.55,1.45),(1.05,0.25),(2.95,2.75)}
     \node[smp] at \p {};
  \draw[sloblue, densely dashed, line width=0.5pt] (1.10,1.50) circle (0.80);
  \node[star, star points=5, star point ratio=2.2, draw=sloblue, fill=sloblue,
        inner sep=1.4pt] at (1.10,1.50) {};
  \node[note, anchor=east] at (0.96,1.50) {$q$};
  \foreach \p in {(0.55,1.05),(0.65,1.95),(1.45,2.05),(1.50,0.95)} \node[rq] at \p {};
  \node[quiet, color=sloblue, anchor=west, align=left] at (0.00,2.85)
       {$R_{\ks}(q)$: the $\ks$ terms $q$ is routed to};
  \draw[sloorange, densely dashed, line width=0.5pt] (2.25,1.30) circle (0.70);
  \node[circle, draw=black, fill=white, inner sep=0pt, minimum size=5.2pt, line width=0.7pt] at (2.25,1.30) {};
  \node[note, anchor=west] at (3.02,1.30) {$y$, a true neighbor};
  \draw[black!55, line width=0.4pt] (2.33,1.30) -- (2.98,1.30);
  \foreach \p in {(2.65,0.85),(2.75,1.75)} \node[ty] at \p {};
  \node[quiet, color=sloorange, anchor=west, align=left] at (1.60,0.20)
       {$T_b(y)$: the $b$ terms $y$ is posted to};
  \node[both] at (1.70,1.70) {};
  \draw[ar] (3.20,2.20) -- (1.82,1.80);
  \node[note, anchor=west, align=left] at (3.25,2.22)
       {shared term $s$:\\[-2pt] $y\in L_s$, and $L_s$ is scanned};
  \node[note, align=left, anchor=north west] at (0.00,-0.15)
       {$y$ is returned $\iff T_b(y)\cap R_{\ks}(q)\neq\emptyset$\\[1pt]
        {\itshape\color{black!55}one event, decided by stored integers}};
\end{scope}
\begin{scope}[shift={(7.2,0)}]
  \node[note, anchor=west, align=left] at (-0.30,2.95)
       {certification sample: $m$ queries, their $k$ exact neighbors};
  \foreach \j in {1,...,8} {
    \node[note, anchor=south] at (0.35+\j*0.32-0.32,2.52) {\tiny $y_{\j}$};
  }
  \foreach \i in {1,...,6} {
    \node[note, anchor=east] at (0.00,2.30-\i*0.30+0.30) {$q_{\i}$};
  }
  \foreach \i/\row in {0/{1,1,1,1,1,1,1,1},1/{1,1,1,0,1,1,1,1},2/{1,1,1,1,1,1,1,1},
                       3/{1,1,1,1,1,1,0,1},4/{1,1,1,1,1,1,1,1},5/{1,0,1,1,1,1,1,1}} {
    \foreach \v [count=\j from 0] in \row {
      \ifnum\v=1
        \node[circle, fill=sloblue!70, inner sep=0pt, minimum size=4.8pt]
             at (0.35+\j*0.32,2.30-\i*0.30) {};
      \else
        \node[circle, draw=sloorange, line width=0.9pt, inner sep=0pt, minimum size=4.8pt]
             at (0.35+\j*0.32,2.30-\i*0.30) {};
      \fi
    }
  }
  \node[circle, fill=sloblue!70, inner sep=0pt, minimum size=4.8pt] at (3.25,2.30) {};
  \node[note, anchor=west] at (3.38,2.30) {$y$ and $q$ share a term};
  \node[circle, draw=sloorange, line width=0.9pt, inner sep=0pt, minimum size=4.8pt] at (3.25,2.00) {};
  \node[note, anchor=west] at (3.38,2.00) {they do not: $y$ is missed};
  \node[note, anchor=north west, align=left] at (-0.30,0.45)
       {$\mathrm{cov}(b,\ks)=\dfrac{\#\text{filled}}{m\,k}=\mathrm{recall@}k$\\[4pt]
        {\itshape\color{black!55}the brute-force pass that finds the $y$'s is paid once;}\\[-1pt]
        {\itshape\color{black!55}every $(b,\ks)$ after that is a prefix intersection}};
\end{scope}
\end{tikzpicture}}
\caption{The certificate, drawn. Left: one query $q$ and one of its
true neighbors $y$. The $\ks$ terms $q$ is routed to and the $b$ terms
$y$ is posted to are both sets of stored integers; $y$ is returned if
and only if the two sets intersect, because the list of a shared term
is scanned in full and $y$ is in it. Right: the same test over a
certification sample of $m$ queries and their $k$ exact neighbors. The
fraction of pairs that share a term is $\mathrm{cov}(b,\ks)$, and by
Proposition~\ref{obs:identity} it is the recall the index will serve.
The ground-truth pass that supplies the $y$'s is paid once; every
$(b,\ks)$ afterwards is a prefix intersection.}
\label{fig:cert}
\end{figure}
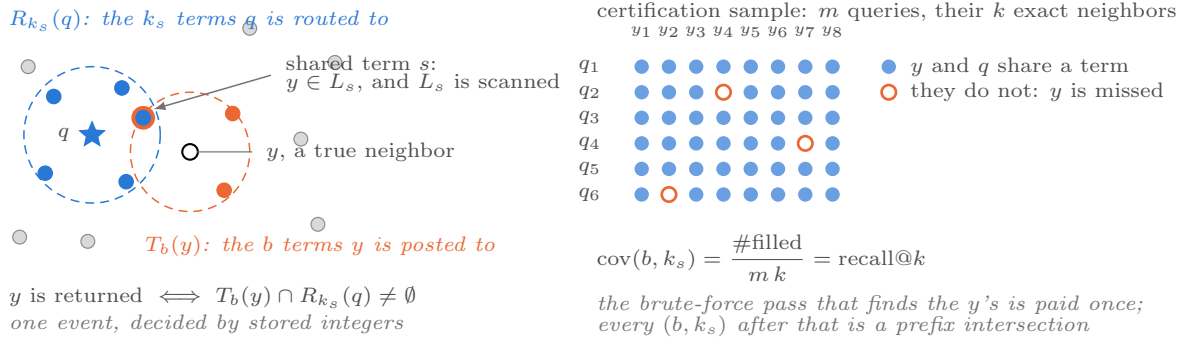

\begin{algorithm}[t]
\caption{\textsc{Certify}: the recall surface from stored integers.
$Q$ is a sample of $m$ queries from the workload; $\mathrm{sig}(\cdot)$
are the stored signatures of Algorithm~\ref{alg:solo}; the
brute-force pass on line~2 is the only step whose cost grows with
$n$, and it is paid once.}
\label{alg:certify}
\begin{algorithmic}[1]
\Procedure{Certify}{$Q,k,\kb,\ks^{\max}$}
  \ForAll{$q\in Q$} $N_k(q)\gets$ exact $k$-NN of $q$ in $X$
     \Comment{one brute-force pass; $m\cdot n$ distances}
  \EndFor
  \ForAll{$q\in Q$} $R(q)\gets\textsc{Route}(q,S,\ks^{\max})$, ranked
     \Comment{or exact top-$\ks^{\max}$ over $S$, to exclude the router}
  \EndFor
  \ForAll{$b\le\kb$, $\ks\le\ks^{\max}$}
    \State $\mathrm{cov}(b,\ks)\gets\dfrac{1}{mk}\sum_{q\in Q}\sum_{y\in N_k(q)}
      \bigl[\,\mathrm{sig}(y)[1..b]\cap R(q)[1..\ks]\neq\emptyset\,\bigr]$
  \EndFor
  \State \Return the surface $\mathrm{cov}(\cdot,\cdot)$
    \Comment{$=\mathrm{recall@}k$ at every $(b,\ks)$, Proposition~\ref{obs:identity}}
\EndProcedure
\end{algorithmic}
\end{algorithm}

\paragraph{Computing a certificate.}
Figure~\ref{fig:cert} is the membership event as a picture, and
Algorithm~\ref{alg:certify} is the procedure. An operator draws $m$
queries from their own workload, computes their exact $k$ nearest
neighbors once by brute force, routes each query once, and then
evaluates the indicator of Figure~\ref{fig:cert} for every pair and
every $(b,\ks)$ of interest. After the brute-force pass the work is
$m\,k\,\kb\,\ks^{\max}$ integer comparisons --- seconds on one core at
$m=10^4$ --- and it covers the entire operating surface, including
configurations that will never be deployed and values of $b$ that
have never been materialized as lists. For the recursive index the
indicator composes: inside a decomposed list the same test is applied
with the local signatures and the local routing, and $y$ is returned
iff it is covered at every level on its way to a leaf; the measured
depth-2 index lands on that composed prediction to $2\times10^{-4}$
(\S\ref{sec:recursion}).

\paragraph{What the certificate guarantees, and what it does not.}
\begin{enumerate}
\item \emph{It is an identity, not a bound.} The certified figure is
  the expected recall@$k$ the index serves on the query distribution
  the sample was drawn from --- not a lower bound on it, not a proxy
  correlated with it. It is an expectation over queries: a single query
  may be covered or not, and the per-query question is the subject of
  \S\ref{sec:oos}.
\item \emph{It is tied to a query distribution.} $Q$ must come from
  the workload whose recall is being stated. A different workload needs
  its own certificate, which costs one brute-force pass and no serving.
  The index also carries a per-query provision for queries from outside
  the certified distribution (\S\ref{sec:oos}).
\item \emph{It does not depend on the router.} Coverage is a property
  of the stored signatures and the routed term sets; replacing the
  router by exact search over $S$ moves the served recall by
  $\le10^{-4}$ (Remark~\ref{rem:router}), and line~3 of
  Algorithm~\ref{alg:certify} can use exact routing to state the
  router's contribution separately (\S\ref{sec:msmarco} does).
\item \emph{It does not depend on the kernel, provided the kernel is
  exact over the scanned set.} The quantized screens of
  \S\ref{sec:sq} re-rank $R$ survivors with the true distance and are
  measured bit-identical to the full-precision scan at the serving
  defaults; a kernel that dropped a candidate before evaluating
  $d(q,\cdot)$ would break the identity, which is why every kernel in
  the paper is checked against the full scan.
\item \emph{It holds up to ties at rank $k$, and for the $k$ of the
  ground-truth pass.} A pass at $k_{\max}$ certifies every
  $k\le k_{\max}$.
\item \emph{It survives updates in its integer part.} An insert adds
  one signature and a delete removes one; the term sets of every other
  object are untouched. The ground-truth part is a property of the
  database and is refreshed with it when a re-stated figure is wanted
  --- one brute-force pass, no serving.
\item \emph{Its precision is that of a mean over $m$ queries, and it
  is measured.} Pairs within a query are not independent, so the
  binomial standard error over $mk$ pairs understates the uncertainty.
  Table~\ref{tab:certboot} resamples queries
  (\num{certboot/sift/nboot} clustered bootstrap replicates, exact
  routing, no serving) on three corpora: at coverage $\ge0.99$ the
  $95\%$ interval has half-width at most
  \num{certboot/halfwidth-max-99} and the clustered standard error
  exceeds the binomial one by at most
  $\num{certboot/ratio-max-99}\times$; the ratio grows with the miss
  rate, reaching $\num{certboot/ratio-max}\times$ at GloVe's lowest
  cell (coverage \num{certboot/glove/min-cov}). The agreements reported
  below should be read at that scale.
\end{enumerate}

\begin{table}[t]
\centering\small
\begin{tabular}{lcccccc}
\toprule
dataset & $(b,\ks)$ & cov & 95\% CI & SE, clustered & SE, binomial & ratio (cell; grid range)\\
\midrule
SIFT-1M & $(16,32)$ & 0.9972 & [0.9968, 0.9976] & 0.00019 & 0.00017 & 1.15; 0.98--1.35\\
GloVe-1.2M & $(32,128)$ & 0.9905 & [0.9898, 0.9912] & 0.00038 & 0.00031 & 1.23; 0.98--2.06\\
Deep-10M & $(4,256)$ & 0.9944 & [0.9938, 0.9949] & 0.00030 & 0.00024 & 1.25; 0.96--1.51\\
\bottomrule
\end{tabular}

\caption{Precision of the certificate. For each corpus, the cell of the
$(b,\ks)$ grid of Table~\ref{tab:surface} (extended to $\ks\le512$)
whose coverage first exceeds $0.99$: its $95\%$ interval from a
clustered bootstrap over the \num{certboot/sift/m} certification
queries, the clustered standard error, the binomial standard error over
$mk$ pairs, and their ratio at that cell and across the whole grid.
Coverage is computed from stored signatures with exact routing over
$S$; nothing is served.}
\label{tab:certboot}
\end{table}

\paragraph{Reading a certificate.}
Table~\ref{tab:surface} is one. Every entry is a
$\mathrm{cov}(b,\ks)$ computed by Algorithm~\ref{alg:certify} from
stored signatures; no entry was obtained by serving. An operator who
needs recall $0.99$ on GloVe reads along the row for the $\ks$ they
can afford, finds the first $b$ that meets it, and serves that
configuration knowing its recall before the first query arrives. The
contrast with a navigable graph is not that the graph cannot be
measured --- it can, against the same ground truth --- but that each
cell of its tuning table must be obtained by serving that cell, a
cell never served has no number, the numbers are specific to the
traversal parameters that produced them, and a rebuilt index must be
re-served. For \solo{} the table is a property of the stored data.

Empirically the identity holds to four decimals on GloVe and Deep; on
SIFT the residual ($\le 9\times10^{-4}$) is fully attributed to ties
(integer coordinates; zero membership misses; 28/2000 queries have
tied 10-th distances within the exact ground truth itself). The
sharpest test is the hardest corpus at the largest scale: on MS MARCO
Web Search at $10^{8}$ --- $768$-dimensional text embeddings under
inner product, with queries drawn from a different distribution than
the passages --- the certificate computed from the stored signatures
alone matches the served recall to
$\num{cert/msmarco100m-max-err}$ at all
$\num{cert/msmarco100m-points}$ operating points, including those where
the served recall is poor ($\num{cert/msmarco100m-best-recall}$ at
best, \S\ref{sec:msmarco}). An index that predicts its own recall
before serving predicts it when the answer is unwelcome too, which is
the property that makes it worth deploying: routing exactly rather than
with the router costs a further
$\num{cert/msmarco100m-router-cost}$ of coverage, so the two sources of
loss are separately visible in the certificate.

\paragraph{The certificate, and what it costs.}
We call a recall figure \emph{certified} when it is computed from the
stored index rather than observed by running the system, and
Proposition~\ref{obs:identity} is what makes such a figure available:
it turns an empirical question into an arithmetic one
(Algorithm~\ref{alg:certify}). Two
consequences distinguish this from ordinary offline tuning. First, the
certified quantity \emph{is} the served recall rather than a proxy for
it --- Proposition~\ref{obs:identity} is an identity, not a bound or a
correlation. Second, the certificate is transferable: an operator
computes it on their own query distribution, which is the distribution
their recall will actually be drawn from, and re-certifies after a
workload shift without touching the index. For a navigable graph no
analogous object exists at any price: recall is a property of the
traversal, so the only way to learn it is to run it.

\paragraph{Out of sample.}
The certificate is an expectation over the query distribution it was
computed on. A query drawn from outside that distribution needs a
provision, and the index carries one --- the $\ks$ routed distances are
order statistics of a uniform sample of the query's own distance
distribution, so routing has already paid for a probe of the query's
neighborhood. Section~\ref{sec:oos} develops that provision and
measures it across seven corpora and thirteen query sets, including a
cross-modal shift where it fails.

\subsection{One rule: recursive decomposition}\label{sec:recursion}
Hubness skews list sizes (max/mean ${\sim}48\times$ on GloVe at every
$b$). \solo{} does not fight skew at build time; it applies its own
rule to any list larger than a bound $\tau$: sample the list at
$\alpha'$, post members to their $\blocal$ nearest \emph{local} sample
points, and route queries that enter the list to $\kslocal$ local
cells. Leaves --- lists below the bound, at any depth --- are always
scanned. Recall composes multiplicatively across levels, but so does
the replication insurance (\S\ref{sec:insurance}); measured inside
GloVe's oversized lists, local coverage reaches $0.9969$
($\blocal{=}8,\kslocal{=}32$) and $0.9994$ ($\kslocal{=}64$), and the
depth-2 index lands on its composed prediction to $2\times10^{-4}$.

\subsection{Quantized leaves}\label{sec:sq}
The leaf scan is a streaming pass over posting-ordered candidates ---
the computation SIMD integer arithmetic is best at. \solo{} uses the
simplest quantizer that could work: a global affine 8-bit grid (two
constants, one pass over the data, no training), or --- halving the
bytes streamed --- a per-dimension affine 4-bit grid ($2d$ constants),
packed two codes per byte. Serving screens with integer squared
distances, keeps the $R$ best \emph{distinct} candidates, and re-ranks
them at full precision; the returned ranking is exact over the
screened set. At every operating point measured --- three datasets,
$10^6$ to $10^8$ --- 8-bit codes with $R{=}40$ already reproduce the
full-precision results exactly, and we serve with $R{=}100$; the
coarser 4-bit screen needs $R{=}200$ to match on GloVe (SIFT is
unchanged at $R{=}100$), and all SQ4 rows in this paper use $R{=}200$.
For non-vector metric spaces the same architecture runs with plain
distances; quantization is an accelerator of the verifier, never part
of the core.

\paragraph{How far down does the stream go?} The screen is a sequential
pass, so its cost is bytes, and the question is how coarse a code still
recovers the pool's own ceiling after exact refine. We measured a third
rung below SQ4 --- a 2-bit ternary code (keep the $d/2$ largest $|u_i|$ as
$\pm1$, zero the rest), scored against a \emph{float} query, which leaves
the op count of an SQ8 dot product and divides its bytes by
four~\cite{connor2026evp}. Inside a routed pool at $\ks{=}32$, the pass set
needed to recover $99.9\%$ of the ceiling is $R{=}\num{rq/glove-screenK-sq8}$
for SQ8 but $\num{rq/wiki-screenK-evpft}$ at $d{=}1024$ and
$\num{rq/glove-screenK-evpft}$ at $d{=}100$: two bits genuinely rank worse,
by one to two orders of magnitude in $R$. It is still the cheaper design,
because the pass set is two orders of magnitude smaller than the pool ---
counting bytes moved per query, a 2-bit screen feeding an SQ8 re-rank
feeding the exact refine moves ${\sim}3.6\times$ less than an SQ8 screen at
$d{=}1024$ and ${\sim}2.9\times$ less at $d{=}100$ --- projected from the
measured pool and pass-set sizes with the gather penalty taken from this
paper's own SQ8-gather-versus-SQ4-stream measurement, not measured end to
end. The middle rung is what makes it work: it turns a wide
\emph{exact} gather into a wide SQ8 gather plus a narrow exact one. Two
practical corollaries: quantizing the query as well (a fully bitwise
product) is never worth it --- the query is one vector and keeping it float
buys one to seven points of coverage --- and the object id should leave the
code stream for a parallel array, since the screen never reads an id and
only the survivors do, which is the difference between
$\num{rq/page-d96-evp}$ and $\num{rq/page-d96-evpsplit}$ records per
$4$\,KB page at $d{=}96$, against $\num{rq/page-d96-sq8}$ for SQ8. These are $10^6$-scale laptop measurements
(\S\ref{sec:future}); the serving stack in this paper is SQ8 and SQ4
throughout.

The screen/scorer distinction is load-bearing, and we
measured its other side: SPANN built and served \emph{natively} in
Int8 on the identical 8-bit grid --- quantized distances as the final
scorer --- saturates at $\num{sat/spann-int8}$ on Deep-100M against float ground
truth ($58.2$\,h build, same machine), four points below its own
float build, while the same grid used as a screen with exact refine
is lossless. Quantization must screen, never score.

\subsection{Serving kernels and the memory model}\label{sec:kernels}
Serving state: router, postings, codes, vectors. Three facts matter at
$n\ge10^8$.

\emph{Visited state must be bits.} The scan deduplicates across
$\ks$ lists. An epoch-stamped 32-bit visited array costs
4\,B/object/thread --- 25.6\,GB at $10^8$ with 64 threads: invisible in
RAM benchmarks, fatal under a memory cap. A per-thread bitmap with
touched-clearing costs 12.5\,MB and is \emph{faster} (cache-resident;
$+15\%$ throughput). Deduplication must occur \emph{before} the
screen's $R$-heap: the most duplicated candidates are exactly the
nearest ones, and an undeduplicated heap fills with copies (we measured
the failure: recall $0.9982\to0.8394$, worsening with $\ks$).

\emph{The scan must stream, not gather.} With codes stored per
object, scanning a list gathers one cache line per posting from a
random address; with the codes \emph{duplicated in posting order}
(list-major: the record of list $\ell$'s $i$-th member sits at
position $i$ of $\ell$'s code stream) each of the $\ks$ lists is one
sequential segment, and the same layout is what the disk mode reads
(below). At $10^6$ on one laptop core the stream costs ${\sim}1.7$\,ns
per record at either 8 or 4 bits (instruction-bound: the 4-bit form
halves bytes, not instructions) plus $0.5$--$1.5\,\mu$s per segment
of prefetcher ramp-up, against $9$--$25$\,ns per gathered record; at all
ten cores the stream is bandwidth-bound and the 4-bit codes buy a
further $1.4\times$. The layout is written once, at $b{=}\kb$, with
each list ordered by the \emph{rank} its members assign to it --- the
objects for which $\ell$ is the nearest sample point first, then those
at rank 2, and so on. Two dials then cost nothing at build time: the
served $b$ is a prefix of every list, and the prefix may \emph{taper}
with routing position, $b(t)=\lceil b_0(1-t/\ks)^p\rceil$ for the
$t$-th routed list, $p{=}0$ being the rectangle of
Table~\ref{tab:solo-points}. A tapered frontier reads fewer records
than any rectangle of equal recall (on SIFT-1M, $0.9991$ from 54\% of
the records of the $(16,128)$ rectangle) --- the rank-$t$ lists far
from the query contribute mostly rank-1 members, which the frontier
keeps. The router is served at $ef{=}\ks$: raising it to $4\ks$ moves
no recall by more than $10^{-3}$ and costs four times the routing
time.

\emph{Load must not rebuild.} The router serializes to a flat,
position-independent, mmap-able image; loading is
$43$--$90\,\mu$s at every scale measured, versus seconds-to-minutes of
reconstruction, with bitwise-identical search. Under a memory cap the
mapped router is file-backed and reclaimable; a rebuilt router is
anonymous memory and is not.

\emph{Disk wants layout, not paging.} Serving the same index by
letting vector reads fault through \texttt{mmap} reads 6\,TB in 80
minutes under a 4\,GB cap (a 4\,KB page per 384\,B row, evicted before
reuse): an ${\sim}85\times$ amplification. \solo{}'s disk mode instead
\emph{duplicates codes into the lists} (SPANN's trade --- disk is the
cheap resource): leaf blocks of $[\mathrm{id}\,|\,\mathrm{code}]$
records in posting order, read sequentially into bounded per-thread
buffers, screened, then refined by ${\sim}R$ row reads from the single
full-precision file. At $(b,\ks)=(16,256)$ this is 37\,MB sequential
$+\;0.4$\,MB random per query, all explicit.

\subsection{The codes as a proximity sketch}\label{sec:sketch}
The ranked signatures serve a second role beyond routing: for objects
\emph{already indexed}, company overlap is a calibrated proximity
estimator computed by integer set intersection --- no vector access.
Measured on SIFT-1M (800 anchors), mean overlap at $b{=}64$ falls
smoothly from $38.8$ shared terms at neighbor rank 1 through $23.4$ at
rank 100 to $16.3$ at rank 1000, against a random-pair floor of
$0.35$ --- a dynamic range of ${\sim}110{:}1$, monotone across three
decades of rank; at $b{=}16$ a $\ge1$-shared-term test captures
$95\%$ of rank-100 pairs at a $1.8\%$ false-pair rate. Three uses
follow. (i) \emph{Result-set operations in the disk regime}:
re-ranking, diversification, or clustering of a returned candidate
set by code overlap costs zero reads where vector fetches are the
expensive resource (\S\ref{sec:capped}). (ii) \emph{Similarity
self-joins}: shared-term thresholds are a data-defined banding scheme
over the inverted file the index already maintains.
(iii) \emph{Standing queries}: a subscription encoded once matches
arriving objects' codes --- which their own inserts compute anyway ---
by intersection. One boundary: for a \emph{fresh} query,
encoding \emph{is} a search, so code comparison cannot accelerate
what produces it --- data-adapted resolution costs a search to encode,
the opposite end of the trade LSH \cite{indyk1998lsh} occupies; indeed \misi{}'s vote was
exactly code-overlap used as a ranking, and \S\ref{sec:compare}
shows scanning dominates it.

\subsection{Updates and determinism}\label{sec:dynamics}
Insert: one router search, $b$ list appends. Delete: removal from $b$
lists --- \emph{exact}, no tombstones, no graph repair, no rebalancing
protocol; no other object's representation changes. Every build and
every query is bitwise deterministic for a fixed seed at any thread
count; recall is machine-independent (our d0 and laptop runs agree bit
for bit, as they must).

\section{Measured laws}\label{sec:laws}
Three regularities of the index, each measured across datasets and
scales: coverage depends on the touched posting mass and not on how it
is split between $b$ and $\ks$; the work a recall target costs grows
sublinearly in $n$; and a covered neighbor is found in several touched
lists at once, which is what makes every lossy step downstream of
routing forgiving.

\subsection{The equal-work law and the regime spectrum}
Coverage collapses onto contours of the touched posting mass:
$\mathrm{cov}(b,\ks)\approx f_X(b\cdot\ks)$
(Figure~\ref{fig:equalwork}, Table~\ref{tab:surface}). On GloVe at ${\sim}140$k scanned
candidates, four different factorizations
$(b,\ks)\in\{(4,256),(8,128),(16,64),(64,16)\}$ cover
$\num{law/glove-contour1024-lo}$--$\num{law/glove-contour1024-hi}$;
the law reappears at Deep-100M
($(16,256)\approx(32,128)$ at $0.998$, equal throughput). Inside a
contour the split is free, so one chooses small $b$ (storage, build)
and buys recall with $\ks$ (routing over $\alpha n$ points is cheap).
The \emph{level} of $f_X$ is a one-scalar regime signature ---
$b\cdot\ks$ for coverage $0.99$: SIFT ${\sim}512$, Deep-10M
${\sim}1$--$2$k, GloVe ${\gtrsim}8$k --- the concentrated/spread
dichotomy of metric spaces, measurable from the index itself.

\begin{figure}[t]
\centering
\includegraphics[width=0.78\textwidth]{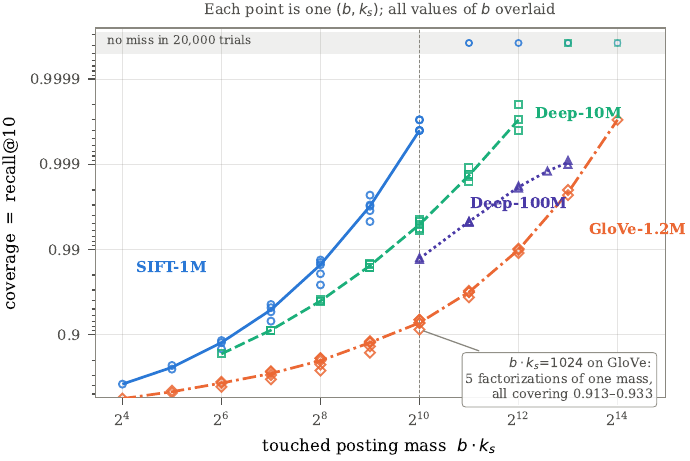}
\caption{The equal-work law. Every measured $(b,\ks)$ pair is plotted
against the posting mass it touches; the values of $b$ (1--64 on the
first three datasets, 8--32 on Deep-100M) are overlaid, not separated.
That they fall on one curve per dataset \emph{is} the law: within a
contour the factorization is free. The curves' horizontal offsets are
the regime spectrum --- SIFT reaches coverage $0.99$ at $b\cdot\ks$
${\sim}512$, GloVe needs ${\gtrsim}8$k. Deep-100M points are served
recalls, equal to coverage by Proposition~\ref{obs:identity}; the
others are coverage computed from signatures.}
\label{fig:equalwork}
\end{figure}

\begin{table}[t]
\centering\small
\begin{tabular}{lccccc}
\toprule
$\mathrm{cov}(b,\ks)$ & $b{=}4$ & $b{=}8$ & $b{=}16$ & $b{=}32$ & $b{=}64$\\
\midrule
SIFT, $\ks{=}64$ & 0.9849 & 0.9977 & 0.9997 & 1.0000 & 1.0000\\
SIFT, $\ks{=}128$ & 0.9965 & 0.9996 & 1.0000 & 1.0000 & 1.0000\\
GloVe, $\ks{=}128$ & 0.8626 & 0.9272 & 0.9683 & 0.9902 & 0.9980\\
GloVe, $\ks{=}256$ & 0.9126 & 0.9635 & 0.9890 & 0.9977 & 0.9997\\
\bottomrule
\end{tabular}

\caption{Coverage under bounded assignment ($\alpha{=}2\%$, seed 1,
2{,}000 queries, true 10-NN). By Proposition~\ref{obs:identity} these
\emph{are} served recalls.}
\label{tab:surface}
\end{table}

\subsection{Scale}
Work per query at fixed recall grows sublinearly in $n$, as
$n^{0.4\text{--}0.6}$ over the four measured decades, and the scanned
fraction of the database falls throughout. At fixed $(b,\ks)$ the level
of $f$ drifts down with $n$ ($0.9994\to0.9982$ at $(16,256)$ from
Deep-10M to Deep-100M); the drift is geometry, not router error
(Remark~\ref{rem:router}): finer cells at $|S|=2$M dilute the true
neighbor's membership. The dial compensates. At target $0.998$, each
scale interpolated on its own (evaluations, recall) frontier, a query
costs $\num{scale/evals-1e6-998}\to\num{scale/evals-1e7-998}\to
\num{scale/evals-1e8-998}\to\num{scale/evals-1e9-998}$ evaluations at
$n=10^{6}\dots10^{9}$ --- exponents $\num{scale/exp-lo-998}$,
$\num{scale/exp-hi-998}$ and $\num{scale/exp-top-998}$ per decade, the
last the steepest --- while the scanned fraction falls
$\num{scale/frac-1e6-998}\%\to\num{scale/frac-1e7-998}\%\to
\num{scale/frac-1e8-998}\%\to\num{scale/frac-1e9-998}\%$, so at $10^{9}$
a query at recall $0.998$ touches one part in $1400$ of the collection.
Two qualifications bound the series. The $10^{9}$ point is certified
rather than served --- coverage read from the stored signatures
(\S\ref{sec:deep1b}), which Proposition~\ref{obs:identity} equates with
recall; the served points at $10^{9}$ stop at
$\num{deep1b/floor-recall}$, below this target, and there the identity
holds to $\num{deep1b/s2/identity-err}$ --- and its interpolation
bracket is wide. And because the exponent rises with $n$, a fitted
exponent extrapolates optimistically past $10^{9}$; the paper makes no
claim there (\S\ref{sec:ceiling}). At the looser target $0.995$ the
series is
$\num{scale/evals-1e6-995}\to\num{scale/evals-1e7-995}\to\dots\to
\num{scale/evals-1e9-995}$; the $10^{8}$ sweep never dipped below recall
$0.998$, so its $0.995$ entry is an upper bound and no exponent is
fitted through it.

\subsection{Replication is insurance}\label{sec:insurance}
Among covered GloVe true neighbors at $\ks{=}64$, the count $j$ of
touched lists containing the neighbor has median ${\sim}13$; only
$2.5\%$ have $j{=}1$. A neighbor is lost only if all $j$ appearances
are missed: a local miss rate $\varepsilon$ costs
${\sim}\sum_j p_j\varepsilon^{j}$ globally. This is why the lossless
behavior of the SQ8 screen and the near-ceiling composition of the
recursion are not luck; the same $j$-structure recurses ($j$ median
3--8 inside decomposed lists).

\subsection{Choosing the assignment rule}\label{sec:rule}
Multi-assignment can also be done by distance ratio (``closure'',
SPANN's rule: assign to cells within $(1+\epsilon)$ of the nearest).
At equal mean replication, rank-$b$ dominates closure on both
datasets: SIFT, closure $\epsilon{=}0.1$ ($\bar b{=}7.3$) covers
$0.9960$ at $\ks{=}128$ vs.\ $0.9996$ for rank-$8$; GloVe, closure
$\epsilon{=}0.1$ ($\bar b{=}42.6$) covers $0.9788$ vs.\ $0.9902$ for
rank-$32$. The rank order of the signature is the better boundary
detector at this vocabulary scale, and it is the representation the
build already produces.

\section{The certificate out of sample}\label{sec:oos}
Proposition~\ref{obs:identity} certifies recall on the query
distribution the certificate was computed on. Production query
distributions drift, and an index whose selling point is a recall
figure must say what that figure is worth when they do. This section
develops the per-query provision sketched in \S\ref{sec:identity},
fits it on seven corpora and thirteen query sets, and reports the one
regime in which it fails.

\paragraph{The provision the index carries.}
The certificate is an expectation over the query distribution it was
computed on; a query from outside that distribution needs a provision,
and the index carries one. Because $S$ is a uniform sample of $X$, the
$\ks$ routed distances $d(q,s_{(1)})\le\dots\le d(q,s_{(\ks)})$ are
order statistics of a uniform sample of $q$'s own distance
distribution --- an unbiased probe of the query's neighborhood that
routing has already paid for. The local intrinsic dimension they
estimate (the maximum-likelihood estimator of Levina and
Bickel~\cite{levina2004mle}, used as a query-hardness statistic by
Amsaleg et al.~\cite{amsaleg2015lid}) stratifies coverage. On a
200{,}000-object corpus of text embeddings with in-distribution and
out-of-distribution query sets (experiment \texttt{gate\_o} in the
artifact; exact signatures, so the measurement is of the index rather
than of the inner search), the shift raises the miss rate
\num{ood/miss-ratio-8-64}$\times$ at $(b,\ks)=(8,64)$
(\num{ood/cov-in-8-64} to \num{ood/cov-ood-8-64}) and moves the median
estimated dimension from \num{ood/lid-med-in} to \num{ood/lid-med-ood},
while coverage \emph{conditioned} on the estimate barely moves. A
table of coverage per decile of the statistic, fitted on the
in-distribution certification sample alone, predicts the
out-of-distribution recall to \num{ood/err-table-8-64} at $(8,64)$ and
\num{ood/err-table-8-32} at $(8,32)$, against \num{ood/err-naive-8-64}
and \num{ood/err-naive-8-32} for the unconditioned certificate. The
same table drives the dial: choosing $\ks$ per query as the smallest
value whose row meets a target of $0.99$ meets it on the shifted
queries (\num{ood/adaptive-cov-b8}--\num{ood/adaptive-cov-b32} for
$b\in\{8,16,32\}$) with
\num{ood/miss-reduction-min}--\num{ood/miss-reduction-max}\% fewer
misses than a fixed $\ks$ touching the same posting mass. A query
whose statistic falls outside the calibrated range is served and
reported as uncertified. Underneath the table is the level of $f_X$
resolved per query: the rate $\kappa$ in
$1-\mathrm{cov}\approx\exp(-\kappa\,b\ks/|S|)$, fitted per decile of
the statistic at $(4,8)$, runs monotonically from \num{ood/kappa-lo} to
\num{ood/kappa-hi} and agrees decile by decile across the two query
sets. The triangle inequality also gives a deterministic exactness
radius --- if $2\,d(q,y_{10}) < d(q,s_{(\ks)})-d(q,s_{(b)})$ the served
answer is provably exact in any metric space --- but at these
dimensions it certifies under \num{ood/exact-frac-max}\% of queries,
and we do not rely on it.

\paragraph{The rate is a property of the query; the curve, of the corpus.}
The certificate is per corpus; the statistic behind it need not be.
Writing the touched mass as $b\ks\,n/|S|$ objects,
$1-\mathrm{cov}\approx\exp(-(\kappa/n)\cdot\mathrm{mass})$: $\kappa/n$
is the coverage bought per touched object, comparable across corpora
of different $n$ and $|S|$. Table~\ref{tab:lidlaw} repeats the
measurement on the corpora of \S\ref{sec:exp} and on three larger
ones: LAION-10M (CLIP image embeddings under cosine, the SISAP 2024
query set split in half), and two with a query shift built in ---
PubMed23 (24M abstract embeddings, with an out-of-distribution query
file), and Yandex T2I-10M under $\ell_2$ and under inner product (maximum inner-product search through the
standard augmentation), where the in-distribution queries are images
and the shifted queries are text. On every set the estimate stratifies
coverage (Spearman $\rho$ from $\num{lidlaw/rho-min}$ to
$\num{lidlaw/rho-max}$ at $(4,8)$) and the rate falls monotonically
from the first to the last decile of the statistic, by
\num{lidlaw/spread-min}--\num{lidlaw/spread-max}$\times$. A decile
table fitted on one query set predicts another to within
\num{lidlaw/err-max} on every in-domain transfer --- the other half of
a random split (\num{lidlaw/err-laion-ood} on LAION-10M), the shifted
wiki queries, and the shifted PubMed queries (\num{lidlaw/err-pubmed23-ood} against
\num{lidlaw/err-naive-pubmed23-ood} for the unconditioned
certificate). On PubMed23 the dial meets its $0.99$ target on the
shifted queries
(\num{lidlaw/pubmed23-adaptive-cov-min}--\num{lidlaw/pubmed23-adaptive-cov-max}
for $b\in\{8,16,32\}$) touching
\num{lidlaw/pubmed23-adaptive-mass-saving-min}--\num{lidlaw/pubmed23-adaptive-mass-saving-max}\%
less posting mass than the fixed $\ks$ that reaches the same coverage.

Figure~\ref{fig:lidlaw} plots $\kappa/n$ against the absolute value of
the statistic. Within a corpus the curve is invariant to everything we
varied. SIFT-1M and Deep-10M, ten times apart in $n$ and $|S|$, agree
bin by bin (median ratio \num{lidlaw/ratio-sift-deep10m} over
\num{lidlaw/ratio-sift-deep10m-bins} shared bins); the two halves of
the LAION split coincide (ratio \num{lidlaw/ratio-laion-ood-laion-in});
the shifted wiki and PubMed queries land on their in-distribution
curves (ratios \num{lidlaw/ratio-wiki-ood-wiki-in} and
\num{lidlaw/ratio-pubmed23-ood-pubmed23-in}); and T2I under inner
product lands on T2I under $\ell_2$ (ratio
\num{lidlaw/ratio-t2i-ip-in-t2i-in}) --- the rate is a property of the
query in its corpus, not of the distance function. Across corpora the
curves are parallel but offset, by up to
\num{lidlaw/ratio-wiki-in-sift}$\times$ at equal estimated dimension
(GloVe \num{lidlaw/ratio-glove-sift}$\times$, LAION
\num{lidlaw/ratio-laion-in-sift}$\times$ and PubMed23
\num{lidlaw/ratio-pubmed23-in-sift}$\times$ above SIFT, T2I
\num{lidlaw/ratio-t2i-in-deep10m}$\times$ above Deep-10M). The offset
is corpus-specific rather than a split by distance family --- the
three cosine corpora (wiki, LAION, PubMed23) sit at three levels
between the $\ell_2$ descriptors and the angular GloVe --- so the estimate is one of at least two statistics behind the
local level of $f_X$, and the table is fitted per corpus. The
cross-modal shift is the case the table does not cover, and the
certificate says so: text queries against the image corpus fall beyond
the 99th percentile of the image queries' estimates for
\num{lidlaw/beyond-t2i-ood}\% ($\ell_2$) and
\num{lidlaw/beyond-t2i-ip-ood}\% (inner product) of the set, and where
the two overlap the text queries sit
\num{lidlaw/ratio-t2i-ood-t2i-in}$\times$ above the image curve at
equal estimate: the shift changes the neighborhood structure, not only
its dimension. These queries are served uncertified; the decile table
would predict them to \num{lidlaw/err-t2i-ood} and
\num{lidlaw/err-t2i-ip-ood}, and the dial reaches
\num{lidlaw/t2i-adaptive-cov-min} and
\num{lidlaw/t2i-ip-adaptive-cov-min} at its $\ks$ cap against a target
of $0.99$. Certifying them takes a table fitted on a sample of text
queries, which is the same construction with a different calibration
set. The exactness radius is inert on every corpus (at most
\num{lidlaw/exact-frac-max}\% of queries certified).

\begin{figure}[t]
\centering
\includegraphics[width=\textwidth]{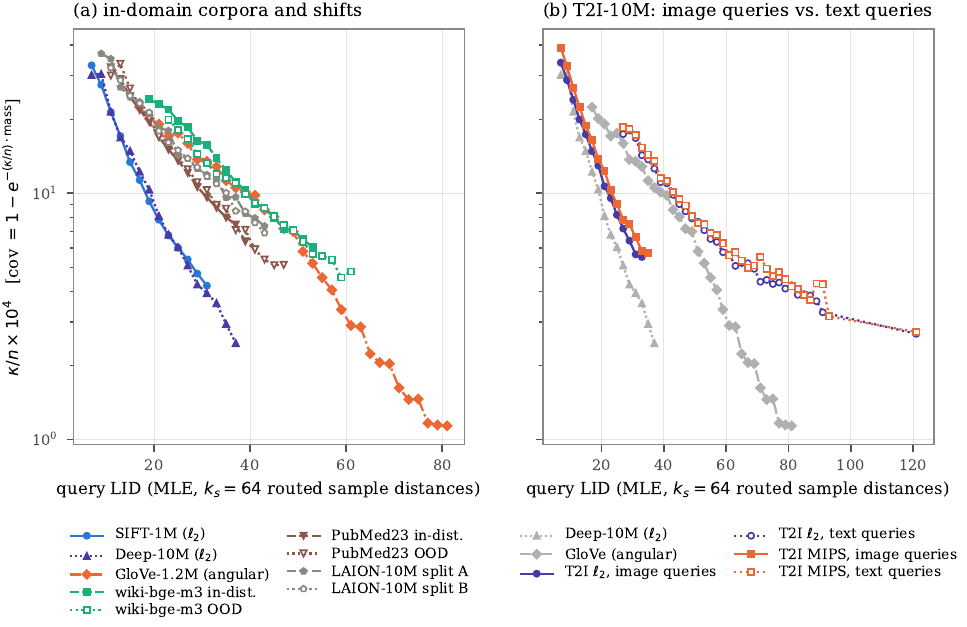}
\caption{The coverage rate per touched object, $\kappa/n$ in
$1-\mathrm{cov}=\exp(-(\kappa/n)\cdot\mathrm{mass})$, against the
query's intrinsic-dimension estimate (bins of width 2 with at least
100 queries; $\kappa$ fitted at $(b,\ks)=(4,8)$, the lowest-mass grid
point, where no bin is near saturation and the miss floor of
approximate signatures cannot enter). (a) Within a corpus one curve:
SIFT-1M and Deep-10M differ $10\times$ in $n$; the second half of
the LAION split and the shifted wiki and PubMed23 queries (hollow)
fall on their first-half and in-distribution curves. Across
corpora the curves are parallel with corpus-specific offsets. (b) On
T2I-10M the inner-product curve coincides with the $\ell_2$ curve, and
the text queries (hollow) sit above the image queries at equal
estimate --- a shift the estimate does not absorb.}
\label{fig:lidlaw}
\end{figure}

\begin{table}[t]
\centering\footnotesize
\setlength{\tabcolsep}{3pt}
\begin{tabular}{llcccrl}
\toprule
corpus (metric, $n$) & queries & LID (p10--p90) & cov$(8,64)$ & $\kappa/n\!\times\!10^4$ (dec.\ 1$\to$10) & $\rho$ & held-out err.\\
\midrule
SIFT-1M ($\ell_2$, 1M) & in-dist. & 18.4 (11--25) & 0.9977 & 26.7$\to$5.03 (5$\times$) & $-0.52$ & 0.0003 (0.0004)\\
\addlinespace[2pt]
Deep-10M ($\ell_2$, 10M) & in-dist. & 21.1 (13--30) & 0.9851 & 23.5$\to$3.22 (7$\times$) & $-0.57$ & 0.0000 (0.0001)\\
\addlinespace[2pt]
GloVe-1.2M (angular, 1.18M) & in-dist. & 53.0 (26--75) & 0.8792 & 20.0$\to$1.08 (18$\times$) & $-0.71$ & 0.0021 (0.0055)\\
\addlinespace[2pt]
wiki-bge-m3 (cosine, 0.2M) & in-dist. & 35.6 (24--48) & 0.9959 & 24.0$\to$6.18 (4$\times$) & $-0.52$ & --\\
 & OOD & 40.0 (28--53) & 0.9926 & 21.8$\to$5.69 (4$\times$) & $-0.48$ & 0.0009 (0.0033)\\
\addlinespace[2pt]
LAION-10M (cosine, 10M) & split A & 24.5 (13--44) & 0.9913 & 31.4$\to$5.76 (5$\times$) & $-0.59$ & --\\
 & split B & 24.1 (13--43) & 0.9920 & 34.0$\to$5.43 (6$\times$) & $-0.60$ & 0.0003 (0.0006)\\
\addlinespace[2pt]
PubMed23 (cosine, 24M) & in-dist. & 22.8 (14--34) & 0.9937 & 30.3$\to$6.29 (5$\times$) & $-0.55$ & --\\
 & OOD & 25.8 (16--40) & 0.9853 & 31.8$\to$6.11 (5$\times$) & $-0.58$ & 0.0039 (0.0084)\\
\addlinespace[2pt]
T2I-10M ($\ell_2$, 10M) & image$\to$image & 16.9 (9--27) & 0.9953 & 31.6$\to$5.39 (6$\times$) & $-0.58$ & --\\
 & text$\to$image & 58.4 (36--107) & 0.9281 & 24.1$\to$6.04 (86\% beyond) & $-0.47$ & 0.0440 (0.0672)\\
\addlinespace[2pt]
T2I-10M (MIPS, 10M) & image$\to$image & 18.3 (10--30) & 0.9964 & 34.7$\to$6.15 (6$\times$) & $-0.57$ & --\\
 & text$\to$image & 57.4 (36--98) & 0.9543 & 27.0$\to$6.38 (66\% beyond) & $-0.46$ & 0.0233 (0.0421)\\
\bottomrule
\end{tabular}

\caption{The per-query certificate across corpora ($\alpha{=}2\%$,
seed 1, true 10-NN, $10^4$ queries per set; exact signatures on
SIFT-1M, GloVe and wiki, \graft{} signatures on Deep-10M, PubMed23 and
T2I-10M). LID is the Levina--Bickel estimate over the $\ks{=}64$ routed
sample distances; $\kappa/n$ is the rate fitted at $(4,8)$ in the first
and last decile of the estimate, with the ratio between them, or, for
the text queries binned on the image queries' deciles, the share of
queries beyond the calibrated range; $\rho$ is Spearman's correlation
between the estimate and per-query coverage at $(4,8)$; the held-out
error is the mean absolute error of the coverage predicted by the
decile table at $(8,64)$ on queries it was not fitted on --- the other
half of a split for a single query set, the shifted set otherwise ---
with the error of the unconditioned certificate in parentheses.}
\label{tab:lidlaw}
\end{table}

\section{Evaluation}\label{sec:exp}
Every number below was produced by the run that measured it, into a
manifest, and every baseline was run by us on the same machine, queries
and ground truth, swept to saturation. The section reports the operating
surface, serving under enforced memory caps against five other systems,
the in-RAM comparison, the hardest corpus we found, and Deep-1B.

\paragraph{Setup.}\label{sec:setup}
Machines: a 10-core laptop (M5, 24\,GB) and a 4-socket Xeon E7-4809
v3 (Haswell, 2\,GHz, 8 cores/16 threads per socket, 1.5\,TB RAM, NVMe;
``d0''), whose four NUMA nodes deliver ${\sim}37$\,GB/s each and
${\sim}50$\,GB/s in aggregate to one process interleaved across them. Datasets: SIFT-1M \cite{jegou2011bigann} ($d{=}128$),
GloVe-1.2M \cite{pennington2014glove} ($d{=}200$, angular), Deep-10M and
Deep-100M \cite{babenko2016deep1b} ($d{=}96$);
10{,}000 queries each; recall@10 vs.\ exact ground truth. Protocol:
three interleaved repetitions, median throughput; fixed seeds; every
number is written by its run into a versioned manifest. Every baseline
is swept to saturation rather than reported at one operating point, so
that comparisons are between frontiers and not between tunings
\cite{aumuller2020annbenchmarks}. The fanless
laptop drifts by up to $1.7\times$ in throughput between a cold and a
hot run at identical work, so the laptop comparisons of
Table~\ref{tab:compare} time both systems in the same interleaved
repetitions (each repetition runs both sweeps back to back, the order
alternating), and Table~\ref{tab:solo-points} reports the min--max
spread over seed runs rather than a single figure. Seed
stability over 5 seeds: $\sigma(\mathrm{recall})\le10^{-4}$ (SIFT),
${\sim}5\times10^{-5}$ (Deep-10M), ${\sim}6\times10^{-4}$ (GloVe);
Deep-100M rows are single-seed (a seed is a full signature build), with
the 10M seed variance bounding the expected fluctuation.

\subsection{Operating surface}\label{sec:surface-exp}
Table~\ref{tab:solo-points} shows representative \solo{} operating
points. The dial is monotone everywhere measured --- recall never
saturates short of scanning everything --- and every recall agrees with
the certificate of Proposition~\ref{obs:identity}.

\begin{table}[t]
\centering\small
\begin{tabular}{llcc}
\toprule
dataset (machine) & configuration & recall@10 & qps\\
\midrule
SIFT-1M (laptop) & $b{=}2,\ks{=}192$, SQ4 & 0.9917 & 27{,}306--28{,}115\\
 & $b_0{=}8,\ks{=}64,p{=}0.5$, SQ4 & 0.9939 & 21{,}639--34{,}240\\
 & $b{=}8,\ks{=}64$, SQ4 & 0.9964 & 16{,}874--29{,}188\\
 & $b_0{=}16,\ks{=}64,p{=}1$, SQ4 & 0.9973 & 15{,}169--26{,}756\\
 & $b_0{=}8,\ks{=}128,p{=}0.5$, SQ4 & 0.9986 & 10{,}664--19{,}259\\
 & $b{=}8,\ks{=}128$, SQ4 & 0.9990 & 8{,}063--15{,}628\\
 & $b_0{=}16,\ks{=}128,p{=}0.5$, SQ4 & 0.9993 & 6{,}285--11{,}890\\
\midrule
GloVe-1.2M (laptop) & $b_0{=}32,\ks{=}128,p{=}0.5$, SQ4 & 0.9793 & 993--1{,}217\\
 & $b{=}16,\ks{=}256$, SQ4 & 0.9860 & 773--932\\
 & $b_0{=}32,\ks{=}256,p{=}1$, SQ4 & 0.9881 & 709--881\\
 & $b_0{=}32,\ks{=}256,p{=}0.5$, SQ4 & 0.9910 & 570--694\\
 & $b{=}32,\ks{=}256$, SQ4 & 0.9936 & 426--503\\
 & $b_0{=}32,\ks{=}512,p{=}1$, SQ4 & 0.9942 & 414--484\\
 & $b_0{=}32,\ks{=}512,p{=}0.5$, SQ4 & 0.9950 & 338--393\\
 & $b{=}32,\ks{=}512$, SQ4 & 0.9955 & 251--287\\
 & depth-2, $\ks{=}64,\kslocal{=}16$, exact & 0.9883 & 150\\
 & depth-2, $\ks{=}64,\kslocal{=}32$, exact & 0.9899 & 124\\
 & depth-2, $\ks{=}128,\kslocal{=}16$, exact & 0.9973 & 68\\
 & depth-2, $\ks{=}128,\kslocal{=}32$, exact & 0.9978 & 52\\
\midrule
Deep-10M (laptop) & $b{=}8,\ks{=}128$, SQ8 & 0.9952 & 3{,}659\\
 & $b{=}8,\ks{=}256$, SQ8 & 0.9982 & 1{,}754\\
 & $b{=}16,\ks{=}128$, SQ8 & 0.9985 & 1{,}884\\
 & $b{=}16,\ks{=}256$, SQ8 & 0.9994 & 685\\
\midrule
Deep-100M (d0) & $b{=}16,\ks{=}256$, SQ8 & 0.9982 & 427\\
 & $b{=}16,\ks{=}384$, SQ8 & 0.9988 & 320\\
 & $b{=}16,\ks{=}512$, SQ8 & 0.9990 & 252\\
 & $b{=}32,\ks{=}128$, SQ8 & 0.9981 & 471\\
 & $b{=}32,\ks{=}256$, SQ8 & 0.9991 & 268\\
\bottomrule
\end{tabular}

\caption{\solo{} operating points. Deep rows use the SQ8 screen at the
serving default $R{=}100$; their recalls are bit-identical to the
full-precision scan of the same configuration, and every measured row
is listed. SQ4 rows are the list-major 4-bit stream with rank-ordered
postings (\S\ref{sec:kernels}; $R{=}200$, router $ef{=}\ks$, all ten
cores, $10^4$ queries), from signatures built at
$\kb{=}\num{sift/kb}$ (SIFT) and $\kb{=}\num{glove/kb}$ (GloVe); a
row with $p$ is a tapered frontier $b(t)=\lceil b_0(1-t/\ks)^p\rceil$,
and the rows shown are the non-dominated points of each sweep, thinned
to recall steps of $2\times10^{-4}$. Depth-2 GloVe rows are
full-precision scans of the recursive index (\S\ref{sec:recursion}).
Deep rows were served from signatures built once at $\kb{=}64$ ---
$b$ is serve-time. SIFT SQ4 recall is the median over
\num{sift/seeds} seeds and its throughput the full min--max spread
across them (laptop thermal drift), GloVe SQ4 over \num{glove/seeds};
other rows are single-seed.}
\label{tab:solo-points}
\end{table}

\subsection{Serving Deep-100M at 1\,GB resident, and below}\label{sec:capped}
Table~\ref{tab:capped} is the disk mode of \S\ref{sec:kernels} under
enforced cgroup caps: swap off, every input evicted from the page cache
\emph{before} it is mapped, one operating point per process, the
scope's own \texttt{memory.peak}. The object level is the same in every
row --- posting-ordered SQ8 leaf blocks on NVMe, read by explicit
\texttt{pread}, results bit-identical to RAM serving --- and what
changes is the router. At depth 1 it is the flat $2$M-point image,
$\num{d2c/router-flat-mib}$\,MiB, file-backed and reclaimable page by
page. At depth 2 it is the index's own rule pointed upward
(\S\ref{sec:ceiling}): a $2\%$ sample of the vocabulary,
$\num{d2c/router-d2-mib}$\,MiB resident
($\num{d2c/router-ratio}\times$ less), with the vocabulary's own
posting lists as an SQ4 stream on NVMe read by the same kernel, so a
query routes in two hops --- $\num{d2c/router-io-mb}$\,MB of counted
IO --- and then scans the object level exactly as before.

\begin{table}[t]
\centering\small
\resizebox{\textwidth}{!}{
\begin{tabular}{clrrcrrrr}
\toprule
depth & cap & thr. & $\ks$ & recall@10 & qps cold / warm & route ms & router MiB & B/obj.\\
\midrule
1 & 2\,GB & 16 & 256 & 0.9965 & 20.9 / 20.2 & 1.0 & 963 & 21.5\\
1 & 2\,GB & 16 & 512 & 0.9979 & 10.7 / 10.5 & 1.3 & 963 & 21.5\\
1 & 1\,GB & 16 & 256 & 0.9965 & 8.0 / 1.7 & 547.3 & 963 & 10.7\\
1 & 1\,GB & 16 & 512 & \multicolumn{3}{l}{\emph{timed out (thrashing)}} & 963 & 10.7\\
\midrule
2 & 2\,GB & 16 & 256 & 0.9964 & 19.9 / 19.9 & 3.6 & 18 & 21.5\\
2 & 2\,GB & 16 & 512 & 0.9977 & 10.2 / 9.8 & 7.5 & 18 & 21.5\\
2 & 1\,GB & 16 & 256 & 0.9964 & 17.6 / 16.3 & 12.8 & 18 & 10.7\\
2 & 1\,GB & 16 & 512 & 0.9977 & 9.1 / 8.7 & 18.4 & 18 & 10.7\\
2 & 768\,MB & 16 & 256 & 0.9964 & 16.5 / 15.1 & 16.2 & 18 & 8.1\\
2 & 768\,MB & 8 & 512 & 0.9977 & 8.9 / 8.3 & 25.9 & 18 & 8.1\\
2 & 768\,MB & 16 & 512 & \multicolumn{3}{l}{\emph{OOM-killed}} & 18 & 8.1\\
2 & 512\,MB & 8 & 256 & 0.9964 & 13.8 / 12.8 & 30.4 & 18 & 5.4\\
2 & 512\,MB & 4 & 512 & 0.9977 & 5.5 / 5.3 & 54.8 & 18 & 5.4\\
2 & 512\,MB & 8 & 512 & \multicolumn{3}{l}{\emph{OOM-killed}} & 18 & 5.4\\
2 & 384\,MB & 8 & 256 & 0.9964 & 13.0 / 12.3 & 33.5 & 18 & 4.0\\
2 & 384\,MB & 4 & 512 & 0.9977 & 5.2 / 5.0 & 63.2 & 18 & 4.0\\
2 & 384\,MB & 8 & 512 & \multicolumn{3}{l}{\emph{OOM-killed}} & 18 & 4.0\\
2 & 256\,MB & 4 & 256 & 0.9964 & 7.2 / 7.3 & 66.8 & 18 & 2.7\\
2 & 256\,MB & 2 & 512 & 0.9977 & 2.9 / 2.9 & 129.0 & 18 & 2.7\\
2 & 256\,MB & 4 & 512 & \multicolumn{3}{l}{\emph{OOM-killed}} & 18 & 2.7\\
\bottomrule
\end{tabular}
}
\caption{Deep-100M under enforced cgroup caps through the flat router
(depth 1) and through the router indexed by \solo{} (depth 2). $b{=}16$,
$R{=}100$, 1{,}000-query subset, cold / warm reps; ``route'' is the
warm rep's routing time. Resident B/object is the cap divided by $n$;
the cgroup peak equals the cap in every finished row. A row that did
not finish says why: the flat router at 1\,GB thrashed past a 40-minute
limit, and the depth-2 points below 1\,GB were killed at $\ks{=}512$ at
the thread count that served $\ks{=}256$ --- the larger pass set's
buffers, not the router --- and finish at half that thread count, the
row above each kill. Rows at different thread counts are never
compared.}
\label{tab:capped}
\end{table}

\paragraph{The flat router does not fit 1\,GB at $10^8$.}
At 2\,GB depth 1 serves at its uncapped rate; at 1\,GB it reads
$\num{d2c/d1-1g-ks256-qps-cold}$ qps cold and
$\num{d2c/d1-1g-ks256-qps-warm}$ warm, with routing at
$\num{d2c/d1-1g-ks256-route-ms}$\,ms per query, recall unchanged: the
beam paging its graph while the scan's own reads evict it. Its floor
sits between 1.5 and 2\,GB and depends on the operating point ---
Table~\ref{tab:capped-floor} --- $16$--$21$ B/object against the
$\num{ceiling/bytes-per-object}$\,B/object router term of
\S\ref{sec:ceiling}, the difference being the scratch and the reads in
flight.

\paragraph{Depth 2 removes the wall and costs nothing above it.} At
2\,GB the two routes read the same recall and rate
($\num{d2c/d2-2g-ks512-recall}$ at $\num{d2c/d2-2g-ks512-qps-warm}$ vs.\
$\num{d2c/d1-2g-ks512-recall}$ at $\num{d2c/d1-2g-ks512-qps-warm}$ qps at
$\ks{=}512$); at 1\,GB depth 2 serves $\num{d2c/d2-1g-ks512-recall}$ at
$\num{d2c/d2-1g-ks512-qps-warm}$ qps where depth 1 collapsed, and the
route costs $\num{d2c/d2-1g-ks512-route-ms}$\,ms of the query. That is
$\num{capped/bytes-per-object}$ resident bytes per object at recall
$\num{capped/1g-recall-d2}$.

\paragraph{Below 1\,GB the cap is on the threads, not the index.} With an $\num{d2c/router-d2-mib}$-MiB router the
resident state is per-thread scratch and buffers --- the dedupe bitmap
alone is $n/8$ bytes a thread, $200$\,MB at 16 threads --- so the
floor moves with the thread count: $512$\,MB at 8 threads serves
$\num{d2c/d2-512m-ks256-qps-warm}$ qps ($\num{d2c/d2-512m-b-per-object}$
B/object) and $\num{capped/solo-floor-mb}$\,MB at
$\num{capped/solo-floor-threads}$ threads
$\num{d2c/d2-256m-ks256-qps-warm}$ qps ($\num{capped/solo-bytes}$
B/object), both at recall $\num{capped/solo-floor-recall}$. The
$\ks{=}512$ points killed at those thread counts finish at half of
them, every one at recall $\num{capped/solo-floor512-recall}$: 768\,MB at
$\num{capped/solo-ks512-768m-threads}$ threads
($\num{capped/solo-ks512-768m-qps}$ qps), 512 and 384\,MB at
$\num{capped/solo-ks512-512m-threads}$
($\num{capped/solo-ks512-512m-qps}$ and
$\num{capped/solo-ks512-384m-qps}$ qps), and
$\num{capped/solo-floor512-mb}$\,MB at
$\num{capped/solo-floor512-threads}$ threads
($\num{capped/solo-floor512-qps}$ qps) --- the 1\,GB row's recall at the
256\,MB row's $\num{capped/solo-bytes}$ B/object, bought with threads.
The resident
budget is router plus threads times (scratch plus reads in flight), the
three-term law of \S\ref{sec:ceiling}, and at $10^8$ the second term is
the whole of it once the router recurses.

\paragraph{The flat router's floor, and a measurement trap.}
Table~\ref{tab:capped-floor} measures the flat router at 1\,GB under
three protocols. With the router image evicted before it is mapped, so
that its pages are charged to the scope, the row reads
$\num{floor/1g-evict-qps-ks256-cold}$ / $\num{floor/1g-evict-qps-ks256-warm}$
qps. With the image read once inside the scope before the timed
repetitions, the cold repetition reads
$\num{floor/1g-prewarm-qps-ks256-cold}$ and the warm one
$\num{floor/1g-prewarm-qps-ks256-warm}$, as the scan evicts the image.
With the image already in page cache from another process, charged to
nobody, the same row reads $19.8$ qps: the cap then bounds the scan's
page cache and not the router, and the figure is not a 1\,GB
measurement. A cgroup cap on a process whose index is mapped bounds
only what that process faulted in itself; every row in this paper
therefore evicts every input before mapping it, and the recursion was
served at this scale because the flat router's true floor sits above
1\,GB.

\begin{table}[t]
\centering\small
\resizebox{\textwidth}{!}{
\begin{tabular}{llcrr}
\toprule
cap & router image at start & recall@10 ($\ks{=}256/512$) & qps cold/warm, $\ks{=}256$ & $\ks{=}512$\\
\midrule
2\,GB & evicted & 0.9965 / 0.9979 & 20.3 / 20.5 & 10.5 / 10.3\\
1.5\,GB & evicted & 0.9965 / 0.9979 & 20.3 / 19.9 & 5.3 / 1.1\\
1\,GB & evicted & 0.9965 / -- & 7.7 / 1.7 & timed out\\
1\,GB & pre-read in scope & 0.9965 / 0.9979 & 20.5 / 10.9 & 1.5 / 0.8\\
\midrule
1\,GB (Table~\ref{tab:capped}) & cached before the scope & 0.9965 / 0.9979 & 19.8 / 19.3 & 10.2 / 10.0\\
\bottomrule
\end{tabular}
}
\caption{The flat router's floor at $10^8$, 16 threads, 1{,}000
queries, one process per cap. ``Evicted'': the router image is
\texttt{fadvise}d out before mapping, so its pages are charged to the
scope; ``pre-read in scope'': read once before the timed reps; the last
row has the image cached by another process before the scope existed. Recall is identical in every row; the cap moves
throughput only.}
\label{tab:capped-floor}
\end{table}

\subsection{The same cap, six other systems}\label{sec:capped-baselines}
Table~\ref{tab:capped-baselines} runs SPANN, DiskANN, vote-served
\misi{}, a Faiss IVF-PQ index, ScaNN~\cite{guo2020scann} and
RaBitQ~\cite{gao2024rabitq,gao2025exrabitq} through the identical
protocol, each swept to its own recall ceiling. Two things separate the arms, and only
one of them is the memory floor.

The floors are $\num{capped/solo-bytes}$, $\num{capped/rabitq-rq1-codes/bytes}$,
$\num{capped/ivfpq-bytes}$, $\num{capped/scann-ah/bytes}$,
$\num{capped/rabitq-rq4-codes/bytes}$, $80$, $80$, $96$,
$\num{capped/scann-bf16/bytes}$ and $\num{capped/rabitq-rq4-refine/bytes}$
bytes per object, and they order the
systems exactly: below $8$\,GB the graph and cluster indexes stop ---
not degrade, stop. SPANN is killed at $6$\,GB and every cap under it,
\misi{} at $4$, DiskANN at $8$, ScaNN's codes-only index at
$\num{capped/scann-ah/died-gb}$ and its re-ranking index at
$\num{capped/scann-bf16/died-gb}$; above their floors they pay almost
nothing for the cap, SPANN serving $\num{capped/spann-8g-recall}$ at
$\num{capped/spann-8g-qps}$ qps under an $8$\,GB cap against
$\num{capped/spann-uncapped-qps}$ uncapped. That is the shape of a
resident working set: it fits or it does not, and page cache is the only
thing the surplus was buying.

The second separation is the interesting one, because it is not about
memory at all. IVF-PQ \emph{does} run at $1$\,GB --- it is the only
baseline that does --- so the region is not empty for want of a small
enough index. It is empty because the small index cannot answer: at
$\num{capped/ivfpq-bytes}$\,B/object its recall saturates at
$\num{capped/ivfpq-ceiling}$ and stays there at every cap and every
$\mathit{nprobe}$ from $8$ to $512$, since PQ distances are the final
score rather than a screen --- the same ceiling SPANN's Int8 scorer
meets at $\num{sat/spann-int8}$ (\S\ref{sec:compare}). And what the cap
costs it is throughput, catastrophically: the same configuration runs at
$\num{capped/ivfpq-1g-qps}$ qps at $1$\,GB and $182$ at $8$, a
$40\times$ collapse, because a probed list is a random read. \solo{}
under the same reduction of cap moves by $\num{capped/solo-1g-vs-2g-pct}\%$
($\num{d2c/d2-1g-ks512-qps-warm}$ vs.\ $\num{d2c/d2-2g-ks512-qps-warm}$
qps through the depth-2 route), because its reads were sequential and
counted before the query began. The $85\times$ read
amplification of \S\ref{sec:kernels} is not a property of graphs; it is
a property of not knowing the access order, and it reappears here in a
system with no graph in it.

ScaNN makes the same point from the other side, and from the system a
reader is likeliest to name. Its tree-and-asymmetric-hashing index
(20{,}000 leaves, 4-bit codes with two dimensions per block, anisotropic
threshold $0.2$, trained on $2.5$\,M points) comes in two forms. Scoring
on the codes alone it is $\num{capped/scann-ah/bytes}$\,B/object on
disk --- the IVF-PQ shape --- and its recall saturates at
$\num{capped/scann-ah/ceiling}$ at every cap it survives and every leaf
count from $400$ to $1{,}600$: a 4-bit code as the final score is a
worse scorer than an 8-bit one, as it should be. Re-ranking the top
candidates against a bfloat16 copy of the dataset recovers the quality
--- $\num{capped/scann-bf16/ceiling}$ at a re-rank depth of $5{,}000$
and $3{,}200$ leaves, $\num{capped/scann-bf16/uncapped-qps}$ qps
uncapped on 64 threads, above HNSW's saturation --- but the copy is
$192$ of the $\num{capped/scann-bf16/bytes}$\,B/object, ScaNN loads its
artifacts into memory, and the index is killed at
$\num{capped/scann-bf16/died-gb}$\,GB and below; at
$\num{capped/scann-bf16/floor-gb}$\,GB it serves
$\num{capped/scann-bf16/32g-recall}$ at
$\num{capped/scann-bf16/32g-qps}$ qps (depth $2{,}000$). So at Deep-100M
ScaNN occupies both ends of the table at once, and neither is the
region in question: small and capped at $\num{capped/scann-ah/ceiling}$,
or accurate and resident at $\num{capped/scann-bf16/floor-gb}$\,GB.
Where it runs under a cap its rate is unchanged, as SPANN's and
DiskANN's are: the artifacts fit or they do not.

RaBitQ closes the quantizer column. Its 1-bit codes
($\num{capped/rabitq-rq1-codes/bytes}$\,B/object with ids) run at every
cap of the table and saturate at $\num{capped/rabitq-rq1-codes/ceiling}$;
the 4-bit extension ($\num{capped/rabitq-rq4-codes/bytes}$\,B/object)
saturates at $\num{capped/rabitq-rq4-codes/ceiling}$, within three
thousandths of IVF-PQ's ceiling, and like IVF-PQ it pays the cap in
throughput, $\num{capped/rabitq-rq4-codes/1g-qps}$ qps at 1\,GB against
$\num{capped/rabitq-rq4-codes/12g-qps}$ at 12, because a probed list of
codes is a random read. Re-ranking against the exact vectors reaches
$\num{capped/rabitq-rq4-refine/ceiling}$ with the 4-bit codes and
$\num{capped/rabitq-rq1-refine/ceiling}$ with the 1-bit ones, and the
library materializes the flat copy on load, so the re-ranked index runs
at $\num{capped/rabitq-rq4-refine/floor-gb}$\,GB and is killed at
$\num{capped/rabitq-rq4-refine/died-gb}$. Three quantizers, five code
lengths, one shape: codes as the final score saturate below $0.9$ at
every cap, and codes as a screen before the true distance reach $0.99$
only with the true distance resident.

\begin{table}[t]
\centering\small
\resizebox{\textwidth}{!}{
\begin{tabular}{lrcccccc}
\toprule
\multirow{2}{*}{system} & \multirow{2}{*}{B/obj} & \multicolumn{6}{c}{recall@10 / qps under an enforced cap}\\
\cmidrule(lr){3-8}
 & & 1\,GB & 2\,GB & 4\,GB & 6\,GB & 8\,GB & 12\,GB\\
\midrule
\textbf{SOLO} (scan, depth 2) & 2.7 & 0.9977 / 9.1 & 0.9977 / 10.2 & $\cdot$ & $\cdot$ & $\cdot$ & $\cdot$\\
Faiss IVF-PQ & 56.1 & 0.8747 / 4.5 & 0.8747 / 6.0 & 0.8747 / 15.4 & 0.8747 / 178 & 0.8747 / 182 & 0.8747 / 182\\
ScaNN (AH 4-bit codes) & 52.2 & --- & --- & --- & --- & 0.4281 / 4{,}525 & 0.4281 / 4{,}608\\
ScaNN (AH + bf16 re-rank) & 244.2 & --- & --- & --- & --- & --- & ---\\
RaBitQ (4-bit codes) & 76.1 & 0.8715 / 4.1 & 0.8715 / 5.1 & 0.8715 / 9.7 & 0.8715 / 32.3 & 0.8715 / 96.5 & 0.8715 / 96.5\\
RaBitQ (1-bit codes) & 28.1 & 0.3450 / 167 & 0.3450 / 624 & 0.3450 / 1{,}371 & 0.3450 / 1{,}350 & 0.3450 / 1{,}382 & 0.3450 / 1{,}300\\
RaBitQ (4-bit + flat re-rank) & 460.1 & --- & --- & --- & --- & --- & ---\\
RaBitQ (1-bit + flat re-rank) & 412.1 & --- & --- & --- & --- & --- & ---\\
SPANN & 80 & --- & --- & --- & --- & 0.9911 / 185 & 0.9981 / 46.6\\
\misi{} (votes) & 80 & --- & --- & --- & 0.9936 / 7.2 & 0.9936 / 8.1 & $\cdot$\\
DiskANN & 96 & --- & --- & --- & --- & --- & 0.9984 / 568\\
\bottomrule
\end{tabular}
}
\caption{Deep-100M under enforced cgroup caps (E5: swap off, inputs
evicted, the scope's own \texttt{memory.peak}), each system at the
setting reaching its highest recall at that cap. A dash is an OOM kill;
a dot is a cap not run, which for \solo{} is every cap above $4$\,GB
(its recall is already at ceiling and its rate within
$\num{capped/solo-1g-vs-2g-pct}\%$ of uncapped at $1$\,GB) and for
\misi{} is $12$\,GB. Bytes per object are the measured floor: \solo{}'s
is the $\num{capped/solo-floor-mb}$\,MB depth-2 point of
Table~\ref{tab:capped} (at $\num{capped/solo-floor-threads}$ threads,
$\ks{=}256$), the others the smallest cap each survived --- IVF-PQ's
codes plus ids, ScaNN's codes and leaf map, RaBitQ's codes plus ids
(plus, for the re-ranking indexes, ScaNN's bfloat16 and RaBitQ's float
copy of the dataset), SPANN's and \misi{}'s postings, DiskANN's PQ
table and graph.}
\label{tab:capped-baselines}
\end{table}

\subsection{Comparison with prior indexes}\label{sec:compare}
Table~\ref{tab:compare} places \solo{} against the strongest systems
we could run on the same machines, queries, and ground truth, and
Figure~\ref{fig:pareto} draws the Deep-100M in-RAM frontier.

\begin{figure}[t]
\centering
\includegraphics[width=0.72\textwidth]{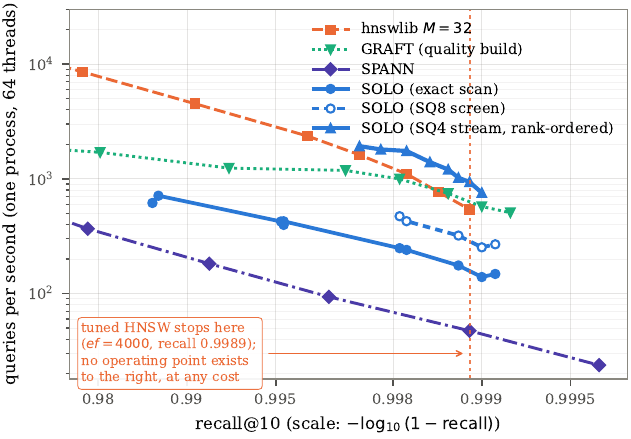}
\caption{Deep-100M in RAM, same machine, queries, and ground truth;
every system swept to saturation; one process, 64 threads. Served
with the gather kernel and 8-bit codes (hollow circles), \solo{} is
slower than the graphs wherever they can go; served with the
list-major 4-bit stream and rank-ordered postings of
\S\ref{sec:kernels} (triangles; timed against hnswlib in the same
interleaved repetitions --- the hnswlib curve drawn is its saturation
sweep, which the interleaved rerun reproduces within 6\%, identical recall), the same
index is the faster system at every recall above
$\num{deep100m/speedup-floor}$, by
$\num{deep100m/speedup-lo}$--$\num{deep100m/speedup-hi}\times$, and
continues to $\num{deep100m/stream-max-recall}$ past the orange wall,
where the tuned graph runs out of dial: hnswlib has no operating
point to its right at any $ef$. The scan architectures do not have
that wall --- SPANN climbs past it too, at a steeper price. DiskANN
is not plotted: it belongs to the capped-memory regime of
Table~\ref{tab:capped-baselines} ($0.9984$ at $568$ qps from $12$\,GB --- faster
than \solo{} there, and unable to run below $9.6$\,GB).}
\label{fig:pareto}
\end{figure}

\begin{table}[t]
\centering\small
\begin{tabular}{llccl}
\toprule
scale & system & recall@10 & qps & note\\
\midrule
\multirow{9}{*}{Deep-100M (d0, RAM)}
 & hnswlib $M{=}32$ & 0.9989 & 531 & saturated ($ef{=}4000$)\\
 & \graft{} & 0.9992 & 506 & quality build, $ef{=}4000$\\
 & \solo{} & 0.9991 & 268 & SQ8 gather, dial monotone\\
 & \solo{} & 0.9989 & 936 & $b_0{=}32,\ks{=}512,p{=}1$, SQ4, matched\\
 & \solo{} & 0.9990 & 755 & $b_0{=}32,\ks{=}512,p{=}0.5$, SQ4\\
 & SPANN & 0.9916 & 182 & authors' cfg, SIR $=512$\\
 & SPANN & 0.9989 & 47 & SIR $=2048$\\
 & SPANN & 0.9996 & 24 & SIR $=4096$\\
 & SPANN, native Int8 & 0.9561 & 73 & saturated (scorer, not screen)\\
\midrule
\multirow{3}{*}{Deep-100M (d0, $4\times$ socket-local)}
 & hnswlib $M{=}32$ & 0.9989 & 876 & saturated ($ef{=}4000$)\\
 & \solo{} & 0.9989 & 1{,}918 & $b_0{=}32,\ks{=}512,p{=}1$, SQ4, matched\\
 & \solo{} & \textbf{0.9990} & 1{,}631 & $b_0{=}32,\ks{=}512,p{=}0.5$, SQ4\\
\midrule
\multirow{3}{*}{Deep-100M (d0, capped)}
 & DiskANN & 0.9984 & 568 & 12\,GB; OOM $<9.6$\,GB\\
 & \misi{} (vote, disk) & 0.9936 & 8--30 & 8--16\,GB\\
 & \solo{} (leaf blocks) & 0.9982 & 20 & \textbf{1\,GB}\\
\midrule
\multirow{3}{*}{Deep-10M (d0)}
 & SPANN & 0.9921 & 384 & authors' cfg, SIR $=256$\\
 & SPANN & 0.9993 & 99 & SIR $=1024$\\
 & \solo{} & 0.9994 & 460 & $b{=}16,\ks{=}256$, f32\\
\midrule
\multirow{5}{*}{GloVe-1.2M (laptop)}
 & hnswlib $M{=}32$ & 0.9863 & 268 & saturated ($ef{=}3442$)\\
 & \misi{} (vote) & 0.9803 & 149 & $\ks{=}256$, $C{=}50{,}000$\\
 & \solo{} & 0.9881 & 808 & $b_0{=}32,\ks{=}256,p{=}1$, SQ4, matched\\
 & \solo{} & 0.9955 & 271 & $b{=}32,\ks{=}512$, SQ4\\
 & \solo{} & \textbf{0.9978} & 52 & depth-2, exact\\
\midrule
\multirow{3}{*}{SIFT-1M (laptop)}
 & hnswlib $M{=}32$ & 0.9991 & 4{,}166 & saturated ($ef{=}452$)\\
 & \solo{} & 0.9992 & 8{,}836 & $b_0{=}16,\ks{=}128,p{=}1$, SQ4, matched\\
 & \solo{} & \textbf{0.9993} & 6{,}977 & $b_0{=}16,\ks{=}128,p{=}0.5$, SQ4\\
\bottomrule
\end{tabular}

\caption{Same-machine comparisons. Each row is a system at the ceiling
of its own sweep, pooled over every seed run, so no system is penalised
for an unlucky one; where \solo{}'s sweep passes the baseline's
ceiling, its fastest row at or above that ceiling is added
(``matched''). The hnswlib and \solo{} SQ4 rows --- laptop and
Deep-100M --- were timed in the same interleaved repetitions (Setup),
laptop \solo{} rows being medians over \num{sift/interleaved-seeds}
(SIFT) and \num{glove/interleaved-seeds} (GloVe) index seeds; SQ4 rows
are the list-major stream at $ef{=}\ks$, $R{=}200$ on the laptop and
$R{=}2000$ with a $0.1\%$-clipped grid at $10^8$. The ``$4\times$
socket-local'' block serves both systems as four processes pinned one
per NUMA node, each with its own copy of its index (and stream),
queries split four ways, aggregate throughput; the single-process
block above it is the conservative row. NAPP \cite{tellez2013napp,naidan2015permutation},
discussed in the text, reaches $1.0$ at $10^6$ with a constant-size
vocabulary that does not scale its resolution with $n$.}
\label{tab:compare}
\end{table}

\paragraph{Throughput depends on the kernel, not the architecture.}
Served by the gather kernel with 8-bit codes, \solo{} is slower than
hnswlib wherever the graph can go --- by $2$--$5\times$ at $10^8$ --- and
\graft{}'s quality-tuned build attains $0.9992$. Served by the list-major
4-bit stream with rank-ordered postings (\S\ref{sec:kernels}), the same
index is the faster system at every recall the tuned graph reaches above
$\num{deep100m/speedup-floor}$ at $10^8$, by
$\num{deep100m/speedup-lo}$--$\num{deep100m/speedup-hi}\times$ in one
process: at hnswlib's ceiling, $\num{sat/hnsw-deep100m}$ at
$\num{sat/hnsw-deep100m-qps}$ qps against \solo{}'s
$\num{deep100m/matched-recall}$ at $\num{deep100m/matched-qps}$, and
$\num{deep100m/stream-max-recall}$ at $\num{deep100m/stream-max-qps}$
where the graph has stopped. The same holds above
$\num{sift/speedup-floor}$ on SIFT-1M
($\num{sift/speedup-lo}$--$\num{sift/speedup-hi}\times$; at hnswlib's
ceiling $\num{sat/hnsw-sift}$ at $\num{sat/hnsw-sift-qps}$ qps against
$\num{sift/matched-recall}$ at $\num{sift/matched-qps}$) and above
$\num{glove/speedup-floor}$ on GloVe
($\num{glove/speedup-lo}$--$\num{glove/speedup-hi}\times$; at the graph's
ceiling $\num{sat/hnsw-glove}$ at $\num{sat/hnsw-glove-qps}$ qps against
$\num{glove/matched-recall}$ at $\num{glove/matched-qps}$), where the
dial continues to $\num{glove/flat-max-recall}$. The single-process
$10^8$ rows are bound by d0's memory system: the stream draws
${\sim}70\%$ of the ${\sim}50$\,GB/s its four sockets deliver to one
interleaved process, while hnswlib's random walk is latency-bound.
Serving both systems as four socket-local replicas lifts hnswlib by
$\num{deep100m/repl-hnsw-gain}\times$ and \solo{} by
$\num{deep100m/repl-solo-gain}\times$, the stream becoming compute-bound
at ${\sim}20$\,ns per record per core, and the matched-recall ratio
becomes $\num{deep100m/repl-speedup-lo}$--$\num{deep100m/repl-speedup-hi}\times$
($\num{deep100m/repl-matched-qps}$ against $\num{deep100m/repl-hnsw-qps}$
qps at the graph's ceiling). Recall is bit-identical between the two
protocols.

\paragraph{What the graphs do not offer.} A monotone dial (hnswlib has
no operating point above $\num{sat/hnsw-deep100m}$ on this dataset at any
$ef$); an a-priori certificate; one-search inserts and exact deletes
(both graphs are static or repair-based); a build without clustering or
graph refinement ($5$\,min at $10^7$ against $19$\,min for \graft{}'s
tuned build at the same scale and $5.5$\,h for its quality build at
$10^8$); and the 1\,GB serving mode of Table~\ref{tab:capped}. The
claim there is capability, not efficiency per gigabyte: DiskANN serves
$0.9984$ at $568$ qps from $12$\,GB and is the faster system per unit of
memory it is given, and it does not run at all below $9.6$\,GB, where
\solo{} still serves this dataset.

\paragraph{Vote-based serving of the identical postings} --- \misi{}'s
path, top-$C$ by shared-term count then verification --- reaches
$\num{sat/misi-glove}$ on GloVe where the same data scanned reaches
$\num{glove/solo-max}$: ranking, not content, was the binding
constraint, consistent with the vote-resolution analysis of
\cite{chavez2026misi}.

\paragraph{SPANN} confirms the family diagnosis at both scales. Its
dial, like \solo{}'s, does not saturate --- it is a scan architecture
--- but its high-recall cost is steeper. At Deep-10M under its authors'
configuration \solo{} serves matched recall at $2$--$4.6\times$ the
throughput ($0.9994$ at $460$ qps against $0.9993$ at $99$), with a build
of ${\sim}5$ minutes against SPANN's $3.66$ hours ($3.1$\,h of it
$k$-means head selection), consistent with rank-$b$ dominating closure at
equal replication (\S\ref{sec:rule}). At Deep-100M SPANN's dial climbs
past everything ($0.9996$ at 24 qps), but at matched recall \solo{}
serves $5$--$7\times$ the throughput ($0.9988$ at $320$ against $0.9989$
at $47$), and the build asymmetry is stark: $57.5$\,h ($48.4$\,h of it
$k$-means head selection) and a 235\,GB store against \solo{}'s minutes
from signatures and 160\,GB of leaf blocks, on the same machine, data and
ground truth.

\subsection{MS MARCO Web Search at $10^{8}$: where the architecture is
  weakest}\label{sec:msmarco}
Every corpus above is either low intrinsic dimension (SIFT, Deep) or
$10^{6}$ (GloVe). MS MARCO Web Search is neither: $\num{msmarco/n-d}$
passage embeddings of $768$ dimensions under inner product, with queries
that are \emph{queries} while the database is \emph{passages} --- the
out-of-distribution case of \S\ref{sec:oos} at scale. We report it because
it is the hardest case we could find for a coverage-based index, and
because what survives there is more informative than another win on Deep.

Three things had to be right before the numbers meant anything, and
two of them concern the screen and apply to any per-dimension grid at
this dimension. Maximum inner product is reduced to a metric by the sphere
augmentation $x'=[x,\sqrt{M^{2}-\lVert x\rVert^{2}}]$, $q'=[q,0]$, so that
$\lVert q'-x'\rVert^{2}=\lVert q\rVert^{2}+M^{2}-2\langle q,x\rangle$ and
the existing L2 machinery orders by inner product exactly; objects are
routed as $[x,0]$, which costs nothing and yields the plain
$\langle x,s\rangle$ assignment kernel. The screen needed two corrections.
First, its integer SSD weights dimension $j$ by $1/\mathrm{step}_j^{2}$, so
a per-dimension grid serves a \emph{reweighted} metric: invisible where the
per-dimension ranges are homogeneous (SIFT, Deep), fatal at $768$
dimensions where they span $0.51$--$2.00$. One step for every coded
dimension is required. Second, the augmented coordinate must not be coded
at all --- $\lVert x\rVert$-derived, query-independent, and with a range
that dwarfs the data dimensions' --- so $M^{2}-\lVert x\rVert^{2}$ is
carried as a per-record bias added to the screen's integer SSD, with the
$\mathrm{f}32$ refinement reading the augmented rows. With the grid
per-dimension and the coordinate coded, recall is flat in $\ks$ and rises
with $R$ (the signature of a screen ranking by noise); with both
corrections it is flat in $R$ --- $R{=}100$, $500$ and $2000$ agree to four
decimals --- and tracks coverage.

\begin{table}[t]
\centering\small
\begin{tabular}{ccrrrrrr}
\toprule
& & \multicolumn{2}{c}{$10^{7}$} & \multicolumn{4}{c}{$10^{8}$}\\
\cmidrule(lr){3-4}\cmidrule(lr){5-8}
$b$ & $\ks$ & recall & qps & certificate & recall & qps & p99 (ms)\\
\midrule
8 & 64 & 0.7983 & 928 & 0.7696 & 0.7695 & 1{,}057 & 37\\
8 & 128 & 0.8623 & 640 & 0.8200 & 0.8199 & 575 & 51\\
8 & 256 & 0.9111 & 386 & 0.8602 & 0.8602 & 308 & 78\\
16 & 64 & 0.8728 & 862 & 0.8179 & 0.8178 & 820 & 47\\
16 & 128 & 0.9212 & 506 & 0.8644 & 0.8644 & 440 & 71\\
16 & 256 & 0.9539 & 275 & 0.8987 & 0.8987 & 235 & 110\\
\bottomrule
\end{tabular}

\caption{MS MARCO Web Search, $\alpha{=}2\%$, exact top-10 by inner
product. \emph{certificate} is coverage computed from the stored
signatures under the deployed router, before serving. Candidates per query
are $n$-independent by construction ($\num{msmarco/cand-1e7}$ at $10^{7}$
vs.\ $\num{msmarco/cand-1e8}$ at $10^{8}$ for $b{=}16,\ks{=}256$), and the
latency tail stays tight; what changes with $n$ is the level of the
coverage curve.}
\label{tab:msmarco}
\end{table}

\paragraph{The result, including where the index is not competitive.} At
$10^{8}$ the best point we measured is
$\num{msmarco/recall-1e8-b16-ks256}$ at
$\num{msmarco/qps-1e8-b16-ks256}$ qps, where tuned graph indexes on this
corpus operate near $0.99$. The cause is visible in the table: at matched
work the coverage level \emph{drifts down with $n$} by
$\num{msmarco/drift}$ per decade here
($\num{msmarco/drift-1e7}\to\num{msmarco/drift-1e8}$ at $b{=}16,\ks{=}256$)
against $0.0012$ on Deep over the same step. This is the spread regime of
\S\ref{sec:laws} compounding with scale: where coverage is a slow function
of $b\cdot\ks$, a decade of $n$ costs real recall, and buying it back needs
$b{=}32$ --- a stream we could not hold --- or a much larger $\ks$. On this
corpus, at this scale, \solo{} is not competitive at high recall, and no
amount of kernel work changes that; it is a property of how the neighbours
of a $768$-dimensional text embedding distribute over a $2\%$ sample.

\paragraph{What survives is the claim the paper is actually about.} The
certificate column is computed from the stored signatures alone --- no
serving, no stream, no kernels --- and it matches the served recall to
$\num{cert/msmarco100m-max-err}$ at all
$\num{cert/msmarco100m-points}$ operating points. Proposition~\ref{obs:identity}
therefore holds at $10^{8}$, on $768$-dimensional inner-product data, with
out-of-distribution queries, in the regime where the index performs
\emph{worst}: an operator running this workload would have known these
recalls from one ground-truth pass, before deploying anything, and would
have rejected the configuration for that reason. An index that predicts
its own recall only when the answer is favourable would not be worth
much. The same computation separates the two losses: routing exactly
rather than with the router would raise coverage by
$\num{cert/msmarco100m-router-cost}$
(to $\num{cert/msmarco100m-cov-ceiling}$ at $b{=}16,\ks{=}256$), so the
ceiling and the router's share of the gap are each visible before a query
is served.

\subsection{Deep-1B: the surface certified before it is served, then served}
\label{sec:deep1b}
At $10^{9}$ the asymmetry the certificate exploits is at its widest. A
serving artifact for a billion objects is hundreds of gigabytes and a day
of machine time; the \emph{signatures} are $\num{deep1b/sig-gib}$\,GiB and
answer, from one ground-truth pass, what recall the index would give at
every $(b,\ks)$ --- before any of that is built. That is what was run
first: a uniform $2\%$ vocabulary ($\num{deep1b/ns}$ points, GRAFT-indexed
in $\num{deep1b/router-h}$\,h) and one signature per object
($\num{deep1b/sig-gib}$\,GiB at $k_b{=}16$), and read the operating
surface off it (Table~\ref{tab:deep1b}). Then we built the stream and
served it (Table~\ref{tab:deep1b-served}).

\begin{table}[t]
\centering\small
\begin{tabular}{ccrrrr}
\toprule
$b$ & $\ks$ & cov.\ (exact route) & cov.\ (router) & cand./query & MB/query\\
\midrule
8 & 64 & 0.9388 & 0.9389 & 52{,}449 & 2.5\\
8 & 128 & 0.9709 & 0.9711 & 102{,}304 & 4.9\\
8 & 256 & 0.9883 & 0.9885 & 198{,}846 & 9.5\\
8 & 512 & 0.9955 & 0.9955 & 385{,}355 & 18.5\\
16 & 64 & 0.9726 & 0.9725 & 105{,}492 & 5.1\\
16 & 128 & 0.9889 & 0.9889 & 206{,}311 & 9.9\\
16 & 256 & 0.9960 & 0.9960 & 401{,}879 & 19.3\\
16 & 512 & 0.9986 & 0.9986 & 780{,}212 & 37.5\\
\bottomrule
\end{tabular}

\caption{Deep-1B ($10^{9}\times96$, L2), $\alpha{=}2\%$, $k_b{=}16$, the
track's $10^{4}$ queries against published ground truth. Coverage
\emph{is} the recall scan-only serving would deliver
(Proposition~\ref{obs:identity}); MB/query is the SQ4 scan the posting
budget implies. Lists average $800$ objects with a max/mean of
$\num{deep1b/list-skew}$.}
\label{tab:deep1b}
\end{table}

\paragraph{The certificate.} Three readings. First, the surface is
comfortable: recall $\num{deep1b/cov-b16-ks256}$ costs $19.3$\,MB of scan
per query --- four hundredths of one percent of the database --- and
$\num{deep1b/best-cov}$ costs $\num{deep1b/best-mb}$\,MB. Second,
\emph{the router is free at this scale}: coverage under the
$\num{deep1b/ns}$-point GRAFT router matches exact routing over the whole
vocabulary to within $\num{deep1b/router-gap}$ at every point, so the
approximation we make for speed costs nothing certifiable here (contrast
MS MARCO, \S\ref{sec:msmarco}, where the same substitution costs
$\num{cert/msmarco100m-router-cost}$). Third, the flat router at this $n$
is $\num{deep1b/router-flat-mib}$\,MiB of resident state --- above any
sane cap --- so a deployed billion-scale index is necessarily the depth-2
design of \S\ref{sec:ceiling}, not the flat one. The rest of this
subsection measures both.

\paragraph{Served, flat router (depth 1).} The $b{=}8$ SQ4 stream is
$\num{deep1b/s2/stream-gb}$\,GB and lives on a rotational hardware RAID
--- the slowest storage in this paper, chosen because it was the only
one large enough --- with the $\num{deep1b/s2/router-gb}$\,GB router
mapped, under a 16\,GB cgroup cap, 64 threads. At $(b,\ks,R)=(8,512,2000)$
the served recall@10 is $\num{deep1b/s2/recall}$ on
$\num{deep1b/s2/nq}$ queries; the certificate's account of the same
queries is coverage $\num{deep1b/s2/cov-same}$ less a screen loss of
$\num{deep1b/s2/screen-loss}$, and the residual is
$\num{deep1b/s2/identity-err}$ --- Proposition~\ref{obs:identity} holds
at $10^{9}$ as it did at $10^{6}$--$10^{8}$. Throughput is
$\num{deep1b/s2/qps-pread}$\,qps, $\num{deep1b/s2/io-mb}$\,MB of counted
IO per query, and where it comes from matters more than the number: the
same artifact served through page faults on the mapped stream reads
$\num{deep1b/s2/qps-mmap}$\,qps ($\num{deep1b/s2/pread-gain}\times$
slower; scan $\num{deep1b/s2/scan-ms-mmap}$ against
$\num{deep1b/s2/scan-ms-pread}$\,ms/query), because the kernel's
readahead and fault path, not the array, were the bottleneck. Serving is
explicit \texttt{pread} of the $\ks$ list segments. A three-stage
cascade with an SQ8 re-rank on NVMe between the 4-bit screen and the
exact refine is lossless and worth $\num{deep1b/s2/cascade-gain}\times$
--- once readahead is fixed the refine is cheap enough that a
$96$\,GB side-index buys nothing.

\begin{table}[t]
\centering\small
\resizebox{\textwidth}{!}{
\begin{tabular}{clrrrrrrrr}
\toprule
depth & cap & thr. & queries & $\ks$ & recall@10 & qps & route & scan & resident B/obj.\\
& & & & & & & \multicolumn{2}{c}{(ms/query)} & \\
\midrule
1 & 16\,GB & 64 & 500 & 512 & 0.9912 & 25.44 & 12.4 & 26.9 & 17.18\\
1 & 12\,GB & 64 & 500 & 512 & 0.9912 & 1.84 & 516.6 & 27.2 & 12.88\\
1 & 8\,GB & 64 & 500 & 512 & 0.9912 & 0.40 & 2442.2 & 38.0 & 8.59\\
\midrule
2 & 16\,GB & 64 & 500 & 512 & 0.9904 & 30.15 & 9.1 & 24.1 & 17.18\\
2 & 8\,GB & 64 & 500 & 512 & 0.9904 & 25.18 & 9.2 & 30.5 & 8.59\\
2 & 4\,GB & 64 & 500 & 512 & 0.9904 & 19.59 & 11.4 & 39.6 & 4.29\\
2 & 2\,GB & 64 & 500 & 256 & 0.9832 & 20.52 & 12.8 & 35.9 & 2.15\\
2 & 1\,GB & 16 & 200 & 512 & 0.9925 & 11.92 & 28.2 & 55.7 & 1.07\\
2 & 1\,GB & 8 & 200 & 512 & 0.9925 & 7.14 & 46.5 & 93.5 & 1.07\\
2 & \textbf{512\,MB} & 8 & 200 & 512 & \textbf{0.9925} & \textbf{6.93} & 47.7 & 96.6 & \textbf{0.54}\\
\bottomrule
\end{tabular}
}
\caption{Deep-1B served: depth 1 (flat $\num{deep1b/ns}$-point router,
$\num{deep1b/router-flat-mib}$\,MiB) against depth 2 (the router
indexed by \solo{}: a $2\%$ sample of it, $\num{deep1b/router-d2-mib}$\,MiB)
under falling cgroup caps, $b{=}8$, $R{=}2000$, explicit \texttt{pread}
of the $\ks$ list segments on both levels. Resident B/object is the cap
divided by $n$ (the cgroup peak equals the cap in every row). Rows at
different thread or query counts are reported but never compared: the
200-query subset reads $\num{deep1b/floor-recall}$ where the 500-query
set reads $\num{deep1b/recall-d2}$.}
\label{tab:deep1b-served}
\end{table}

\paragraph{Served, depth 2, and the collapse it removes.} The router is
a database of $\num{deep1b/ns}$ points, so we indexed it with the same
rule: a $2\%$ sample of it ($400$k points) is the resident router, and
the vocabulary's own posting lists are an SQ4 stream on NVMe served by
the same kernel as the object level. The resident router falls from
$\num{deep1b/router-flat-mib}$ to $\num{deep1b/router-d2-mib}$\,MiB
($\num{deep1b/router-ratio}\times$, the $1/\alpha$ the law predicts),
and the second hop costs $\num{deep1b/router-io-mb}$\,MB of counted IO.
Depth 2 is not a throughput sacrifice: at 16\,GB it serves
$\num{deep1b/qps-d2-16g}$ against depth 1's $\num{deep1b/qps-d1-16g}$\,qps
($\num{deep1b/d2-gain-16g}\times$), because a two-hop route whose second
hop is a counted read beats walking a mapped graph of ten gigabytes
(routing $\num{deep1b/route-ms-d2-16g}$ vs.\
$\num{deep1b/route-ms-d1-16g}$\,ms/query). And depth 1 does not degrade
below its footprint, it collapses: $\num{deep1b/qps-d1-16g}\to
\num{deep1b/qps-d1-12g}\to\num{deep1b/qps-d1-8g}$\,qps at $16/12/8$\,GB,
all of it in routing, recall unchanged --- the beam traversal paging
its graph, the cliff of \S\ref{sec:capped-baselines} reproduced inside
\solo{}. Depth 2 at 8\,GB serves $\num{deep1b/qps-d2-8g}$\,qps,
$\num{deep1b/d2-gain-8g}\times$ depth 1 at the same cap, and at
4\,GB $\num{deep1b/qps-d2-4g}$. (Thrashing rows carry a $\sim2\times$
run-to-run spread; the collapse is one to two orders of magnitude, not
a value.)

\paragraph{The floor.} Below 4\,GB the cap is no longer a cap on the
router but on the per-thread scratch, and two resident terms surface
there; the law of \S\ref{sec:ceiling} accounts for both. The scan kernels' dedupe scratch was a
bitmap of $n$ bits \emph{per thread} ---
$\num{deep1b/bitmap-mib-thread}$\,MiB at $10^{9}$,
$\num{deep1b/bitmap-gib-64}$\,GiB across 64 threads, against a
$\num{deep1b/router-d2-mib}$\,MiB router --- and is now an
open-addressed set sized by the candidates a query can touch. The
rank-starts array of the layout is $\Theta(\alpha n)$ and does not
recurse ($\num{deep1b/starts-mib-before}$\,MiB loaded); serving at a
fixed $b$ reads two of its $b{+}1$ columns and lists are contiguous, so
one CSR offsets array ($\num{deep1b/starts-mib-after}$\,MiB,
$\num{deep1b/starts-ratio}\times$) carries both, mapped rather than held.
With both, $10^{9}$ objects serve at recall@10
$\num{deep1b/floor-recall}$ from a
$\textbf{\num{deep1b/floor-cap-mb}\,MB}$ cap
($\num{deep1b/floor-threads}$ threads, $\num{deep1b/floor-qps}$\,qps,
$\num{deep1b/floor-b-per-object}$ resident bytes per object). Against
the same cap and thread count at $10^{8}$ (Table~\ref{tab:capped}:
$\num{deep1b/1e8-512m-b-per-object}$\,B/object at recall
$\num{deep1b/1e8-512m-recall}$), a decade more data costs
$\num{deep1b/floor-vs-1e8}\times$ \emph{fewer} resident bytes each; against
the flat router's floor at $10^{8}$
($\num{deep1b/1e8-flat-floor-b-per-object-lo}$--$\num{deep1b/1e8-flat-floor-b-per-object-hi}$\,B/object),
thirty to forty times fewer.
That inversion is the recursion; what it does and does not say about
$n$ is the subject of \S\ref{sec:ceiling}.

\paragraph{Depth 3, and the 8-bit router.} Two more things were built on
the same deployment, each closing a step the law of \S\ref{sec:ceiling}
would otherwise leave as arithmetic (Table~\ref{tab:deep1b-stage4}). \emph{Depth 3}: a
$2\%$ sample of the $400$k-point depth-2 router --- $8{,}000$ points,
$\num{s4/d3-router-mib}$\,MiB resident, $\num{s4/d3-router-ratio}\times$
below depth 2 --- with the depth-2 router's own posting lists as a third
SQ4 stream on NVMe. A query now routes in three hops, the third a counted
$\num{s4/d3-extra-io-mb}$\,MB. At 16\,GB it costs $\num{s4/d3-route-ms}$
against $\num{s4/d2-route-ms}$\,ms of routing and returns the same
coverage and recall to four decimals ($\num{s4/d3-cov-diff-16g}$,
$\num{s4/d3-recall-diff-16g}$), at $\num{s4/d3-qps-16g}$ vs.\
$\num{s4/d2-qps-16g}$\,qps against the same-day depth-2 control. Down to
$256$\,MB the two depths are indistinguishable, because the depth-2
router is file-backed and a $189$\,MiB image still mostly fits; at
$128$\,MB it does not --- depth 2 is killed at 4 threads and, at 2
threads, pages its beam at $\num{s4/d2-floor-qps}$\,qps --- while depth
3 serves $10^{9}$ at recall $\num{s4/d3-floor-recall}$ from
$\num{s4/d3-floor-cap-mb}$\,MB ($\num{s4/d3-floor-threads}$ threads,
$\num{s4/d3-floor-qps}$\,qps, $\num{s4/d3-floor-b-per-object}$ resident
bytes per object). That is the recursion doing at $10^{9}$ what the law
says a level does: the router term divided by $1/\alpha$ again, the
floor handed over entirely to the threads. \emph{The 8-bit router}: the
depth-2 router rebuilt over 8-bit-quantized copies of its points
(\S\ref{sec:ceiling}) routes to coverage and recall identical to the
float router's at every cap tried, in the same time. Stage 4 stored that
router in float32, so it showed only that the substitution costs
nothing; stage 5 stores it as 8-bit rows.

\paragraph{The 8-bit router, resident.} The same depth-2 router image,
its graph byte for byte, with its rows re-encoded as 8-bit integers on
one global affine grid and searched by an integer beam with the query
quantized on the same grid (Table~\ref{tab:deep1b-stage5}). The resident
image falls from $\num{s5/float-router-mib}$ to
$\num{s5/u8-router-mib}$\,MiB ($\num{s5/router-ratio}\times$; what
remains is mostly graph state, the $g$ term of \S\ref{sec:ceiling}), and
at every paired cap --- 16\,GB on 64 threads, 512\,MB on 8, 256\,MB on 4
--- coverage and recall are identical to the float control to four
decimals ($\num{s5/16g-t64-cov-diff}$, $\num{s5/512m-t8-recall-diff}$,
$\num{s5/256m-t4-recall-diff}$) and routing is not slower
($\num{s5/16g-t64-u8-route-ms}$ vs.\ $\num{s5/16g-t64-route-ms-float}$\,ms
at 64 threads). Where the float router pages, the 8-bit one fits: at
128\,MB on 2 threads the float depth-2 route paged its beam at
$\num{s4/d2-floor-qps}$\,qps, the 8-bit one serves
$\num{s5/128m-t2-u8-qps}$\,qps with routing at
$\num{s5/128m-t2-u8-route-ms}$\,ms, and at $\num{s5/u8-floor-cap-mb}$\,MB
it serves $10^{9}$ at recall $\num{s5/u8-floor-recall}$ at
$\num{s5/u8-floor-qps}$\,qps --- the floor depth 3 reached with a third
level (at $\num{s4/d3-floor-qps}$\,qps), reached at depth 2 by storing
the router in a quarter of the bytes. 64\,MB is killed at either depth,
and 128\,MB at 4 threads is the margin where either router may be
killed or page, as the two rows there say.

\begin{table}[t]
\centering\small
\resizebox{\textwidth}{!}{
\begin{tabular}{llrrccrrrr}
\toprule
cap & router rows & thr. & queries & coverage & recall@10 & qps & route ms & router MiB & B/obj.\\
\midrule
16\,GB & float & 64 & 500 & 0.9936 & 0.9904 & 25.84 & 9.1 & 188.7 & 17.18\\
16\,GB & u8 & 64 & 500 & 0.9936 & 0.9904 & 30.29 & 8.7 & 78.9 & 17.18\\
\midrule
512\,MB & float & 8 & 200 & 0.9925 & 0.9925 & 6.98 & 48.1 & 188.7 & 0.54\\
512\,MB & u8 & 8 & 200 & 0.9925 & 0.9925 & 6.89 & 47.3 & 78.9 & 0.54\\
\midrule
256\,MB & float & 4 & 200 & 0.9925 & 0.9925 & 3.65 & 90.7 & 188.7 & 0.27\\
256\,MB & u8 & 4 & 200 & 0.9925 & 0.9925 & 3.70 & 87.3 & 78.9 & 0.27\\
\midrule
128\,MB & float & 4 & 200 & 0.9925 & 0.9925 & 0.81 & 1059.0 & 188.7 & 0.13\\
128\,MB & u8 & 4 & --- & \multicolumn{4}{l}{\emph{OOM-killed or timed out}} & --- & 0.13\\
\midrule
128\,MB & u8 & 2 & 200 & 0.9925 & 0.9925 & 1.91 & 169.2 & 78.9 & 0.13\\
\midrule
96\,MB & u8 & 2 & 200 & 0.9925 & 0.9925 & 1.90 & 168.0 & 78.9 & 0.10\\
\midrule
64\,MB & u8 & 2 & --- & \multicolumn{4}{l}{\emph{OOM-killed or timed out}} & --- & 0.07\\
\bottomrule
\end{tabular}
}
\caption{Deep-1B, stage 5: the depth-2 router served from a float image
(188.7\,MiB) and from the same image with its rows as 8-bit integers on
one global grid (78.9\,MiB; graph sections byte-identical), each pair at
the same cap, thread count and query count on the same day. $b{=}8$,
$\ks{=}512$, $\ks'{=}256$, $R{=}2000$, explicit \texttt{pread} at both
levels. Coverage is computed on the served queries from the stored
signatures. Rows at different thread or query counts are never
compared.}
\label{tab:deep1b-stage5}
\end{table}

\begin{table}[t]
\centering\small
\resizebox{\textwidth}{!}{
\begin{tabular}{llrrccrrrr}
\toprule
cap & route & thr. & queries & coverage & recall@10 & qps & route ms & router MiB & B/obj.\\
\midrule
16\,GB & depth 2 (control) & 64 & 500 & 0.9936 & 0.9904 & 25.60 & 9.1 & 188.7 & 17.18\\
16\,GB & depth 2, 8-bit router & 64 & 500 & 0.9936 & 0.9904 & 30.00 & 9.1 & 188.7 & 17.18\\
16\,GB & depth 3 & 64 & 500 & 0.9936 & 0.9904 & 29.14 & 10.0 & 3.6 & 17.18\\
\midrule
512\,MB & depth 2 (control) & 8 & 200 & 0.9925 & 0.9925 & 6.88 & 48.7 & 188.7 & 0.54\\
512\,MB & depth 2, 8-bit router & 8 & 200 & 0.9925 & 0.9925 & 6.99 & 47.9 & 188.7 & 0.54\\
512\,MB & depth 3 & 8 & 200 & 0.9925 & 0.9925 & 6.73 & 50.5 & 3.6 & 0.54\\
\midrule
256\,MB & depth 2 (control) & 4 & 200 & 0.9925 & 0.9925 & 3.66 & 91.5 & 188.7 & 0.27\\
256\,MB & depth 3 & 4 & 200 & 0.9925 & 0.9925 & 3.19 & 128.0 & 3.6 & 0.27\\
\midrule
128\,MB & depth 2 (control) & 4 & 200 & \multicolumn{4}{l}{\emph{OOM-killed}} & 188.7 & 0.13\\
128\,MB & depth 3 & 4 & 200 & 0.9925 & 0.9925 & 3.12 & 136.5 & 3.6 & 0.13\\
\midrule
128\,MB & depth 2 (control) & 2 & 200 & 0.9925 & 0.9925 & 0.49 & 1645.7 & 188.7 & 0.13\\
128\,MB & depth 3 & 2 & 200 & 0.9925 & 0.9925 & 1.59 & 270.3 & 3.6 & 0.13\\
\midrule
96\,MB & depth 3 & 2 & 200 & 0.9925 & 0.9925 & 1.64 & 269.8 & 3.6 & 0.10\\
\bottomrule
\end{tabular}
}
\caption{Deep-1B, stage 4: the depth-2 router rebuilt over 8-bit
copies of its points, and depth 3 (an $8{,}000$-point router, the
depth-2 router's lists as a third SQ4 stream), each beside a same-day
depth-2 float control at the same cap, thread count and query count.
$b{=}8$, $\ks{=}512$, $\ks'{=}256$, $\ks''{=}128$, $R{=}2000$, explicit
\texttt{pread} at every level. Coverage is computed on the served
queries from the stored signatures. Rows at different thread or query
counts are never compared.}
\label{tab:deep1b-stage4}
\end{table}

\subsection{Build and asset ledger}\label{sec:build}
From scratch: one router build plus $n$ independent searches ---
SIFT-1M in 5\,s (laptop), Deep-10M in ${\sim}5$\,min (d0, 64
threads); assignment for every $b$ is free thereafter. Assets at
Deep-100M ($b{=}16$): postings 6.4\,GB; SQ8 codes 9.6\,GB (158\,s);
optional disk leaf blocks 160\,GB (830\,s, sequential); router image
1.0\,GB. No codebooks, no graphs, no training anywhere in the
pipeline.

\subsection{Ablation: recursion vs.\ per-list graphs}\label{sec:ablation}
Before adopting the one-rule recursion we built the natural
alternative: attach a small graph index to every oversized list and
answer by per-list graph searches. The recursion wins at every matched
GloVe operating point on \emph{both} axes: $\num{abl/rec-recall-ks64}$
at $\num{abl/rec-qps-ks64}$ qps against the hybrid's
$\num{abl/hyb-recall-ks64}$ at $\num{abl/hyb-qps-ks64}$ ($\ks{=}64$),
and $\num{abl/rec-recall-ks128}$ at $\num{abl/rec-qps-ks128}$ against
$\num{abl/hyb-recall-ks128}$ at $\num{abl/hyb-qps-ks128}$
($\ks{=}128$) --- while replacing 58 minutes of graph construction
with 2 minutes of GEMM assignments and 4.2\,GB of graphs with
postings. One caveat on the throughput half of that claim: the
$\ks{=}64$ rows of the two designs were measured in different sessions
on the same laptop, so their ratio carries the machine's thermal drift
as well as the design difference; the $\ks{=}128$ pair was measured
interleaved, and the recall comparison --- stable to $10^{-4}$ across
re-runs --- is unaffected at both settings. Uniformity costs nothing;
the auxiliary structure retired.

\section{The resident-memory law}\label{sec:ceiling}
Three terms are resident when \solo{} serves from storage: a router over
$\alpha^{L}n$ points, per-thread scratch, and the page cache of the list
reads. The rule applied to the router divides the first by $1/\alpha$ per
level, and the floors of \S\ref{sec:capped} and \S\ref{sec:deep1b} are
that division carried out through depth 3 at $10^{9}$. This section
states the law, the ceiling it implies at each depth, what it says and
does not say about $n$, and what the router term costs under
quantization.

\paragraph{Three terms.}
With sample fraction $\alpha$, a per-sample-point cost $c_r$
($c_r=4d+g$ for a float router, $c_r=d+g$ with 8-bit router vectors;
$g\approx64$\,B of graph state, measured), $T$ serving threads and an
operating point $\mathit{op}=(b,\ks,R)$,
\[
\begin{aligned}
\mathrm{resident}(n,L,T,\mathit{op})\;=\;&
\underbrace{\alpha^{L}\,n\,c_r}_{\text{router}}
\;+\;\underbrace{T\cdot s(\mathit{op})}_{\text{scratch}}\\
&\;+\;O(\text{page cache}),
\end{aligned}
\]
the last term the page cache of the $\ks$ list reads, where $L$ is the
recursion depth of the router (\S\ref{sec:recursion},
pointed upward; $L=1$ is the flat router) and $s(\mathit{op})$ is the
per-thread dedupe scratch, an open-addressed set of $8$\,B slots whose
count is the next power of two above twice the candidates a query can
touch. Everything else the layout needs --- the rank-starts offsets
($\Theta(\alpha n)$, and they do \emph{not} recurse), the code stream,
the base vectors --- is mapped or read explicitly and reclaimable page by
page. The scratch and page-cache terms were identified at $10^{9}$
(\S\ref{sec:deep1b}), where each exceeded the router it was compared
against: a bitmap of $n$ bits per thread was
$\num{deep1b/bitmap-gib-64}$\,GiB against a
$\num{deep1b/router-d2-mib}$\,MiB router, and the starts array
$\num{deep1b/starts-mib-before}$\,MiB before it was compacted.

\paragraph{The flat ceiling, and why it binds.}
At depth 1, $\alpha=2\%$, $d=96$ the router term is
$\num{ceiling/bytes-per-object}$\,B/object, so the flat ceiling is
$n^{\ast}(1)=M/(\alpha c_r)\approx\num{ceiling/nstar-1gb}\times10^{8}$
at $M=1$\,GB. The flat router's measured floor at $10^8$
(Table~\ref{tab:capped-floor}),
$\num{deep1b/1e8-flat-floor-b-per-object-lo}$--$\num{deep1b/1e8-flat-floor-b-per-object-hi}$\,B/object
once scratch and the reads in flight are counted, sits above it: the
flat-router design does not reach 1\,GB at $10^8$, and the Deep-100M
deployment at that cap is the depth-2 one (Table~\ref{tab:capped}). At
$10^{9}$ the same arithmetic gives
$\num{ceiling/1e9-flat-analytic-gib}$\,GiB and the built router is
$\num{deep1b/router-flat-mib}$\,MiB, the difference being graph state
above the $64$-byte estimate. The ceiling cannot be dodged by shrinking
$\alpha$ --- a sublinear vocabulary is the constant-pivot regime whose
resolution does not scale, and the linear vocabulary is the recall
engine --- nor by paging the router: beam traversal is random access,
the one pattern that thrashes under a cap, and
Table~\ref{tab:deep1b-served} shows what that looks like,
$\num{deep1b/qps-d1-16g}\to\num{deep1b/qps-d1-12g}\to\num{deep1b/qps-d1-8g}$\,qps
as the cap falls from 16 to 8\,GB under a $\num{deep1b/s2/router-gb}$\,GB
router, with recall untouched.

\paragraph{The rule pointed upward.}
The router is a database of $\alpha n$ points, so \solo{} indexes it:
the resident state becomes the sample of the sample, and the router's
posting lists move into a code stream on storage like everything else.
The router term becomes $\alpha^{L}n\,c_r$ and the ceiling
\[
n^{\ast}(L)=\frac{M-T\,s(\mathit{op})}{\alpha^{L}c_r}
      \;\approx\;n^{\ast}(1)\cdot(1/\alpha)^{\,L-1}
\]
when the scratch is small against $M$: each level multiplies the ceiling
by $1/\alpha=50$, and each level adds one routing stage that is a
counted read of ${\sim}\ks\cdot(b_r/\alpha)$ codes
($\num{deep1b/router-io-mb}$\,MB per query at $10^{9}$). The step is
measured where it starts. At $10^{8}$ the $2$M-point vocabulary (a
$\num{upward/resident-flat-mib}$\,MiB flat router) was itself indexed by
\solo{} --- a $2\%$ sample of $40$k points, $\num{upward/resident-mib}$\,MiB
resident, $\num{upward/ratio}\times$ less --- with the vocabulary posted
to it at $b_r{=}8$ and queries routed in two hops, the router level
scanned by the same fused kernel as the object level. Routing loses what
one expects and the object level absorbs it: at $\ks{=}64$ the two-hop
route returns $\num{upward/agree-d2-ks64}$ of the exact top-$\ks$
vocabulary terms against the flat router's $\num{upward/agree-d1-ks64}$,
yet end-to-end recall falls only from $\num{upward/recall-d1-ks64}$ to
$\num{upward/recall-d2-ks64}$, at most $\num{upward/max-recall-loss}$
across the grid, because a covered neighbor sits in a median of several
touched lists (\S\ref{sec:insurance}). Throughput costs
$\num{upward/thru-ks256}\times$ at $\ks{=}256$
($\num{upward/qps-d1-ks256}\to\num{upward/qps-d2-ks256}$ qps) where the
flat router fits in RAM, and the certificate composes: coverage under
the depth-2 route predicts the served recall to
$\num{upward/compose-err}$, the same residual as at depth 1. At $10^{9}$,
where the flat router does not fit, the sign of the throughput cost
flips --- depth 2 is $\num{deep1b/d2-gain-16g}\times$ faster than depth 1
at 16\,GB and $\num{deep1b/d2-gain-8g}\times$ at 8\,GB
(Table~\ref{tab:deep1b-served}) --- and the measured floor is
$\num{deep1b/floor-cap-mb}$\,MB at recall $\num{deep1b/floor-recall}$,
$\num{deep1b/floor-b-per-object}$\,B/object, the whole of it the
$\num{deep1b/router-d2-mib}$\,MiB depth-2 router plus eight threads of
scratch; depth 3 serves from $\num{s4/d3-floor-cap-mb}$\,MB
(Table~\ref{tab:deep1b-stage4}). At $M=1$\,GB and $d=96$ the arithmetic
places $n^{\ast}$ near $5.6\times10^{9}$ at depth 2 and
$3\times10^{11}$--$10^{12}$ at depth 3 (float or 8-bit top router).

\paragraph{What the law says about $n$, and what it does not.}
The router term is linear in $n$: recursion divides its constant by $50$
per level and does not remove it, so the $\num{deep1b/router-d2-mib}$\,MiB
depth-2 router at $10^{9}$ is ${\sim}\num{ceiling/1e10-d2-router-gib}$\,GiB
at $10^{10}$ at the same depth, and ${\sim}\num{ceiling/1e10-d3-router-mib}$\,MiB
at depth 3. The scratch is independent of $n$ only at a fixed operating
point, and holding recall fixed does not fix the operating point: the
candidates a query must touch at a recall target grow with $n$ as the
exponents of \S\ref{sec:laws} say, $n^{0.4\text{--}0.6}$ over the
measured range and steepest in the last decade, so $T\cdot s(\mathit{op})$
grows the same way --- a few MiB per thread at $10^{9}$, and not a
constant. ``Deliverable at any database size'' therefore means that for
every $n$ there is a depth $L$ at which the resident state fits a fixed
budget, each such $L$ a build; it is measured through depth 3 at
$10^{9}$, where depth 3 serves from $\num{s4/d3-floor-cap-mb}$\,MB and
depth 2 no longer fits $128$\,MB. What would block $10^{10}$ in this
implementation is not memory: the object identifiers are 32-bit in both
kernels and in the layout builder, the layout build sorts $n\cdot b$
identifiers in RAM, and the stream and base together are
${\sim}7.6$\,TiB. Those are engineering, stated here so that the ceiling
is not read as a claim past the last measured point.

\paragraph{The router term under quantization.}
The law admits an 8-bit router, and what that costs is a question about
the certificate: quantizing the router changes which vocabulary points a
query routes to, and coverage is recall (Prop.~\ref{obs:identity}). We
measured it on four corpora spanning dimension and intrinsic dimension,
ranking the vocabulary under each code, \emph{paired} over the same
$\num{rq/deep-nq}$ queries under both routers --- which is what makes it
decidable, since the paired difference is dominated by the many queries
whose routed pool is identical under the two codes. An 8-bit router is
free: at $\ks{=}32$ the paired cost is $\num{rq/sq8-worst-diff}$ on every
corpus, with a bootstrap interval no wider than $\pm\num{rq/sq8-worst-ci}$,
and the number of queries whose routed pool changes at all is
$\num{rq/wiki-sq8-ndiff}$ of $\num{rq/wiki-nq}$ on wiki-1024,
$\num{rq/deep-sq8-ndiff}$ on Deep-96, $\num{rq/sift-sq8-ndiff}$ on
SIFT-128 and $\num{rq/glove-sq8-ndiff}$ on GloVe. Four bits is free on
three of the four ($\num{rq/deep-sq4-diff}$, $\num{rq/sift-sq4-diff}$,
$\num{rq/wiki-sq4-diff}$) and costs $\num{rq/glove-sq4-diff}$ on GloVe.
Two bits splits by concentration rather than by dimension: at
essentially equal dimension the penalty is $\num{rq/deep-evp-diff}$ on
Deep-96 against $\num{rq/glove-evp-diff}$ on GloVe-100, and only
$\num{rq/wiki-evp-diff}$ at $d{=}1024$, where the ternary code keeps
$512$ coordinates --- the same regime split the equal-work law draws
(\S\ref{sec:laws}). Quantizing the query as well costs a further
$\num{rq/glove-evptt-diff}$ on GloVe and $\num{rq/deep-evptt-diff}$ on
Deep, for one vector, and is never worth it.

\begin{table}[t]\centering\small
\caption{The router term under quantization: $c_r$ and the depth-1 ceiling
$n^{\ast}$ at $M=1$\,GB, from $c_r=(\text{bytes per router vector})+g$ with
$g\approx64$\,B and $\alpha=2\%$ (\S\ref{sec:ceiling}); coverage cost
measured \emph{paired} over $\num{rq/deep-nq}$ queries per corpus.}
\label{tab:router-q}
\resizebox{\textwidth}{!}{
\begin{tabular}{@{}lrrrrl@{}}
\toprule
& \multicolumn{2}{c}{$d=96$} & \multicolumn{2}{c}{$d=1024$} & \\
\cmidrule(lr){2-3}\cmidrule(lr){4-5}
router codes & $c_r$ & $n^{\ast}$ & $c_r$ & $n^{\ast}$ & coverage cost at $\ks{=}32$\\
\midrule
float32 & 448 & $1.2{\times}10^{8}$ & 4{,}160 & $1.3{\times}10^{7}$ & reference\\
8-bit & 160 & $3.4{\times}10^{8}$ & 1{,}088 & $4.9{\times}10^{7}$ & \textbf{none}: 0.0000 everywhere, $\le$ 49/2{,}000 queries affected\\
4-bit & 112 & $4.8{\times}10^{8}$ & 576 & $9.3{\times}10^{7}$ & none, except -0.0072 on GloVe\\
2-bit ternary & 88 & $6.1{\times}10^{8}$ & 320 & $1.7{\times}10^{8}$ & -0.0018 ($d{=}1024$), -0.0121 (Deep-96), -0.0731 (GloVe)\\
\bottomrule
\end{tabular}
}
\end{table}

Table~\ref{tab:router-q} prices the consequence: an 8-bit router
multiplies the depth-1 ceiling by $\num{rq/gain-sq8-d96}\times$ at
$d{=}96$ ($n^{\ast}$ from $\num{ceiling/nstar-1gb}\times10^{8}$ to
$\num{rq/nstar-8}\times10^{8}$ at $M{=}1$\,GB) and by
$\num{rq/gain-sq8-d1024}\times$ at $d{=}1024$, and by the same factor at
every depth; at embedding dimensions the 2-bit router is worth another
$3.4\times$ for about $20\%$ more scanned data (float to 2 bits is
$\num{rq/gain-evp-d1024}\times$ the ceiling at $d{=}1024$). These are
$10^6$-scale laptop measurements with $\num{rq/deep-nq}$ queries, one
$\alpha$ and one $b$; the substitution is confirmed end to end at
$10^{9}$ on the depth-2 router (Table~\ref{tab:deep1b-stage4}), rebuilt
over 8-bit copies of the same $400$k points, which routes to the same
coverage and recall to four decimals on the same queries
($\num{s4/e1-cov-diff-16g}$, $\num{s4/e1-recall-diff-16g}$) in the same
time ($\num{s4/e1-route-ms-sq8}$ vs.\ $\num{s4/e1-route-ms-float}$\,ms);
and stored as 8-bit rows (Table~\ref{tab:deep1b-stage5}) the resident
image falls from $\num{s5/float-router-mib}$ to
$\num{s5/u8-router-mib}$\,MiB and serves $10^{9}$ from
$\num{s5/u8-floor-cap-mb}$\,MB at depth 2, the floor depth 3 needed a
third level for (\S\ref{sec:deep1b}). The ratio is
$\num{s5/router-ratio}\times$ rather than the $4\times$ of the vector
bytes because at low $d$ the $64$-byte constant $g$ is the binding term
--- $40\%$ of an 8-bit $c_r$ at $d{=}96$, and more than half of the
8-bit image measured --- so below 8 bits the graph state, not the
vectors, is what must be attacked.

\section{Where \solo{} sits}\label{sec:where}
Table~\ref{tab:compare} reports throughput against HNSW because a reader
expects it, but a throughput contest is a contest over a region both
architectures already reach, decided by kernels and by the machine
(\S\ref{sec:setup}). What distinguishes \solo{} is a set of operating
points no measured system has, and the structure that produces them.
Four points, each with its evidence in this paper.

\paragraph{It serves where nothing else fits and answers.}
At Deep-100M \solo{} serves recall $\num{capped/1g-recall-d2}$ from
$\num{capped/peak-gb}$\,GB of enforced resident memory and from
$\num{capped/solo-floor512-mb}$\,MB on $\num{capped/solo-floor512-threads}$
threads (Table~\ref{tab:capped}). Six other systems were run through the
identical protocol (Table~\ref{tab:capped-baselines}): SPANN, DiskANN and
vote-served \misi{} are killed below $8$, $12$ and $6$\,GB, the
re-ranking indexes of ScaNN and RaBitQ below
$\num{capped/scann-bf16/floor-gb}$ and
$\num{capped/rabitq-rq4-refine/floor-gb}$, and the indexes that do run
small --- Faiss IVF-PQ, ScaNN scoring on its 4-bit codes, RaBitQ on its
1- and 4-bit codes --- saturate at $\num{capped/ivfpq-ceiling}$,
$\num{capped/scann-ah/ceiling}$, $\num{capped/rabitq-rq1-codes/ceiling}$
and $\num{capped/rabitq-rq4-codes/ceiling}$ recall. The region is the conjunction of
high recall and a tight budget, and it is the conjunction, not either
half, that no other system reaches. The reason is a mechanism. Paging a
random-access navigation structure amplifies reads ${\sim}85\times$ ---
letting vector reads fault through \texttt{mmap} pulls $6$\,TB in $80$
minutes under a $4$\,GB cap (\S\ref{sec:kernels}) --- because a beam goes
where the graph leads and the access order is not known until the query
is over; a scan's access order is known before the query starts, which
is what makes explicit, sequential, counted IO possible. The alternative
that fits under the cap, scoring with quantized codes instead of
verifying with the true distance, cannot reach the recall: SPANN's
native Int8 scorer saturates at $\num{sat/spann-int8}$ on the identical
grid that is lossless as a screen with exact refinement
(\S\ref{sec:sq}), and IVF-PQ, ScaNN and RaBitQ repeat the mechanism in
three systems with no graph in them. Under a cap the arms split by access
order: capping IVF-PQ from 8 to 1\,GB costs it $40\times$ in throughput
where it costs \solo{} $\num{capped/solo-1g-vs-2g-pct}\%$ from 2 to
1\,GB, because a probed list is a random read and a scan is not.

\paragraph{The load decomposes vertically, which a graph cannot do.}
Every non-leaf list in \solo{} is a complete \solo{} index behind a
two-argument interface: given a query and a local budget, return
candidates. Nothing crosses a level except a candidate set. A beam
search exposes no such interface, which is why distributed graph systems
replicate whole indexes rather than split them. The decomposition
behaves: recall composes multiplicatively across levels and the depth-2
index lands on its composed prediction to $2\times10^{-4}$
(\S\ref{sec:recursion}); coverage inside the decomposed lists is not
degraded but improved, up to $0.9994$; the recursion beat a graph index
attached to each oversized list on recall and throughput together with a
build $25\times$ cheaper (\S\ref{sec:ablation}); and the whole path is
bitwise deterministic at any thread count and across machines
(\S\ref{sec:dynamics}), so a sharded deployment can be verified against
a single-node reference, object for object. Served with its two levels
in separate processes --- a $\num{split/router-mib}$\,MiB routing tier
and $\num{split/bytes-per-query}$ bytes per query on the wire --- the
depth-2 GloVe index returns answers \num{split/identical} to the
single-process index at every operating point; the prototype's hop is
synchronous and both tiers share one machine, so its throughput cost
($12$--$33\%$) measures the prototype, not the design. Because per-level
coverage multiplies and each level is certifiable on its own, the recall
of a distributed deployment is predictable from per-level certificates
before it is assembled.

\paragraph{Replication pays here and does not pay for the graph.}
Four socket-local replicas lift \solo{} by
$\num{deep100m/repl-solo-gain}\times$ and hnswlib by
$\num{deep100m/repl-hnsw-gain}\times$ on the same machine and queries
(\S\ref{sec:compare}). A scan is bandwidth-bound and a walk is
latency-bound, so replication converts memory controllers into
throughput for one and buys much less for the other; the same asymmetry
holds across sockets, NUMA nodes, threads and machines.

\paragraph{The same rule pointed upward.}
The one resident term that grows with $n$ is the router, and each
recursion level divides it by $1/\alpha=50$ (\S\ref{sec:ceiling}). The
arithmetic places $n^{\ast}\approx\num{ceiling/nstar-1gb}\times10^{8}$ at
depth 1 and $M=1$\,GB --- where the measured Deep-100M deployment of
Table~\ref{tab:capped} sits --- $5.6\times10^{9}$ at depth 2 and
$10^{11}$--$10^{12}$ at depth 3, each level adding a few megabytes of
sequential IO per query. Depth 2 is measured at $10^{8}$
($\num{upward/ratio}\times$ less resident routing state for
$\num{upward/max-recall-loss}$ of recall, the certificate composing to
$\num{upward/compose-err}$) and at $10^{9}$, where it is faster than the
flat router and serves from $\num{deep1b/floor-cap-mb}$\,MB; depth 3 is
measured at $10^{9}$ (\S\ref{sec:deep1b}). Work per query grows
sublinearly on top: the scanned fraction of the database falls
throughout the four measured decades (\S\ref{sec:laws}), with the
steepest decade last, so the levels buy capacity without the scan
growing in proportion, and the paper claims nothing past the last
measured point.

\section{Related work}\label{sec:rel}
\emph{Cluster-routed lists.} IVF and its quantized descendants
\cite{jegou2011pq} route to $k$-means cells and scan them; SPANN
\cite{chen2021spann} scales this to storage with balanced hierarchical
clustering, closure multi-assignment ($\le8$ replicas), and an
in-memory centroid graph \cite{sptag2018}, with SPFresh
\cite{xu2023spfresh} adding update protocols.
\solo{} occupies the same serving silhouette with a \emph{random}
vocabulary (no clustering; deterministic; metric-black-box), rank-$b$
assignment (measured to dominate closure, \S\ref{sec:rule}),
serve-time replication, recursion instead of balancing, an exact
scan (its 8-bit screen re-ranks exactly, unlike lossy IVFPQ scoring),
and a certificate.
\emph{Permutation and pivot indexes.} Objects can be described by the
order in which a fixed pivot set is seen from them
\cite{chavez2008permutations}; MI-File \cite{amato2014mifile} indexes
that description in an inverted file and NAPP
\cite{tellez2013napp,naidan2015permutation} keeps only the nearest
pivots and intersects. All of them draw the pivots from a set of
\emph{constant} size: NAPP attains perfect recall at $10^6$
(Table~\ref{tab:compare}, text) but its resolution cannot grow with
$n$, which is the property $10^8$ demands. \misi{} \cite{chavez2026misi} made
the vocabulary a linear-size sample and served by vote-then-verify;
\solo{} keeps the representation and replaces the serving theory:
scan-only, certificate, recursion, quantized leaves.
\emph{Navigable graphs.} HNSW \cite{malkov2020hnsw}, DiskANN
\cite{subramanya2019diskann}, and \graft{} \cite{graft2026} are the throughput
frontier in RAM below their recall ceilings; their recall is emergent,
their updates hard, and their disk incarnations pin navigation state.
The comparison of \S\ref{sec:compare} is deliberately against the
strongest of each.

\section{Limitations and future work}\label{sec:future}

\subsection{Filtered search: what the identity gives, and what it
  still owes}\label{sec:filter}
\emph{This subsection is analysis, not measurement.} Predicate-constrained
search --- the top-$k$ among objects satisfying a filter --- is where
navigable graphs are most awkward: the filtered subgraph of a graph built
for the unfiltered set need not be navigable, and selective predicates
disconnect it, which is why filtered variants are separate systems
\cite{gollapudi2023filtereddiskann,patel2024acorn}. A scan-only inverted
file has no such failure mode, because it has no traversal to break.
Let $F\subseteq X$ be the surviving set of a predicate, given as a
bitmap. Serving intersects each touched list with $F$ and scans the
survivors; the returned answer is the exact top-$k$ of the scanned
survivors, so Proposition~\ref{obs:identity} holds verbatim with every
quantity conditioned on $F$:
\[
\mathrm{recall@}k\;\bigl|\;F \;=\; \Pr\bigl[\,T_b(y)\cap R_{\ks}(q)\neq\emptyset
\;\bigm|\; y\in F\,\bigr].
\]
Two consequences are worth stating because they are properties no graph
index has. First, the \emph{cost} of a filtered query is known before it
runs: the scanned mass is $\sum_{s\in R_{\ks}(q)}|L_s\cap F|$, a sum of
bitmap popcounts over the touched lists, so a planner can price the query
and, using the rank-ordered postings of \S\ref{sec:kernels}, extend $b$
at query time until the surviving mass reaches its target --- the dial the
unfiltered path already uses. Second, the \emph{certificate} conditions
the same way: coverage computed over the surviving objects certifies
filtered recall, without serving the filter, provided the ground-truth
sample is drawn under it.

What this does not yet settle is the level of that conditional coverage.
A predicate correlated with the geometry --- a category that occupies a
region of the space --- removes neighbors and vocabulary mass
unevenly, and at high selectivity the surviving mass per list falls
faster than $b\cdot\ks$ can be raised. The equal-work law predicts the
compensation required; whether it is affordable at selectivities of
$10^{-2}$ and below, and how the result compares with
filtered-DiskANN and ACORN on the same corpus, is an experiment this
paper does not contain, and it claims no outcome for it.

\subsection{Other limitations}
\begin{itemize}
\item \textbf{One machine for the $10^8$ throughput comparison.} Its
  memory system (four 2\,GHz Haswell sockets, ${\sim}50$\,GB/s to one
  process) bounds a streaming kernel more than a latency-bound graph
  walk; the socket-local protocol removes most of that bound
  (\S\ref{sec:compare}), and the same kernel runs ten times faster per
  core on the laptop, so the ratio between the two systems on current
  hardware remains to be measured.
\item \textbf{Query-batched scanning} (int8 GEMM/ADC over per-list tiles,
  reuse $n_q\ks/|S|$) is unimplemented. Disk throughput is
  bandwidth-bound at ${\sim}40$\,KB mean reads; batching and SQ4 project
  $50$--$100$ qps at 1\,GB.
\item \textbf{Deep-100M rows are single-seed.} The 10M seed variance
  (\S\ref{sec:setup}) bounds the expected fluctuation.
\item \textbf{The out-of-sample provision} (\S\ref{sec:oos}) is measured
  on \num{lidlaw/n-corpora} corpora, with one random split (LAION-10M),
  two in-domain shifts (wiki, PubMed23) and one cross-modal shift
  (T2I-10M). The inner search adds a term the calibration table absorbs
  by construction, and the three corpora with \graft{} signatures are
  the measurement of it. The corpus-specific offsets of
  Figure~\ref{fig:lidlaw} mean the table is per corpus; whether a table
  fitted on a sample of the cross-modal queries then transfers is
  untested.
\item \textbf{Beyond $10^{9}$} nothing is measured; the engineering
  blockers at $10^{10}$ are listed in \S\ref{sec:ceiling}.
\item \textbf{Per-object adaptive $b$} (rank-gap statistics) may dominate
  both fixed $b$ and closure, and is future work.
\end{itemize}

\section{Conclusion}
Making the serving path trivial --- route, scan, return --- turns recall
into a property of the stored index. One ground-truth pass certifies
every operating point before any is served, and the certificate held to
four decimals from $10^6$ to $10^9$ objects, on the corpus where the
index performs worst as on the ones where it performs best. The same
simplicity turns replication into insurance, quantization into lossless
acceleration, skew into a recursion, and the router into one more
database to index, which is what serves a hundred million vectors from
a quarter of a gigabyte and a billion from under a hundred megabytes.
Inserts are one search and deletes are exact, the build is sampling and
searching, and the comparisons are against the strongest systems we
could run on the same machine, queries and ground truth. At scale, the
right primitive may not be a better navigator, but a better bill of
what must be looked at.

\paragraph{Artifact and reproducibility.}
Every number in this paper was written by the run that produced it
into a versioned manifest recording the code commit, the machine, the
complete parameter set, the seed, and the raw timings; baselines were
run by us on the same machines, queries, and ground truth as \solo{},
under saturation sweeps rather than at a single operating point. The
implementation, the experiment drivers, and the manifests are at
\url{https://github.com/zevahcle/SOLO}.

\end{document}